\documentclass[a4paper,UKenglish,cleveref, autoref, thm-restate]{lipics-v2021}

\usepackage{stmaryrd}
\usepackage{centernot}
\usepackage{booktabs}
\usepackage{appendix}
\usepackage{hyperref}
\usepackage{amsmath}
\usepackage{amssymb}
\usepackage{amsthm}
\usepackage{thmtools}
\usepackage{mathtools}
\usepackage{subcaption}
\usepackage{tikz}
\usetikzlibrary{arrows,automata,shapes,decorations,calc,patterns,positioning}

\newif\iffull
\newcommand{\forfull}[1]{\iffull#1\fi}
\newcommand{\forcamera}[1]{\iffull\else#1\fi}
\newcommand{\altfull}[2]{\iffull#1\else #2\fi}
\fulltrue

\newcommand{\states}{\ensuremath{\mathcal{S}}}
\newcommand{\initstate}{\ensuremath{s_{\mathit{init}}}}
\newcommand{\actionset}{\ensuremath{\mathit{Act}}}
\newcommand{\transrel}{\ensuremath{\mathit{Trans}}}
\newcommand{\env}{\ensuremath{\mathit{e}}}
\newcommand{\varset}{\ensuremath{\mathit{Var}}}
\newcommand{\source}[1]{\ensuremath{\mathit{src}(\mathit{#1})}}
\newcommand{\target}[1]{\ensuremath{\mathit{trgt}(\mathit{#1})}}
\newcommand{\action}[1]{\ensuremath{\mathit{act}(\mathit{#1})}}
\newcommand{\co}{\ensuremath{\cdot}}
\newcommand{\tp}{\ensuremath{\mathit{tt}}}
\newcommand{\fp}{\ensuremath{\mathit{ff}}}
\newcommand{\diam}[1]{\ensuremath{\langle \mathit{#1} \rangle}}
\newcommand{\boxm}[1]{\ensuremath{[ \mathit{#1} ]}}
\newcommand{\lfp}[2]{\ensuremath{\mu \mathit{#1}.\mathit{#2}}}

\newcommand{\lfpb}[2]{\ensuremath{\mu \mathit{#1}.(\mathit{#2})}}

\newcommand{\update}[3]{\ensuremath{#1 [ #2 := #3]}}
\newcommand{\imps}{\ensuremath{\Rightarrow}}
\newcommand{\emptyseq}{\ensuremath{\varepsilon}}
\newcommand{\clos}[1]{\ensuremath{\mathit{#1}^\star}}
\newcommand{\actname}[1]{\ensuremath{\mathit{#1}}}
\newcommand{\allact}{\ensuremath{\actionset}}
\newcommand{\noact}{\ensuremath{\emptyset}}
\newcommand{\comp}[1]{\ensuremath{\overline{\mathit{#1}}}}
\newcommand{\block}{\ensuremath{\mathcal{B}}}
\newcommand{\nonblock}{\ensuremath{\comp{\block}}}
\newcommand{\concsym}{\smile^{\hspace{-.5ex}\raisebox{-.2ex}{\tiny$\bullet$}}}
\newcommand{\nconcsym}{{\centernot\smile}^{\hspace{-.5ex}\raisebox{-.4ex}{\tiny$\bullet$}}}
\newcommand{\conc}{\ensuremath{\mathbin{\concsym}}}
\newcommand{\nconc}{\ensuremath{\mathbin{\nconcsym}}}
\newcommand{\pattern}[1]{\textit{#1}}
\newcommand{\dpre}{\ensuremath{\rho}}
\newcommand{\den}{\ensuremath{\alpha_{\mathit{e}}}}
\newcommand{\ddis}{\ensuremath{\alpha_{\mathit{f}}}}
\newcommand{\dpreb}{\ensuremath{\dpre_{\mathit{b}}}}
\newcommand{\dpres}{\ensuremath{\dpre_{\mathit{s}}}}
\newcommand{\formonsym}{\ensuremath{\phi_{\mathit{on}}}}
\newcommand{\formoffsym}{\ensuremath{\phi_{\mathit{of}}}}
\newcommand{\actelimsym}{\ensuremath{\alpha_{\mathit{el}}}}
\newcommand{\formon}[1]{\ensuremath{\formonsym(#1)}}
\newcommand{\formoff}[1]{\ensuremath{\formoffsym(#1)}}
\newcommand{\actelim}[1]{\ensuremath{\actelimsym(#1)}}
\newcommand{\formdissym}{\ensuremath{\phi_{\mathit{of}}}}
\newcommand{\formdis}[1]{\ensuremath{\formdissym(#1)}}
\newcommand{\elim}[1]{\ensuremath{\# #1}}

\newcommand{\VBar}[2]{($ (0,#1) + #2 $) -- ($(0,-#1)+ #2 $)}
\newcommand{\Above}[2]{($ (0,#1) + #2 $)}
\newcommand{\Below}[2]{($ (0,-#1) + #2 $)}

\usepackage{etoolbox}
\makeatletter
\patchcmd{\hyper@makecurrent}{%
    \ifx\Hy@param\Hy@chapterstring
        \let\Hy@param\Hy@chapapp
    \fi
}{%
    \iftoggle{inappendix}{
        \@checkappendixparam{chapter}%
        \@checkappendixparam{section}%
        \@checkappendixparam{subsection}%
        \@checkappendixparam{subsubsection}%
        \@checkappendixparam{paragraph}%
        \@checkappendixparam{subparagraph}%
    }{}%
}{}{\errmessage{failed to patch}}

\newcommand*{\@checkappendixparam}[1]{%
    \def\@checkappendixparamtmp{#1}%
    \ifx\Hy@param\@checkappendixparamtmp
        \let\Hy@param\Hy@appendixstring
    \fi
}
\makeatletter

\newtoggle{inappendix}
\togglefalse{inappendix}

\apptocmd{\appendix}{\toggletrue{inappendix}}{}{\errmessage{failed to patch}}
\apptocmd{\subappendices}{\toggletrue{inappendix}}{}{\errmessage{failed to patch}}

\AtBeginDocument{%
  }

\forfull{
\pdfoutput=1
\hideLIPIcs
}

\title{Progress, Justness and Fairness in Modal \texorpdfstring{$\mu$}{mu}-Calculus Formulae} 

\author{Myrthe S. C. Spronck}{Eindhoven University of Technology, The Netherlands}{m.s.c.spronck@tue.nl}{https://orcid.org/0000-0003-2909-7515}{}
\author{Bas Luttik}{Eindhoven University of Technology, The Netherlands}{s.p.luttik@tue.nl}{https://orcid.org/0000-0001-6710-8436}{}
\author{Tim A. C. Willemse}{Eindhoven University of Technology, The Netherlands}{t.a.c.willemse@tue.nl}{https://orcid.org/0000-0003-3049-7962}{}

\authorrunning{M.\,S.\,C. Spronck, B. Luttik and T.\,A.\,C.\ Willemse} 

\Copyright{Myrthe S.C. Spronck, Bas Luttik and Tim A.C. Willemse} 



\forcamera{
\relatedversion{Arxiv version tba, we have to consider the possibility this may take up two lines. } 
\relatedversiondetails{Full Version}{URL to related version} 
}

\acknowledgements{We thank an anonymous reviewer for the observation that weak and strong hyperfairness must be distinguished.}

\nolinenumbers 
\EventEditors{Rupak Majumdar and Alexandra Silva}
\EventNoEds{2}
\EventLongTitle{35th International Conference on Concurrency Theory (CONCUR 2024)}
\EventShortTitle{CONCUR 2024}
\EventAcronym{CONCUR}
\EventYear{2024}
\EventDate{September 9--13, 2024}
\EventLocation{Calgary, Canada}
\EventLogo{}
\SeriesVolume{311}
\ArticleNo{25}

\begin{document}
\maketitle
\def\subsectionautorefname{Section}
\def\equationautorefname{Formula}

\begin{abstract}
When verifying liveness properties on a transition system, it is often necessary to discard spurious violating paths by making assumptions on which paths represent realistic executions.
Capturing that some property holds under such an assumption in a logical formula is challenging and error-prone, particularly in the modal \texorpdfstring{$\mu$}{mu}-calculus.
In this paper, we present template formulae in the modal \texorpdfstring{$\mu$}{mu}-calculus that can be instantiated to a broad range of liveness properties. We consider the following assumptions: progress, justness, weak fairness, strong fairness, and hyperfairness, each with respect to actions.
The correctness of these formulae has been proven.
\end{abstract}

\section{Introduction}

Formal verification through model checking requires a formalisation of the properties of the modelled system as formulae in some logic, such as LTL \cite{pnueli1977temporal}, CTL \cite{emerson1982using} or the modal $\mu$-calculus \cite{kozen1983results}.
In this paper, we focus on the modal $\mu$-calculus,
a highly expressive logic used in established model checkers such as mCLR2 \cite{mCRL2toolset} and CADP \cite{garavel2013cadp}.

A frequently encountered problem when checking liveness properties is that spurious violations are found, such as paths on which some components never make progress.
Often, such paths do not represent realistic executions of the system. 
It is then a challenge to restrict verification to those paths that do represent realistic system executions.
For this, we use completeness criteria \cite{glabbeek2020reactive,glabbeek2023modelling}: predicates on paths that say which paths are to be regarded as realistic runs of the system. These runs are called complete runs.
Examples of completeness criteria are progress, justness and fairness.

It turns out that writing a modal $\mu$-calculus formula for a property being satisfied under a completeness criterion is non-trivial.
Since the $\mu$-calculus is a branching-time logic, we cannot separately formalise when a path is complete and when it satisfies the property, and then combine the two formalisations with an implication.
Instead, a more intricate integration of both aspects of a path is needed. Our aim is to achieve such an integration for a broad spectrum of liveness properties and establish the correctness of the resulting formulae. To this end, we shall consider a template property that can be instantiated to a plethora of liveness properties and, in particular, covers all liveness property patterns of \cite{dwyer1999patterns}. Then, we present modal $\mu$-calculus formulae integrating the completeness criteria of progress, justness, weak fairness, strong fairness, and hyperfairness with this template property.

As discussed in \cite{glabbeek2015ccs}, for the formulation of realistic completeness criteria it is sometimes necessary to give special treatment to a set of blocking actions, i.e., actions that require cooperation of the environment in which the modelled system operates. Our template formulae are therefore parameterised with a set of blocking actions. We shall see that, given a set of blocking actions, there are two different interpretations of hyperfairness; we call these weak and strong hyperfairness.

Regarding our presented formulae, the progress formula is similar to those commonly used for liveness properties even when completeness is not explicitly considered.
Our formulae for justness, weak fairness and weak hyperfairness only subtly differ from each other. We characterise the similarities these three share and give a generic formula that can be adapted to represent all completeness criteria that meet these conditions.
Lastly, we observe that strong fairness and strong hyperfairness do not meet these conditions. We give alternative formulae that are significantly more complex. Whether more efficient formulae for these completeness criteria exist remains an open problem.

Modal $\mu$-calculus formulae are often hard to interpret. Accordingly, it is not trivial to see that our formulae indeed express the integration of liveness properties with completeness criteria. \altfull{We have therefore included elaborate correctness proofs in the appendices.}{We therefore include elaborate correctness proofs in the full version of this paper.}

Our work is essentially a generalisation along two dimensions (viz., the completeness criterion and the liveness property) of the works of \cite{remenska2016bringing} and \cite{bouwman2020off,stomp1989mu}.
In \cite{remenska2016bringing}, the tool PASS is presented for automatically translating common property patterns into modal $\mu$-calculus formulae.
Some of those patterns integrate an assumption that excludes paths deemed unrealistic, but since the exact assumption is not stated separately, we cannot make a formal comparison with our approach. 
In \cite{bouwman2020off}, a formula for justness is presented, covering one of the properties we cover. This formula forms the basis for our justness, weak fairness and weak hyperfairness formulae.
Our formulae for strong fairness and strong hyperfairness are in part inspired by the formula for termination under strong fairness presented in \cite{stomp1989mu}.

The organisation of this paper is as follows. 
In \autoref{sec:prelim} we recap the relevant definitions on labelled transition systems, as well as the syntax and semantics of the modal $\mu$-calculus.
In \autoref{sec:motivation}, we motive our work with an example, and in \autoref{sec:completeness} we give the completeness criteria we cover in this paper.
In \autoref{sec:dwyer}, we formally identify the class of liveness properties we study and relate it to a popular class of properties. 
Our template formulae are presented in \autoref{sec:formulae}, combining the completeness criteria from \autoref{sec:completeness} with the property template from \autoref{sec:dwyer}.
We give a small application example in \autoref{app:example} and discuss the scope of our work in \autoref{subsec:whyliveness}.
Finally, we give our conclusions in \autoref{sec:conclusion}.

\section{Preliminaries}\label{sec:prelim}

We represent models as labelled transition systems (LTSs).
In this section, we briefly introduce the relevant definitions on LTSs, as well as the modal $\mu$-calculus.

\subsection{Labelled Transition Systems}\label{subsec:LTS}

\begin{definition}
    An LTS is a tuple $M = (\states, \initstate, \actionset, \transrel)$ where
    \begin{itemize}
        \item $\states$ is a set of states,
        \item $\initstate \in \states$ is the initial state,
        \item $\actionset$ is a set of action labels, also referred to as the alphabet of the LTS, and
        \item $\transrel \subseteq \states \times \actionset \times \states$ is a transition relation.
    \end{itemize}
\end{definition}
In this paper, we only consider finite LTSs, such as the kind used in finite-state model checking. 
In particular, our formulae are proven correct under the assumption that $\actionset$ is finite.
We write $s \xrightarrow{a} s'$ as shorthand for $(s, a, s') \in \transrel$, and for a given transition $t = (s, a, s')$ we write $\source{t} = s$, $\action{t} = a$ and $\target{t} = s'$.

For the definitions below, we fix an LTS $M = (\states, \initstate, \actionset, \transrel)$.

\begin{definition}
    A \emph{path} is an $($alternating$)$ sequence $\pi = s_0 t_1 s_1 t_2 \ldots$ of states $s_0, s_1, \ldots \in \states$ and transitions $t_1, t_2, \ldots \in \transrel$. A path must start with a state, and must be either infinite, or end in a state. In the latter case, the end of the path is referred to as the \emph{final state}. For all $i \geq 0$, $t_{i+1}$ must satisfy $\source{t_{i+1}} = s_i$ and $\target{t_{i+1}} = s_{i+1}$.
\end{definition}
We sometimes refer to transitions on a path as steps.
We say an action occurs on a path if a transition labelled with that action is on the path. 
We call a path on which no action in some set $\alpha$ occurs an \emph{$\alpha$-free} path.
One path can be appended to another: let $\pi' = s_0' t_1' s_1' \ldots t_n' s_n'$ and $\pi'' = s_0'' t_1'' s_1'' \ldots$, where $\pi'$ must be finite and $\pi''$ may be finite or infinite. Then the path $\pi$ defined as $\pi''$ appended to $\pi'$ is written as $\pi = \pi'\co \pi'' = s_0' t_1' s_1' \ldots t_n' s_n' t_1'' s_1'' \ldots$. 
This is only allowed when $s_n' = s_0''$.

\begin{definition}
We say that:
    \begin{itemize}
    \item A transition $t \in \transrel$ is \emph{enabled} in a state $s \in \states$ if, and only if, $\source{t} = s$. 
    \item An action $a \in \actionset$ is \emph{enabled} in a state $s \in \states$ if, and only if, there exists a transition $t \in \transrel$ with $\action{t} = a$ that is enabled in $s$.
    \item An action $a \in \actionset$ is \emph{perpetually enabled} on a path $\pi$ if $a$ is enabled in every state of $\pi$.
    \item An action $a \in \actionset$ is \emph{relentlessly enabled} on a path $\pi$ if every suffix of $\pi$ contains a state in which $a$ is enabled.
    \item A state without enabled actions is called a \emph{deadlock state}.
    \end{itemize}
\end{definition}
Every action that is perpetually enabled on a path is also relentlessly enabled on that path.

\subsection{Modal \texorpdfstring{$\mu$}{mu}-Calculus}\label{subsec:mucalc}

The modal $\mu$-calculus is given in \cite{kozen1983results}. Our presentation of the logic is based on \cite{bradfield2001modal,bradfield200712,bradfield2018mu,mCRL2language}. 

The syntax of the modal $\mu$-calculus is described by the following grammar, in which $a$ ranges over the set of actions $\actionset$, and $X$ ranges over a set of formal variables $\varset$.
$$\phi, \psi ::= \fp \mid X \mid \neg\phi \mid \phi \lor \psi \mid \diam{a}\phi \mid \lfp{X}{\phi}$$
Here $\fp$ is false; $\neg$ represents negation; $\lor$ is disjunction; $\diam{~}$ is the diamond operator; and $\mu$ is the least fixpoint
operator.
We say that $\lfp{X}{\phi}$ binds $X$ in $\phi$. Variables that are unbound in a formula are free, and a formula without free variables is closed.

A modal $\mu$-calculus formula $\phi$ must both adhere to this grammar and be \emph{syntactically monotonic}, meaning that for every occurrence of $\mu X.\psi$ in $\phi$, every free occurrence of $X$ in $\psi$ must always be preceded by an even number of negations.

We give the semantics of a modal $\mu$-calculus formula $\phi$ with respect to an arbitrary LTS $M = (\states, \initstate, \actionset, \transrel )$ and environment $\env : \varset \to 2^{\states}$.
\begin{align*}
    &\llbracket \fp \rrbracket_{\env}^{M} = \emptyset &
    &\llbracket \phi \lor \psi \rrbracket_{\env}^{M} = \llbracket \phi \rrbracket_{\env}^{M} \cup \llbracket \psi \rrbracket_{\env}^{M} &\\
    &\llbracket X \rrbracket_{\env}^{M} = \env(X) &
    &\llbracket \diam{a}\phi \rrbracket_{\env}^{M} = \left\{ s \in \states \mid \exists_{s' \in \states}. s \xrightarrow{a} s' \land s' \in \llbracket \phi \rrbracket_{\env}^{M}\right\} &\\
    &\llbracket \neg \phi \rrbracket_{\env}^{M} = \states \setminus \llbracket \phi \rrbracket_{\env}^{M} &
    &\llbracket \lfp{X}{\phi} \rrbracket_{\env}^{M} = \bigcap\left\{\states' \subseteq \states \mid \states' \supseteq \llbracket \phi \rrbracket_{\update{\env}{X}{\states'}}^{M}\right\}&
\end{align*}
In contexts where the model is fixed, we drop the $M$ from $\llbracket \phi \rrbracket_{\env}^{M}$. Additionally, we drop $\env$ when the environment does not affect the semantics of the formula, e.g. with closed formulae.

We use conjunction, $\land$, and implication, $\imps$, as the usual abbreviations.
We also add several abbreviations: $\tp = \neg\fp$ for true; $\boxm{a}\phi = \neg\diam{a}\neg\phi$ for the box operator; and $\nu X.\phi = \neg\mu X.(\neg\update{\phi}{X}{\neg X})$ for the greatest fixpoint.

To express formulae more compactly, we extend our syntax to allow regular expressions over finite sets of actions to be used in the box and diamond operators.
Since we limit this to finite sets of actions, the syntactical extension does not increase the expressivity of the logic, it merely simplifies the presentation.
This is a common extension of the $\mu$-calculus syntax, for instance shown in \cite{mCRL2language}, based on the operators defined for PDL \cite{fisher1979pdl}.
We overload the symbol for a single action to also represent the singleton set containing that action.
We use union, intersection, set difference, and set complement to describe sets of actions as usual. 
Regular expressions over sets of actions, henceforth referred to as \emph{regular formulae}, are defined by the following grammar:
$$R, Q ::= \emptyseq \mid \alpha \mid R \co Q \mid R + Q \mid \clos{R}$$
The empty sequence is represented by $\emptyseq$, and $\alpha$ ranges over sets of actions. The symbol $\co$ represents concatenation, $+$ the union of formulae, and $\clos{}$ is closure under repetition.

We define the meaning of the diamond operator over the new regular formulae as abbreviations of standard modal $\mu$-calculus formulae:
\begin{align*}
    \diam{\emptyseq}\phi &= \phi & \diam{\alpha}\phi &= \bigvee_{a \in \alpha}\diam{a}\phi & \diam{R \co Q}\phi &= \diam{R}\diam{Q}\phi\\
    \diam{R + Q}\phi &= \diam{R}\phi \lor \diam{Q}\phi & \diam{\clos{R}}\phi &=\lfpb{X}{\diam{R}X \lor \phi} & &
\end{align*}
The box operator is defined dually.
We say a path $\pi$ \emph{matches} a regular formula $R$ if the sequence of actions on $\pi$ is in the language of $R$.

\section{Motivation}\label{sec:motivation}

When analysing algorithms and systems, there are many different properties which may need to be checked.
For instance, when model checking mutual exclusion algorithms we want to check linear properties such as mutual exclusion and starvation freedom, but also branching properties such as invariant reachability of the critical section.
The modal $\mu$-calculus, which subsumes even CTL$^{\star}$, is able to express all these properties and more, and is therefore used in toolsets such as mCLR2 \cite{mCRL2toolset} and CADP \cite{garavel2013cadp}.

An issue that is frequently encountered when checking liveness properties in particular, is that the model admits executions that violate the property but do not represent realistic executions of the real system.
For example, models of algorithms that contain a busy waiting loop usually admit executions where processes do nothing except wait. 
Infinite loops can also be introduced by abstractions of reality, such as modelling a loop to represent an event that occurs an arbitrary, but finite, number of times.
Counterexamples that are due to such modelling artefacts obscure whether the property is satisfied on all realistic executions.
The problem we address in this paper is how to avoid such counterexamples and check properties only on realistic executions.
We illustrate the problem with an example, which we also employ as a running example throughout this paper.

\begin{example}\label{ex:running}
    Consider the coffee machine modelled in \autoref{fig:runex}.
    When a user places an $\actname{order}$ for one or more cups of coffee, they are required to scan their payment $\actname{card}$. 
    If the user prefers using coinage, they switch the machine to its alternate mode ($\actname{to\_cash}$), and then pay in $\actname{cash}$. 
    In the alternate mode, the machine can be switched back using $\actname{to\_card}$.
    After payment, the machine will $\actname{brew}$ the cup(s) of coffee.
    This is modelled as a non-deterministic choice between a looping and a final $\actname{brew}$ action, since at least one cup was ordered.
    Finally, the coffee is $\actname{deliver}$ed and the machine awaits the next order.

    We consider three example properties.
    \begin{enumerate}
        \item \emph{Single order}: whenever an $\actname{order}$ is made, there may not be a second $\actname{order}$ until a $\actname{deliver}$ has taken place, $\boxm{\clos{\allact} \co \actname{order} \co \clos{\comp{\actname{deliver}}} \co \actname{order}}\fp$.
        \item \emph{Inevitable delivery}: whenever an $\actname{order}$ is made, there will inevitably be an occurrence of $\actname{deliver}$, $\boxm{\clos{\allact} \co \actname{order}}\mu X.(\diam{\allact}\tp \land \boxm{\comp{\actname{deliver}}}X)$.
        \item \emph{Possible delivery}: it is invariantly possible to eventually execute the $\actname{deliver}$ action, $\boxm{\clos{\allact}}\diam{\clos{\allact} \co \actname{deliver}}\tp$.
    \end{enumerate}
    
    The described problem occurs with \emph{inevitable delivery}: $s_0 t_1 s_1 t_4 (s_3 t_6)^{\omega}$ is a violating path, on which infinitely many cups are part of the same order.
    Similarly, $s_0 t_1 (s_1 t_2 s_2 t_3)^{\omega}$ violates the property because the user never decides on a payment method. 
    The first counterexample represents an impossible scenario, and the second gives information on problematic user behaviour but tells us little about the machine itself.
\end{example}
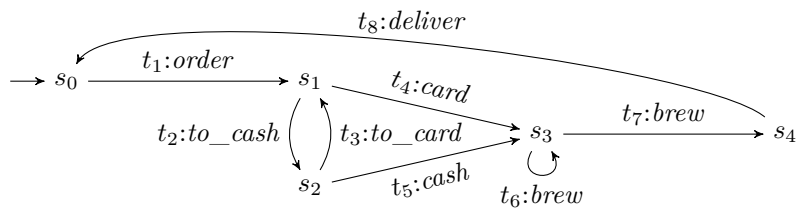
\begin{figure}
    \centering
 \begin{tikzpicture}[->,>=stealth',auto,node distance=10pt and 75pt,initial text=]
  \tikzstyle{every state}=[draw=black,rounded
  rectangle,text=black,inner sep=1mm,minimum size=.5cm]
  \node [initial left] (s0) {$s_0$};
  \node [right=of s0] (s1) {$s_1$};
  \node [below=of s1] (dummy1) {};
  \node [below=of dummy1] (s2) {$s_2$};
  \node [right=of dummy1] (s3) {$s_3$};
  \node [right=of s3] (s4) {$s_4$};
  
  \path[->]
    (s0) edge node[midway,above] {$t_1$:$\actname{order}$} (s1)
    (s1) edge[bend right] node[midway,anchor=east] {$t_2$:$\actname{to\_cash}$} (s2)
    (s2) edge[bend right] node[midway,anchor=west] {$t_3$:$\actname{to\_card}$} (s1)
    (s1) edge node[sloped] {$t_4$:$\actname{card}$} (s3)
    (s2) edge node[swap,sloped] {$t_5$:$\actname{cash}$} (s3)
    (s3) edge[in=300, out=240, loop, looseness=4.8] node[below] {$t_6$:$\actname{brew}$} (s3)
    (s3) edge node[midway,above] {$t_7$:$\actname{brew}$} (s4)
    (s4) edge[out=140, in=60,looseness=0.32] node[above] {$t_8$:$\actname{deliver}$} (s0)
    ;
\end{tikzpicture}
    \caption{The LTS for the running example.}
    \label{fig:runex}
\end{figure}

The kind of spurious counterexamples discussed in the example above primarily occur when checking liveness properties. We therefore focus on liveness properties, such as \emph{inevitable delivery}, in this paper. We will briefly discuss safety properties in \autoref{subsec:whyliveness}. 

There are ad-hoc solutions to exclude unrealistic counterexamples, e.g.\ altering the model to remove the unrealistic executions, or tailoring the formula to exclude specific problematic counterexamples \cite{groote2021tutorial}. Such ad-hoc solutions are undesirable because they clutter the model or the formula, and are therefore error-prone. We aim for a more generic solution, of which the correctness can be established once and for all. Such a generic solution requires, on the one hand, a general method to distinguish between realistic and unrealistic executions, and, on the other hand, a general class of liveness properties.

A general method to distinguish between realistic and unrealistic executions is provided by \emph{completeness criteria} \cite{glabbeek2020reactive,glabbeek2023modelling}, i.e., predicates on paths that label some as complete and all others as incomplete.
If a property is satisfied on all complete paths, it is satisfied under the given completeness criterion.
Completeness criteria give us a model-independent way to determine which paths are unrealistic, and therefore a generic solution to the stated problem.
Depending on the property and the model, we may prefer a different completeness criterion.
We therefore consider several criteria instead of fixing one specific criterion.
These completeness criteria are discussed in \autoref{sec:completeness}.

To find a general class of liveness properties, we take the \emph{property specification patterns} (PSP) of~\cite{dwyer1999patterns} as a starting point.
Since the modal $\mu$-calculus as presented in \autoref{subsec:mucalc} supports references to action occurrences but not state information, we specifically interpret these patterns on action occurrences. 
Our first contribution, in \autoref{sec:dwyer}, will be to characterise a class of liveness properties that subsumes all
liveness properties expressible in PSP.
Our second and main contribution is then presented in \autoref{sec:formulae}, where we combine the identified completeness criteria with our class of liveness properties, yielding template formulae for each combination.

\section{Completeness Criteria}\label{sec:completeness}

It is often assumed, sometimes implicitly, that as long as a system is capable of executing actions, it will continue to do so \cite{glabbeek2019progress}.
One could consider this the ``default'' completeness criterion, also known as \emph{progress} \cite{glabbeek2020reactive}; it says that only paths that are infinite or end in a deadlock state model complete runs and are hence complete paths.
We first present a modified version of the progress assumption that allows some actions to be blocked by the environment. 
We then define the other completeness criteria considered in this paper.
As already remarked in the previous section, the modal $\mu$-calculus is most suited to reasoning about action occurrences. Hence, we focus on completeness criteria
defined on action labels. For more general definitions on sets of transitions, see \cite{glabbeek2019progress}.

\subsection{Progress with Blocking Actions}
In \cite{glabbeek2015ccs}, it is argued that it is useful to consider some actions of an LTS as blocking.
A blocking action is an action that depends on participation by the environment of the modelled system. Consequently, even when such an action is enabled in a state because the system is willing to perform it, it may not be possible for the action to occur because the environment is uncooperative. 
In this paper, we refer to the set of blocking actions as $\block \subseteq \actionset$, and the set of non-blocking actions as $\nonblock = \actionset\setminus\block$. Which actions are in $\block$ is a modelling choice.

The default progress assumption can be adapted to account for blocking actions \cite{glabbeek2019justness,glabbeek2019progress}.
\begin{definition}
    A state $s \in \states$ is a \emph{$\block$-locked state} if, and only if, all actions enabled in $s$ are in $\block$.
    A path $\pi$ is \emph{$\block$-progressing} if, and only if, it is infinite or ends in a $\block$-locked state. 
\end{definition}

We refer to the assumption that only $\block$-progressing paths represent complete executions as $\block$-progress.
The ``default'' completeness criterion is equivalent to $\emptyset$-progress.

\begin{example}\label{ex:blockingacts}
    Consider \autoref{fig:runex}. Here, $\actname{order}$ is an environment action, since it involves the user. If we do not assume that there will always be a next user, we should add $\actname{order}$ to $\block$. 
    In some cases, we may want to consider the possibility that the machine is broken and not capable of producing coffee. In those cases, we should add $\actname{brew}$ to $\block$. 
    Our choice of $\block$ affects which paths are progressing: $s_0 t_1 s_1 t_ 4 s_3$ is not $\emptyset$-progressing, but it is $\{\actname{brew}\}$-progressing.
\end{example}

All completeness criteria we discuss in this paper are parameterised with a set of blocking actions.
The justness and fairness assumptions discussed in the remainder of this section label paths as incomplete if certain actions do not occur. Since it can never be assumed that the environment supports the occurrence of blocking actions, we do not want justness and fairness to label paths as incomplete due to the non-occurrence of blocking actions.

For readability the prefix $\block$- will sometimes be dropped from the names of the completeness criteria and their acronyms.
From this point, we will always discuss completeness criteria with respect to a set of blocking actions. 

\subsection{Justness}

Justness \cite{glabbeek2019justness,glabbeek2019progress} is a natural extension of progress to exclude infinite paths instead of finite paths.
The idea is that in addition to the system as a whole progressing, individual components in that system should also be able to make progress unless they are prevented from doing so by other components.
It is a weaker, and hence frequently more justifiable, assumption than the fairness assumptions we cover in the next section.
In its original presentation, justness is defined with respect to sets of transitions. Which components contribute to a transition and how they contribute to them determines which transitions interfere with each other.
We here consider justness defined with respect to actions instead, based on \cite{bouwman2020off}.
We do not go into how it is determined which actions interfere with each other here. For discussions on this topic and when the two definitions coincide, see \cite{bouwman2023supporting,bouwman2020off,glabbeek2019justness}.

Intuitively, justness of actions says that if an action $a$ is enabled at some point of a path, then eventually some action that can interfere with the occurrence of $a$ must occur in that path. That action may be $a$ itself.
In order to formalise the concept of interference, we require the concept of a \emph{concurrency relation on actions}, $\conc$. 
\begin{definition}\label{def:concur-act} 
    Relation $\conc \subseteq \actionset \times \actionset$ is a \emph{concurrency relation on actions} 
    if, and only if:
    \begin{enumerate}
        \item $\conc$ is irreflexive, and
        \item for all $a \in \actionset$, if $\pi$ is a path from a state $s$ in which $a$ is enabled to a state $s' \in \states$ such that $a \conc b$ for all actions $b$ occurring in $\pi$, then $a$ is enabled in $s'$. 
    \end{enumerate}
\end{definition}
We write $\nconc$ for the complement of $\conc$. Note that $\conc$ may be asymmetric.

Read $a \conc b$ as ``$a$ is concurrent with $b$'', and $a \nconc b$ as ``$b$ interferes with $a$'' or ``$b$ eliminates $a$''.
A labelled transition system can be extended with a concurrency relation on actions, which produces a \emph{labelled transition system with concurrency} (LTSC).

We here present the definition for justness of actions with blocking actions.
\begin{restatable}{definition}{defja}\label{def:ja} 
    A path $\pi$ satisfies \emph{$\block$-justness of actions} \emph{($\block$-JA)} if, and only if, for each action $a \in \nonblock$ that is enabled in some state $s$ in $\pi$, 
    an action $a' \in \actionset$ occurs in the suffix $\pi'$ of $\pi$ starting in $s$ such that $a \nconc a'$.
\end{restatable}

\begin{example}\label{ex:ja}
    Consider \autoref{fig:runex}, specifically the path $s_0 t_1 (s_1 t_2 s_2 t_3)^{\omega}$.
    On this path the user keeps switching the mode of the machine, without paying.
    To see if this path satisfies $\emptyset$-JA, we need a concrete $\conc$.
    Consider a $\conc$ such that $\actname{card} \nconc \actname{to\_cash}$, $\actname{cash} \nconc \actname{to\_card}$, and $a \nconc a$ for all action labels $a$.
    These are all required for $\conc$ to be a valid concurrency relation. This is because by \autoref{def:concur-act}, $\conc$ must be irreflexive, and when an action is enabled it must remain enabled on any path on which no interfering action occurs. 
    Since $\actname{card}$ is enabled in $s_1$ but not $s_2$, it must be the case that $\actname{card} \nconc \actname{to\_cash}$. Similarly, we must have  $\actname{cash} \nconc \actname{to\_card}$.
    With such a concurrency relation, the path satisfies $\emptyset$-JA since every action that is enabled is subsequently eliminated.
    In this LTS, there is no valid choice of $\conc$ that makes this path violate $\emptyset$-JA.
    However, if we modify \autoref{fig:runex} by replacing both $\actname{card}$ and $\actname{cash}$ with the action $\actname{pay}$, then \autoref{def:concur-act} does not enforce that $\actname{to\_cash}$ and $\actname{to\_card}$ interfere with the actions on $t_4$ and $t_5$, since $\actname{pay}$ is enabled in both $s_1$ and $s_2$. 
    We can choose whether $\actname{pay} \conc \actname{to\_cash}$ and $\actname{pay} \conc \actname{to\_card}$. If $\actname{pay}$ is concurrent with both, then the path $s_0 t_1 (s_1 t_2 s_2 t_3)^{\omega}$ violates $\emptyset$-JA. If either interferes with $\actname{pay}$, then the path satisfies $\emptyset$-JA.    
\end{example}

\subsection{Fairness}

There are situations where we want to exclude a larger set of infinite paths than those excluded by justness, or where we do not have a concurrency relation.
For this, we can use what are called \emph{fairness assumptions} in the literature.
These are a class of predicates on paths that distinguish between \emph{fair} and \emph{unfair} infinite paths.
It is assumed that only the fair paths are complete.
For an overview of many common fairness assumptions, see \cite{glabbeek2019progress}.
In this paper, we consider weak fairness of actions, strong fairness of actions, and (weak and strong) hyperfairness of actions. 
Each of the assumptions we discuss has the general shape, adapted from \cite{apt1988appraising}, ``if it is sufficiently often possible for an action to occur, it will occur sufficiently often''.
What it means for an action to be ``sufficiently often possible'' and ``occur sufficiently often'' depends on the exact assumption.

We first discuss \emph{weak fairness of actions}, which says that actions that are always enabled must eventually occur.
It is one of the most commonly discussed fairness assumptions.
We define weak fairness of actions formally, with respect to a set of blocking actions $\block$.
\begin{restatable}{definition}{defwfa}\label{def:wfa} 
    A path $\pi$ satisfies \emph{$\block$-weak fairness of actions} \emph{($\block$-WFA)} if, and only if, for every suffix $\pi'$ of $\pi$, every action $a \in \nonblock$ that is perpetually enabled in $\pi'$ occurs in $\pi'$.
\end{restatable}

\begin{example}\label{ex:wfa}
    Consider again \autoref{fig:runex}, with $\actname{card}$ and $\actname{cash}$ both replaced by $\actname{pay}$. 
    Then the path $s_0 t_1 (s_1 t_2 s_2 t_3)^{\omega}$ violates $\emptyset$-WFA, since $\actname{pay}$ is perpetually enabled in a suffix of this path without occurring. 
    If there are two separate actions for paying with cash or card, the path satisfies $\emptyset$-WFA because no actions are perpetually enabled in any suffix.
\end{example}

Next, \emph{strong fairness of actions} says that on a path, all actions that are enabled infinitely often, must occur infinitely often.
Formally, we define strong fairness of actions as:
\begin{restatable}{definition}{defsfa}\label{def:sfa} 
    A path $\pi$ satisfies \emph{$\block$-strong fairness of actions} \emph{($\block$-SFA)} if, and only if, for every suffix $\pi'$ of $\pi$, every action $a \in \nonblock$ that is relentlessly enabled in $\pi'$ occurs in $\pi'$.
\end{restatable}

Strong fairness is a stronger assumption than weak fairness, since it classifies more paths as incomplete. 
This follows from perpetual enabledness implying relentless enabledness.

\begin{example}\label{ex:sfa}
    The path $s_0 t_1 (s_1 t_2 s_2 t_3)^{\omega}$ in \autoref{fig:runex} satisfies $\emptyset$-WFA since there are no perpetually enabled actions in any suffix of the path.
    However, $\actname{cash}$ is relentlessly enabled in suffixes of this path, and yet does not occur. Hence, this path violates $\emptyset$-SFA.
\end{example}

Finally, we discuss \emph{hyperfairness of actions}. 
Informally, it says that on all fair paths, every action that can always become enabled must occur infinitely often.
The idea is that if there is always a reachable future where the action occurs, then it is merely unlucky if the action does not occur infinitely often.
The concept of hyperfairness is introduced and named in \cite{attie1989fairness}.
For our presentation of hyperfairness, we use the generalisation from \cite{lamport2000fairness}. 
We first formalise what it means that an action ``can become'' enabled, by defining \emph{reachability}.
\begin{definition}
    We say that:
    \begin{itemize}
        \item A state $s \in \states$ is \emph{$\block$-reachable} from some state $s' \in \states$ if, and only if, there exists a $\block$-free path starting in $s'$ that ends in $s$.
        \item An action $a \in \actionset$ is \emph{$\block$-reachable} from some state $s \in \states$ if, and only if, there exists a state $s' \in \states$ that is $\block$-reachable from $s$ and in which $a$ is enabled.
        \item A state $s \in \states$ or action $a \in \actionset$ is \emph{perpetually $\block$-reachable} on a path $\pi$ if, and only if, it is $\block$-reachable from every state of $\pi$.
        \item A state $s \in \states$ or action $a \in \actionset$ is \emph{relentlessly $\block$-reachable} on a path $\pi$ if, and only if, every suffix of $\pi$ contains a state from which it is $\block$-reachable.
    \end{itemize}
\end{definition}

From the intuitive description of hyperfairness, it is clear it is a variant of weak or strong fairness with reachability instead of enabledness, giving us two possible definitions of hyperfairness. We name the two interpretations weak hyperfairness and strong hyperfairness respectively. Both interpretations of hyperfairness are reasonable, and in fact when not considering blocking actions, they coincide \cite{lamport2000fairness}.
However, this is not the case when blocking actions are included in the definitions. We therefore consider both variants.

\begin{restatable}{definition}{defwhfa}\label{def:whfa}
    A path $\pi$ satisfies \emph{weak $\block$-hyperfairness of actions} \emph{($\block$-WHFA)} if, and only if, for every suffix $\pi'$ of $\pi$, every action $a \in \nonblock$ that is perpetually $\block$-reachable in $\pi'$ occurs in~$\pi'$.
\end{restatable}
\begin{restatable}{definition}{defshfa}\label{def:shfa}
    A path $\pi$ satisfies \emph{strong $\block$-hyperfairness of actions} \emph{($\block$-SHFA)} if, and only if, for every suffix $\pi'$ of $\pi$, every action $a \in \nonblock$ that is relentlessly $\block$-reachable in $\pi'$ occurs in~$\pi'$.
\end{restatable}
Since enabledness implies reachability, WHFA is stronger than WFA, and SHFA is stronger than SFA. Perpetually reachability implies relentless reachability, so SHFA is also stronger than WHFA. 
However, as the next examples will show, SFA and WHFA are incomparable.

\begin{example}\label{ex:hfa}
    The impact of hyperfairness can clearly be seen when non-determinism is used.
    Consider the path $s_0 t_1 s_1 t_4 (s_3 t_6)^{\omega}$ in \autoref{fig:runex}.
    This path satisfies $\emptyset$-SFA, since the only action that is relentlessly enabled on this path, $\actname{brew}$, also occurs infinitely often.
    However, as long as $\actname{deliver} \not \in \block$ and $\actname{brew} \not \in \block$, this path does not satisfy $\block$-WHFA or $\block$-SHFA: $\actname{deliver}$ is $\block$-reachable from $s_3$, and therefore is perpetually and relentlessly $\block$-reachable in a suffix of this path, but does not occur. We here see $\block$-SFA does not imply $\block$-WHFA.
\end{example}

\begin{example}\label{ex:hyperdiff}
    In \autoref{fig:runex}, consider $s_0 t_1 (s_1 t_2 s_2 t_3)^{\omega}$ with $\block = \{\actname{order}, \actname{to\_cash}, \actname{to\_card}\}$. 
    This path satisfies $\block$-WHFA because $\actname{card}$ and $\actname{cash}$ are only $\block$-reachable from $s_1$ and $s_2$ respectively. 
    They are not perpetually $\block$-reachable in any suffix of this path, therefore $\block$-WHFA is satisfied.
    However, they are relentlessly $\block$-reachable, so $\block$-SHFA is violated. This demonstrates that $\block$-WHFA and $\block$-SFHA do not coincide when blocking actions are considered.
    The actions $\actname{card}$ and $\actname{cash}$ are also relentlessly $\block$-enabled, so $\block$-SFA is also violated. Hence, $\block$-WHFA does not imply $\block$-SFA.
\end{example}

\section{A Generalisation of the Property Specification Liveness Patterns}
\label{sec:dwyer}

Dwyer, Avrunin and Corbett observed that a significant majority of properties that are used in practice can be fit into a set of property specification patterns \cite{dwyer1999patterns}.
These patterns consist of a \emph{behaviour} that must be satisfied and a \emph{scope} within a path that delimits where the behaviour must be satisfied. \forfull{We recall the behaviours and scopes presented in \cite{dwyer1999patterns} in \autoref{app:psp}. }We focus on expressing properties that are captured by PSP.

Of all behaviours considered in \cite{dwyer1999patterns}, only \pattern{existence}, \pattern{existence at least}, \pattern{response} and \pattern{chain response} represent pure liveness properties.
The \pattern{global} and \pattern{after} scopes, when combined with any of these four behaviours, give liveness properties.\altfull{ We argue why only these patters of PSP represent pure liveness properties in \autoref{app:psp}.}{\footnote{In the full version, we recap PSP and argue why only these patterns represent pure liveness properties.}}
All other scopes result in safety properties or properties that combine safety and liveness.
Of those, we cover the \pattern{until} and \pattern{after-until} scopes, since we can incorporate those into our formulae with little difficulty.

For behaviours, \pattern{existence at least} says some action in a set $S_r$ must occur at least $k$ times in the scope; when $k = 1$ we call this \pattern{existence}. The \pattern{response} behaviour requires that whenever an action in a set $S_q$ occurs, it must be followed by the occurrence of an action in $S_r$. When chains of action occurrences are used instead of individual action occurrences, this is called \pattern{chain response}.
For the scopes, \pattern{global} refers to the full path and \pattern{after} to the path after the first occurrence of an action in a set $S_a$. The \pattern{until} scope refers to the path before the first occurrence of an action in a set $S_b$, or the full path if no such action occurs. Finally, \pattern{after-until} combines \pattern{after} and \pattern{until}, referring to every subpath of the path that starts after any occurrence of an action in $S_a$ and ends before the following occurrence of an action in $S_b$. If no action in $S_b$ occurs, the behaviour must still be satisfied after $S_a$.

\begin{example}\label{ex:propspatterns}
Consider again the properties we presented in \autoref{ex:running}.
    \emph{Single order} is \pattern{absence after-until}, with $S_a = \{\actname{order}\}$, $S_b = \{\actname{deliver}\}$ and $S_r = \{\actname{order}\}$.
    \emph{Inevitable delivery} is  \pattern{global response} with $S_q = \{\actname{order}\}$ and $S_r = \{\actname{deliver}\}$. 
    \emph{Possible delivery} does not fit into the patterns on occurrences of actions, since it contains a requirement on states, specifically that the state admits a path on which $\actname{delivery}$ occurs.
\end{example}

We want to create formulae for all 16 combinations of the selected behaviours and scopes.
To make our results more compact and generic, we first generalise these 16 patterns into a single template property.
This template works by describing the shape of a violating path for a property that fits one of these patterns.
Intuitively, this shape is:
``after the occurrence of $\dpre$, there are no occurrences of $\ddis$ up until the (optional) occurrence of $\den$''.
For our template formulae to be syntactically correct, it is important that $\dpre$ is a regular formula, describing the prefix that a violating path must have, whereas $\ddis$ and $\den$ are sets of actions. The actions in $\ddis$ are those that are forbidden from occurring after $\dpre$ on a violating path, whereas the actions in $\den$ indicate the end of the scope in which $\ddis$ may not occur.

We formalise this template as follows:
\begin{definition}
    A path $\pi$ is
    \emph{$(\dpre,\ddis,\den)$-violating} if, and only if, there exist $\pi_{\mathit{pre}}$ and $\pi_{\mathit{suf}}$ such that:
    \begin{enumerate}
    \item $\pi = \pi_{\mathit{pre}}\co\pi_{\mathit{suf}}$, and
    \item $\pi_{\mathit{pre}}$ matches $\dpre$, and
    \item $\pi_{\mathit{suf}}$ satisfies at least one of the following conditions:
    \begin{enumerate}
        \item $\pi_{\mathit{suf}}$ is $\ddis$-free, or
        \item $\pi_{\mathit{suf}}$ contains an occurrence of an action in $\den$, and the prefix of $\pi_{\mathit{suf}}$ before the first occurrence of an action in $\den$ is $\ddis$-free.
    \end{enumerate}
\end{enumerate}
\end{definition}
For readability, we frequently refer to $(\dpre,\ddis,\den)$-violating paths as violating paths.
We sometimes summarise condition 3 as ``$\pi_{\mathit{suf}}$ is $\ddis$-free up until the first occurrence of $\den$''. 
See \autoref{fig:violpath} for an illustration of what types of paths are considered violating.

\begin{figure}[htb]
    \centering
    \begin{tikzpicture}[node distance=25pt and 35pt, inner sep=0pt, outer sep = 0pt, minimum size = 0pt]
        \node [draw=none] (1-start_pre) {};
        \node [draw=none, right=of 1-start_pre, xshift=-15pt] (1-start_dis_free) {};
        \node [draw=none, right=of 1-start_dis_free, xshift=20pt] (1-end_dis_free) {};
        
        \node [draw=none, right=of 1-end_dis_free, xshift=-20pt] (2-start_pre) {};
        \node [draw=none, right=of 2-start_pre, xshift=-15pt] (2-start_dis_free) {};
        \node [draw=none, right=of 2-start_dis_free, xshift=20pt] (2-end_dis_free) {};
        \node [draw=none, right=of 2-end_dis_free,xshift=-20pt] (2-end) {};
        
        \node [draw=none, right=of 2-end, xshift=-20pt] (3-start_pre) {};
        \node [draw=none, right=of 3-start_pre, xshift=-15pt] (3-start_dis_free) {};
        \node [draw=none, right=of 3-start_dis_free, xshift=-5pt] (3-end_dis_free) {};
        \node [draw=none, xshift=25pt] (3-end) at (3-end_dis_free){};
        
        \node [draw=none, right=of 3-end, xshift=-20pt] (4-start_pre) {};
        \node [draw=none, right=of 4-start_pre, xshift=-15pt] (4-start_dis_free) {};
        \node [draw=none, right=of 4-start_dis_free, xshift=20pt] (4-end_dis_free) {};
        \node [draw=none, right=of 4-end_dis_free,xshift=-20pt] (4-end) {};
        
        \draw[very thick]
        (1-start_pre) -- node[above,yshift=6pt] {$\dpre$} (1-start_dis_free)
        (1-start_dis_free) -- node[above,yshift=6pt] {$\ddis$-free} (1-end_dis_free)
        ;
        
        \draw[very thick]
        (2-start_pre) -- node[above,yshift=6pt] {$\dpre$} (2-start_dis_free)
        (2-start_dis_free) -- node[above,yshift=6pt] {$\ddis$-free} (2-end_dis_free) {}
        (2-end_dis_free) -- (2-end) {}
        ;
        
        \draw[very thick]
        (3-start_pre) -- node[above,yshift=6pt] {$\dpre$} (3-start_dis_free)
        (3-start_dis_free) -- node[above,yshift=6pt,xshift=12.5pt] {$\ddis$-free} (3-end_dis_free)
        ;
        \draw[dashed, very thick] (3-end_dis_free) -- (3-end);
        
        \draw[very thick]
        (4-start_pre) -- node[above,yshift=6pt] {$\dpre$} (4-start_dis_free)
        (4-start_dis_free) -- node[above,yshift=6pt] {$\ddis$-free} (4-end_dis_free) {}
        ;
        \draw[dashed, very thick] (4-end_dis_free) -- (4-end);
        
        \draw[very thick] \VBar{6pt}{(1-start_pre)} \VBar{6pt}{(1-start_dis_free)} \VBar{6pt}{(1-end_dis_free)};
        
        \draw[very thick] \VBar{6pt}{(2-start_pre)} \VBar{6pt}{(2-start_dis_free)} \VBar{6pt}{(2-end_dis_free)} \VBar{6pt}{(2-end)};
        \draw[very thick] (2-end_dis_free) node[above,yshift=6pt] {$\den$} ;
        
        \draw[very thick] \VBar{6pt}{(3-start_pre)} \VBar{6pt}{(3-start_dis_free)} ;
        
        \draw[very thick] \VBar{6pt}{(4-start_pre)} \VBar{6pt}{(4-start_dis_free)} \VBar{6pt}{(4-end_dis_free)};
        \draw[very thick] (4-end_dis_free) node[above,yshift=6pt] {$\den$};
        
        \draw[draw=none,pattern={north west lines}, opacity=0.5] \Above{4pt}{(1-start_pre)} rectangle \Below{4pt}{(1-start_dis_free)};
        \draw[draw=none,pattern={north east lines}, opacity=0.5] \Above{4pt}{(1-start_dis_free)} rectangle \Below{4pt}{(1-end_dis_free)};
        
        \draw[draw=none,pattern={north west lines}, opacity=0.5] \Above{4pt}{(2-start_pre)} rectangle \Below{4pt}{(2-start_dis_free)};
        \draw[draw=none,pattern={north east lines}, opacity=0.5] \Above{4pt}{(2-start_dis_free)} rectangle \Below{4pt}{(2-end_dis_free)};
        
        \draw[draw=none,pattern={north west lines}, opacity=0.5] \Above{4pt}{(3-start_pre)} rectangle \Below{4pt}{(3-start_dis_free)};
        \draw[draw=none,pattern={north east lines}, opacity=0.5] \Above{4pt}{(3-start_dis_free)} rectangle \Below{4pt}{(3-end)};
        
        \draw[draw=none,pattern={north west lines}, opacity=0.5] \Above{4pt}{(4-start_pre)} rectangle \Below{4pt}{(4-start_dis_free)};
        \draw[draw=none,pattern={north east lines}, opacity=0.5] \Above{4pt}{(4-start_dis_free)} rectangle \Below{4pt}{(4-end_dis_free)};
    \end{tikzpicture}
    \caption{The four types of $(\dpre,\ddis,\den)$-violating paths: finite or infinite, and without or with $\den$. Always, it has a prefix matching $\dpre$ and is $\ddis$-free up until the first occurrence of an action in $\den$.}
    \label{fig:violpath}
\end{figure}
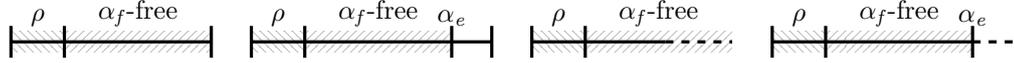

All 16 patterns can indeed be represented by the non-existence of $(\dpre,\ddis,\den)$-violating paths, albeit some more directly than others.
It turns out that $\dpre$, $\ddis$ and $\den$ can mostly be determined separately for behaviour and scope.
For these patterns, $\ddis$ is only affected by behaviour and $\den$ only by scope. However, we must split up the regular formula $\dpre$ into a behaviour component, $\dpreb$, and scope component, $\dpres$, such that $\dpre = \dpres \co \dpreb$.
See \autoref{tab:scopes} and \autoref{tab:behaviours} for how the variables should be instantiated for the four scopes and three of the four behaviours.
For a compact representation, we use $\sum$ to generalise the union operator on regular formulae ($+$).
We also use $x^i$ to represent $i$ concatenations of $x$, where $x^0 = \emptyseq$.

We do not include \pattern{chain response} in \autoref{tab:behaviours}, since it does not fit into a single formula.
However, it is possible to represent \pattern{chain response} as several \pattern{response} formulae placed in conjunction with each other.\altfull{ We include an example of this in \autoref{app:chain-resp}.}{\footnote{We give an example of this in the full version of this paper.}}

\begin{table}[hbt]
    \caption{Variable instantiation for templates.}
    \label{tab:variables}
    \centering
    \begin{subtable}{.38\linewidth}
        \caption{For scopes.}
        \label{tab:scopes}
        \centering
          \begin{tabular}{l|lll}
            \toprule
            Scope &  $\dpres$ & $\den$\\
            \midrule
            Global & $\emptyseq$ & $\noact$ \\
            Until & $\emptyseq$  & $S_b$ \\
            After & $\clos{\comp{S_a}} \co S_a$  & $\noact$ \\
            After-until & $\clos{\allact} \co S_a$ & $S_b$ \\
          \bottomrule
        \end{tabular} 
    \end{subtable}%
    \begin{subtable}{.62\linewidth}
        \caption{For behaviours.}
        \label{tab:behaviours}
        \centering
            \begin{tabular}{l|ll}
                \toprule
                Behaviour & $\dpreb$ & $\ddis$\\
                \midrule
                Existence & $\emptyseq$ & $S_r$ \\
                Existence at least $k$ & $\sum_{0 \leq i < k}(\clos{\comp{\den \cup S_r}} \co S_r)^i$ & $S_r$  \\
                Response & $\clos{\comp{\den}} \co S_q$ & $S_r$ \\
                Chain response & \altfull{See \autoref{app:chain-resp}}{-}& \forcamera{-}\\
                \bottomrule
            \end{tabular}
    \end{subtable}        
\end{table}

\section{Template Formulae}\label{sec:formulae}

In this section, we present the modal $\mu$-calculus formulae representing the non-existence of a violating path, as defined in \autoref{sec:dwyer}, that satisfies one of the completeness criteria from \autoref{sec:completeness}.
We express the non-existence of such a path, rather than expressing the equivalent notion that all complete paths satisfy the property, because we find the resulting formulae to be more intuitive.
We first present a formula for $\block$-progress only.
Subsequently, we give the formulae for weak fairness, weak hyperfairness and justness using a common structure all three share. 
Finally, we present the formulae for strong fairness and strong hyperfairness.
In the justness and fairness formulae, $\block$-progress is also included: these assumptions eliminate unrealistic infinite paths, but we still need progress to discard unrealistic finite paths.

\altfull{The proofs of all theorems in this section are included in \autoref{app:proofs}.}{Proofs of the theorems in this section are included in the full version of this paper. A sketch of the proof of \autoref{thm:gen} is included in \autoref{app:illustration} to illustrate our approach.}

\subsection{Progress}
A formula for the non-existence of a violating path without progress is uninteresting.
If progress is not assumed then all finite paths are complete, and therefore a path consisting of just $\dpre$ is a violating path whenever $\ddis \neq \emptyset$.
The non-existence of a violating path would then be captured by $\neg\diam{\dpre}\tp$.
This is why we include progress in all our formulae.

To represent progress, we must capture that as long as non-blocking actions are enabled, some transitions must still be executed.
The following formula captures the non-existence of violating paths under $\block$-progress:
\begin{equation}\label{mucalc:unfair}
    \neg\diam{\dpre}\nu X. (\diam{\den}\tp \lor \boxm{\nonblock}\fp \lor \diam{\comp{\ddis}}X)
\end{equation}
Intuitively, this formula says that there is no path that starts with a prefix matching $\dpre$, after which infinitely often a transition can be taken that is not labelled with an action in $\ddis$, or such transitions can be taken finitely often before a state is reached that is $\block$-locked or in which $\den$ is enabled. In the former case there is a $\block$-progressing path on which no actions in $\ddis$ occur after $\dpre$. If a state in which $\den$ is enabled is reached, then it is guaranteed a violating and $\block$-progressing path exists: by arbitrarily extending the path as long as non-blocking actions are still enabled, a $\block$-progressing and violating path can be constructed.
\begin{restatable}{theorem}{thmunfair}\label{thm:unfair}
    A state in an LTS satisfies \autoref{mucalc:unfair} if, and only if, it does not admit $\block$-progressing paths that are $(\dpre,\ddis,\den)$-violating. 
\end{restatable}

Since representing a liveness pattern without progress leads to uninteresting formulae, it is unsurprising that previous translations of PSP to the $\mu$-calculus have also implicitly included progress.
For instance, the translations from \cite{CADPRAFMC} for the liveness patterns of PSP are very similar to \autoref{mucalc:unfair}, albeit in positive form and without blocking actions.

\subsection{Weak Fairness, Weak Hyperfairness and Justness}\label{subsec:gen}
For weak fairness, weak hyperfairness and justness, we employ a trick inspired by the formula for justness presented in \cite{bouwman2020off} (which was in turn inspired by \cite{clarke1995efficient}): we translate a requirement on a full path into an invariant that can be evaluated within finitely many steps from every state of the path.
We illustrate this using weak fairness.

On every suffix of a weakly fair path, every perpetually enabled non-blocking action occurs.
To turn this into an invariant, we observe that we can evaluate a property on all suffixes of a path by evaluating it from every state of the path instead.
Next we must determine, within finitely many steps, if an action is perpetually enabled on a possibly infinite path. 
We do this by observing that if an action is not perpetually enabled, it must become disabled within finitely many steps.
An equivalent definition of WFA therefore is: a path $\pi$ satisfies WFA if, and only if, for every state $s$ in $\pi$, every action $a \in \nonblock$ that is enabled in $s$ occurs or becomes disabled within finitely many steps on the suffix of $\pi$ starting in $s$.
This translation of WFA determines three things for every non-blocking action $a$.
First, which actions may need to occur because of $a$; in the case of WFA this is $a$ itself.
Second, when those actions need to occur; for WFA this is when $a$ is enabled. We refer to this as the action being ``on''.
Finally, when those actions do not need to occur; for WFA this is when $a$ becomes disabled. We refer to this as the action being ``off''. When an action that was previously on becomes off, or one of the required actions occurs, we say the action is ``eliminated''.
By choosing different definitions for an action being on or off, and when an action is eliminated, we can also represent justness and weak hyperfairness in the same way.

We find that completeness criteria for which such a translation can be made can be represented using the same generalised formula.
We will present this formula and how to instantiate it for WFA, WHFA and JA.
However, we must first formalise what it means for a predicate on paths to be translatable to an invariant that can be evaluated within finitely many steps.
We introduce the term \emph{finitely realisable (path) predicates} for this purpose.
\begin{restatable}{definition}{defgen}\label{def:gen} 
    A path predicate $P$ is \emph{finitely realisable} if, and only if, there exist mappings $\formonsym$ and $\formoffsym$ from non-blocking actions to closed modal $\mu$-calculus formulae, and a mapping $\actelimsym$ from non-blocking actions to sets of actions, such that:
    \begin{enumerate}
        \item A path $\pi$ satisfies predicate $P$ if, and only if, all states $s$ on $\pi$ satisfy the following: for all $a \in \nonblock$, if $s$ satisfies $\formon{a}$ then the suffix $\pi'$ of $\pi$ starting in $s$ must contain an occurrence of some action in $\actelim{a}$ or a state that satisfies $\formoff{a}$.\label{gen:shape}
        \item A state $s$ is a $\block$-locked state if, and only if, $s \not\in\llbracket \formon{a}\rrbracket$ for all $a \in \nonblock$. \label{gen:locked}
        \item For every state $s$ and for all $a \in \nonblock$, $s \in \llbracket\formon{a}\rrbracket$ implies $s \not\in \llbracket\formoff{a}\rrbracket$. \label{gen:onoff}
        \item For all states $s$ and all $a \in \nonblock$ such that $s\in \llbracket\formon{a}\rrbracket$, if there exists a finite path $\pi$ from $s$ to a state $s'$ such that there is no occurrence of an action in $\actelim{a}$ on $\pi$ and there is no state on $\pi$ that satisfies $\formoff{a}$, then $s' \in \llbracket\formon{a}\rrbracket$. \label{gen:remains}
    \end{enumerate}
    We refer to these four properties as the \emph{invariant property}, the \emph{locking property}, the \emph{exclusive property} and the \emph{persistent property}, respectively.
\end{restatable}

The general formula for finitely realisable predicates is as follows: 
\begin{equation}\label{mucalc:gen}
    \neg\diam{\dpre}\nu X. (\bigwedge_{a \in \nonblock}(\formon{a} \imps 
    \diam{\clos{\comp{\ddis}}}(\diam{\den}\tp \lor (\formoff{a} \land X) \lor 
    \diam{\actelim{a} \setminus \ddis}X)))
\end{equation}
This formula has similarities to \autoref{mucalc:unfair}, particularly how $\dpre$ and $\den$ are integrated. The important part is that after $\dpre$, it must invariantly hold that all non-blocking actions for which $\formon{a}$ is satisfied are later eliminated.
An action $a$ is eliminated if, within finitely many steps, $\formoff{a}$ is satisfied or an action in $\actelim{a}$ occurs. In both cases, the invariant must once again hold. After $\dpre$, no actions in $\ddis$ may occur.
The formula works correctly for finite paths as well as infinite ones: if it is possible to reach a $\block$-locked state after $\dpre$ without taking actions in $\ddis$, then $X$ is satisfied due to the locking property, and a violating path is found. 

\autoref{mucalc:gen} is a template formula in two ways: $\dpre$, $\ddis$ and $\den$ determine what property is captured, and $\formonsym$, $\formoffsym$ and $\actelimsym$ determine the completeness criterion.
In this paper, we only cover how to instantiate the formula for WFA, WHFA and JA, but it can also be used for other finitely realisable predicates. However, the correctness proof of the formula depends on the criterion being \emph{feasible}. Feasibility on paths~\cite{apt1988appraising} is defined as follows.
\begin{restatable}{definition}{deffeas}\label{def:feas}
    A predicate on paths $P$ is \emph{feasible} if, and only if, for every LTS $M$, every finite path $\pi$ in $M$ can be extended to a path $\pi'$ that satisfies $P$ and is still a valid path in $M$.
\end{restatable}
\altfull{That WFA, WHFA and JA are feasible for finite LTSs is proven in \autoref{app:props}.}{That WFA, WHFA and JA are indeed feasible for finite LTSs is proven in the full version.}

\begin{restatable}{theorem}{thmgen}\label{thm:gen}
    For all feasible and finitely realisable path predicates $P$, it holds that an LTSC satisfies \autoref{mucalc:gen} if, and only if, its initial state does not admit $\block$-progressing paths that satisfy $P$ and are $(\dpre,\ddis,\den)$-violating. 
\end{restatable}
By instantiating the theorem for each completeness criterion, we derive the following:

\begin{restatable}{corollary}{corwfa}\label{cor:wfa}
    A state in an LTS satisfies \autoref{mucalc:gen} with $\formon{a} = \diam{a}\tp$, $\formoff{a} = \boxm{a}\fp$ and $\actelim{a} = \{a\}$ for all $a \in \nonblock$ if, and only if, it does not admit $\block$-progressing paths that satisfy $\block$-weak fairness of actions and are $(\dpre,\ddis,\den)$-violating. 
\end{restatable}
\begin{restatable}{corollary}{corwhfa}\label{cor:whfa}
    A state in an LTS satisfies \autoref{mucalc:gen} with $\formon{a} = \diam{\clos{\nonblock} \co a}\tp$, $\formoff{a} = \boxm{\clos{\nonblock} \co a}\fp$ and $\actelim{a} = \{a\}$ for all $a \in \nonblock$ if, and only if, it does not admit $\block$-progressing paths that satisfy weak $\block$-hyperfairness of actions and are $(\dpre,\ddis,\den)$-violating.
\end{restatable}
\begin{restatable}{corollary}{corja}\label{cor:ja}
    A state in an LTSC 
    satisfies \autoref{mucalc:gen} with $\formon{a} = \diam{a}\tp$, $\formoff{a} =\fp$ and $\actelim{a} = \{b \in \actionset \mid a \nconc b\}$ for all $a \in \nonblock$ if, and only if, it does not admit $\block$-progressing paths that satisfy $\block$-justness of actions and are $(\dpre,\ddis,\den)$-violating.
\end{restatable}

\subsection{Strong Fairness and Strong Hyperfairness}
SFA is not finitely realisable because we cannot observe within finitely many steps whether an action is relentlessly enabled: even if we observe several times that it is disabled, it may still be infinitely often enabled along the whole path.
Hence, we cannot use \autoref{mucalc:gen}.

Instead we observe that, on a path, actions that are not relentlessly enabled must eventually become perpetually disabled.
If the path is strongly fair, then all relentlessly enabled non-blocking actions occur infinitely often.
We can therefore say that a path is strongly fair if we can divide all non-blocking actions into two disjoint sets: those that occur infinitely often and those that eventually become perpetually disabled.
This observation is also made in \cite{stomp1989mu}, where a $\mu$-calculus formula for termination under strong fairness is given.

Using this idea, we give the following template formula for SFA:
\begin{equation}\label{mucalc:sfa}
        \neg\diam{\dpre\co\clos{\comp{\ddis}}} ( \diam{\den}\tp \lor \boxm{\nonblock}\fp \lor
        \bigvee_{\emptyset \neq F \subseteq \nonblock}\nu X. ( \bigwedge_{a \in F}\mu W. ( ( \bigwedge_{b \in \nonblock\setminus F}\boxm{b}\fp ) \land ( \diam{a \setminus\ddis}X \lor \diam{\comp{\ddis}}W ) ) ) )
\end{equation}
The use of negation, the exclusion of $\ddis$, and $\dpre$ in the diamond operator at the start of this formula are the same as in \autoref{mucalc:unfair}.
We explain the start of the formula after addressing the part starting with $\bigvee_{\emptyset \neq F \subseteq \nonblock}$.
Here, we use that on a strongly fair path, all non-blocking actions can be divided into those that occur infinitely often and those that become perpetually disabled.
The disjunction over subsets considers all possible ways of selecting some non-empty subset $F$ of $\nonblock$ that should occur infinitely often.
The greatest fixpoint states that infinitely often, all those actions must indeed occur within finitely many steps. Additionally, at no point may a non-blocking action not in $F$ be enabled. 
We exclude $F = \emptyset$ because the logic of the greatest fixed point formula we give relies on there being at least one $a$ in $F$. The special case that $F$ is empty and therefore a $\block$-locked state should be reached, is instead covered by explicitly considering $\boxm{\comp{\block}}\fp$ earlier in the formula.
Returning to the start of the formula, we allow a finite $\ddis$-free path before the greatest fixpoint is satisfied. The reason is that it may take several steps before all the non-blocking actions that are only finitely often enabled become perpetually disabled. 
Since we include a finite prefix already, we also add the cases that an action in $\den$ becomes enabled or that a $\block$-locked state is reached here, rather than deeper into the formula like in \autoref{mucalc:gen}.

\begin{restatable}{theorem}{thmsfa}\label{thm:sfa}
    An LTS satisfies \autoref{mucalc:sfa} if, and only if, its initial state does not admit $\block$-progressing paths that satisfy $\block$-strong fairness of actions and are $(\dpre,\ddis,\den)$-violating. 
\end{restatable}

Due to the quantification over subsets, the formula is exponential in the number of actions in $\nonblock$. 
Beyond small models, it is therefore not practical.
However, it can serve as a basis for future work.
For instance, if fairness is applied to sets of actions rather than individual actions, the formula is exponential in the number of sets instead, which may be smaller depending on how the sets are formed \cite{spronck2023fairness}.

We can adapt the formula for strong fairness to a formula for strong hyperfairness, by replacing perpetual disabledness of non-blocking actions not in $F$ with perpetual unreachability.
\begin{equation}\label{mucalc:shfa}
        \neg\diam{\dpre\co\clos{\comp{\ddis}}} ( \diam{\den}\tp \lor \boxm{\nonblock}\fp \lor
        \bigvee_{\emptyset \neq F \subseteq \nonblock}\nu X. (\bigwedge_{a \in F}\mu W. ( ( \bigwedge_{b \in \nonblock\setminus F}\boxm{\clos{\nonblock}\co b}\fp ) \land ( \diam{a \setminus\ddis}X \lor \diam{\comp{\ddis}}W ) ) ) )
\end{equation}

\begin{restatable}{theorem}{thmshfa}\label{thm:shfa}
    An LTS satisfies \autoref{mucalc:shfa} if, and only if, its initial state does not admit a $\block$-progressing path that satisfies strong $\block$-hyperfairness of actions and is $(\dpre,\ddis,\den)$-violating. 
\end{restatable}
Since we are not aware of other completeness criteria that fit the same structure, we do not provide a generalised formula here like we did with \autoref{mucalc:gen}\forfull{, although we do prove a more general theorem in \autoref{app:sgen-proof}}.

\section{Application Example}\label{app:example}
We here give an example of an application of the template formulae.
In \cite{groote2021tutorial}, several mutual exclusion algorithms are analysed using the mCRL2 toolset. 
Their analysis of Dekker's algorithm \cite{dijkstra1962over} presents the following modal $\mu$-calculus formula for starvation freedom of processes with id's 0 and 1. For clarity, the notation has been adjusted to match the previous sections and action names have been simplified. 
\begin{equation}\label{mucalc:tutorial}
    \boxm{\clos{\allact}} \bigwedge_{i \in \{0,1\}} \boxm{ \{ \actname{wish\_flag}(i,b) \mid b \in \mathbb{B}\}}\mu X. (\boxm{\comp{\actname{enter}(i)}}X \land \diam{\allact}\tp)
\end{equation} 
Starvation freedom is a \pattern{global response} property. In this case, the starvation freedom of a process $i$ is represented as an instantiation of the pattern with $S_q = \{\actname{wish\_flag}(i,b) \mid b \in \mathbb{B}\}$ and $S_r = \{\actname{enter}(i)\}$.
Indeed, the above formula is equivalent to:
\begin{equation}
    \bigwedge_{i \in \{0,1\}}\neg\diam{\clos{\allact} \co \{\actname{wish\_flag}(i,b) \mid b \in \mathbb{B}\}}\nu X. (\diam{\noact}\tp \lor \boxm{\allact}\fp \lor \diam{\comp{\actname{enter}(i)}}X)
\end{equation}
Observe that, when taking $\block = \emptyset$, the above matches a conjunction of two instances of \autoref{mucalc:unfair}, taking $\dpre$, $\den$ and $\ddis$ as suggested in \autoref{tab:variables} for \pattern{global response}.
Thus, this formula captures starvation freedom under $\emptyset$-progress. 
In \cite{groote2021tutorial}, it is reported that mCRL2 finds a violating path for this formula; a path which the authors note is unfair. The exact fairness assumption considered is not made concrete. 
As an ad-hoc solution, the modal $\mu$-calculus formula is adjusted to specifically ignore that counterexample.
Subsequently, mCRL2 finds another counterexample, which the authors again claim is unfair. Instead of creating yet another formula, they move on to Peterson's algorithm, which is deemed easier to analyse.
Using our template formulae, we can easily produce a formula for starvation freedom under several different completeness criteria. We give the formula for $\emptyset$-WFA, as an example.
\begin{multline}
    \bigwedge_{i \in \{0,1\}}\neg\diam{\clos{\allact} \co \{\actname{wish\_flag}(i,b) \mid b \in \mathbb{B}\}}\\
    \nu X. (\bigwedge_{a \in \allact}(\diam{a}\tp \imps 
    \diam{\clos{\comp{\actname{enter}(i)}}}(\diam{\noact}\tp \lor (\boxm{a}\fp \land X) \lor 
    \diam{a \setminus \actname{enter}(i)}X)))    
\end{multline}
We check this formula on the model from \cite{groote2021tutorial} using mCRL2.
Since mCRL2 only supports quantification over data parameters and not over actions, the conjunction over $\actionset$ must be written out explicitly.
The tool reports that the formula is violated. Examining the counterexample reveals this is because actions in the model do not show which process performs the action. Therefore, process $i$ reading value $v$ from a register $r$ is labelled with the same action as process $j$ reading $v$ from $r$. 
We add the responsible process to each action label, and also define $\block = \{\actname{wish\_flag}(i,i,b) \mid i \in \{0,1\}, b \in \mathbb{B}\}$, to capture that processes are allowed to remain in their non-critical section indefinitely. This was not considered in \autoref{mucalc:tutorial}, but it is part of the mutual exclusion problem \cite{dijkstra65,glabbeek2023modelling}.
The tool reports that the modified formula is satisfied. We can therefore conclude that Dekker's algorithm satisfies starvation freedom when assuming weak fairness of actions, as long as it is taken into account for each action which process is responsible for it.

Our other formulae can be used in similar ways.
An example of how to use the justness formula in mCRL2, including a method for encoding the concurrency relation, is given in \cite{bouwman2020off}.

\section{Discussion}\label{subsec:whyliveness}
In this section, we briefly reflect on the coverage of the properties we consider, and our choice in focusing on the modal $\mu$-calculus.

Firstly, we have exclusively addressed liveness properties in this paper thus far. As indicated previously, the problem we are considering primarily crops up for these properties.
This is because, as pointed out in \cite{glabbeek2023modelling}, when a completeness criterion is feasible, assuming the criterion holds true or not has no impact on whether a safety property is satisfied or not.
The reason is that for safety properties on paths, any path that violates the property must contain a finite prefix such that any extension of that prefix also violates the property \cite{lamport1985basic}.
Therefore, if a completeness criterion is feasible, then whenever a model contains incomplete paths that violate a safety property it also contains complete paths that violate the property.
All completeness criteria discussed in \autoref{sec:completeness} are feasible with respect to finite LTSs, and hence we do not need to consider patterns that capture safety properties.
For modal $\mu$-calculus formulae for the safety properties of PSP, without integrated completeness criteria, we refer to \cite{CADPRAFMC} and \cite{remenska2016bringing}.
For properties that are a combination of safety and liveness, the components can be turned into separate formulae and checked separately.

Readers may also wonder about alternative methods of representing properties under completeness criteria, such as using LTL.  
As indicated in \autoref{sec:motivation}, there are many contexts where we also want to consider non-linear properties, and hence the modal $\mu$-calculus is preferred.
Automatic translations from LTL to the modal $\mu$-calculus exist, but can be exponential in complexity \cite{cranen2011linear} and it is unclear at this time if this blow-up is avoided in this case.
Anecdotal evidence \cite{pol2008multi-core} suggests this is not the case for existing translations.
In \cite{glabbeek2023modelling} several completeness criteria are represented in LTL, but it is noted that this translation requires introducing new atomic propositions which hides the complexity of this translation.
The representation of hyperfairness in particular may be expensive, since atomic propositions for all reachable actions are required.
It is also unclear how to combine LTL-based translations effectively with symbolic model checking approaches.
For these reasons, a direct representation in the modal $\mu$-calculus is preferable.

\section{Conclusion}\label{sec:conclusion}

In this paper, we have presented formulae for liveness properties under several completeness criteria. 
As part of this, we defined a property template that generalises the liveness properties of PSP, which has been estimated to cover a majority of properties found in the literature \cite{dwyer1999patterns}. 
The completeness criteria covered are progress, justness, weak fairness, strong fairness, and hyperfairness, all defined with respect to actions and parameterised with a set of blocking actions.
The formulae have all been manually proven to be correct.

\altfull{For future work, one goal is to formalise the proofs in the appendices using a proof assistant.}{For future work, one goal is to formalise our manual proofs using a proof assistant.}
Another avenue for future work is extending our formulae to cover a wider range of completeness criteria and properties. We suggest some potential extensions here.

One of our contributions is the identification of a shared common structure underlying justness, weak fairness and weak hyperfairness: they are finitely realisable path predicates.
Our formula for such predicates can be adapted to arbitrary feasible finitely realisable path predicates.
While we do not have such a generic formula for other completeness criteria, our characterisation of $(\dpre,\ddis,\den)$-violating paths can be used as a basis to express the non-existence of complete paths violating many common properties for different notions of completeness as well, as we demonstrate with strong fairness and strong hyperfairness.
We are especially interested in extending our formulae to allow fairness over sets of actions, rather than individual actions, similar to the task-based definitions from \cite{glabbeek2019progress}.

In terms of properties, we can look at proposed extensions of PSP, such as those suggested in \cite{cobleigh2006user}. There is also the \pattern{constrained chain} behaviour, which is a modification of precedence chain and response chain given in \cite{dwyer1999patterns}. 
There are extensions of PSP to real-time \cite{bellini2009expressing,konrad2005real} and probabilistic \cite{grunske2008specification} contexts as well.
Finally, in \cite{bouwman2023supporting} the formula from \cite{bouwman2020off} that formed the basis of \autoref{mucalc:gen} is extended to also include state information.

There are therefore many potentially useful extensions of the formulae presented in this paper. However, the presented template formulae already cover many completeness criteria and liveness properties, making them  useful for model checking in practice.

\bibliography{bib2doi}

\begin{thebibliography}{10}

\bibitem{lamport1985basic}
Mack~W. Alford, Leslie Lamport, and Geoff~P. Mullery.
\newblock Basic concepts.
\newblock In Mack~W. Alford, Jean{-}Pierre Ansart, G{\"{u}}nter Hommel, Leslie Lamport, Barbara Liskov, Geoff~P. Mullery, and Fred~B. Schneider, editors, {\em Distributed Systems: Methods and Tools for Specification, An Advanced Course, April 3-12, 1984 and April 16-25, 1985, Munich, Germany}, volume 190 of {\em Lecture Notes in Computer Science}, pages 7--43. Springer, 1984.
\newblock \href {https://doi.org/10.1007/3-540-15216-4\_12} {\path{doi:10.1007/3-540-15216-4\_12}}.

\bibitem{alpern2985defining}
Bowen Alpern and Fred~B. Schneider.
\newblock Defining liveness.
\newblock {\em Inf. Process. Lett.}, 21(4):181--185, 1985.
\newblock \href {https://doi.org/10.1016/0020-0190(85)90056-0} {\path{doi:10.1016/0020-0190(85)90056-0}}.

\bibitem{apt1988appraising}
Krzysztof~R. Apt, Nissim Francez, and Shmuel Katz.
\newblock Appraising fairness in languages for distributed programming.
\newblock {\em Distributed Comput.}, 2(4):226--241, 1988.
\newblock \href {https://doi.org/10.1007/BF01872848} {\path{doi:10.1007/BF01872848}}.

\bibitem{attie1989fairness}
Paul~C. Attie, Nissim Francez, and Orna Grumberg.
\newblock Fairness and hyperfairness in multi-party interactions.
\newblock In Frances~E. Allen, editor, {\em Conference Record of the Seventeenth Annual {ACM} Symposium on Principles of Programming Languages, San Francisco, California, USA, January 1990}, pages 292--305. {ACM} Press, 1990.
\newblock \href {https://doi.org/10.1145/96709.96739} {\path{doi:10.1145/96709.96739}}.

\bibitem{bellini2009expressing}
Pierfrancesco Bellini, Paolo Nesi, and Davide Rogai.
\newblock Expressing and organizing real-time specification patterns via temporal logics.
\newblock {\em J. Syst. Softw.}, 82(2):183--196, 2009.
\newblock \href {https://doi.org/10.1016/j.jss.2008.06.041} {\path{doi:10.1016/j.jss.2008.06.041}}.

\bibitem{bouwman2023supporting}
Mark~S. Bouwman.
\newblock {\em Supporting Railway Standardisation with Formal Verification}.
\newblock {Phd Thesis 1 (Research TU/e / Graduation TU/e), Mathematics and Computer Science}, Eindhoven University of Technology, 2023.
\newblock \url{https://pure.tue.nl/ws/portalfiles/portal/307965423/20231023_Bouwman_hf.pdf}.

\bibitem{bouwman2020off}
Mark~S. Bouwman, Bas Luttik, and Tim A.~C. Willemse.
\newblock Off-the-shelf automated analysis of liveness properties for just paths.
\newblock {\em Acta Informatica}, 57(3-5):551--590, 2020.
\newblock \href {https://doi.org/10.1007/s00236-020-00371-w} {\path{doi:10.1007/s00236-020-00371-w}}.

\bibitem{bradfield2001modal}
Julian~C. Bradfield and Colin Stirling.
\newblock Modal logics and mu-calculi: an introduction.
\newblock In Jan~A. Bergstra, Alban Ponse, and Scott~A. Smolka, editors, {\em Handbook of Process Algebra}, pages 293--330. Elsevier Science, 2001.
\newblock \href {https://doi.org/10.1016/b978-044482830-9/50022-9} {\path{doi:10.1016/b978-044482830-9/50022-9}}.

\bibitem{bradfield200712}
Julian~C. Bradfield and Colin Stirling.
\newblock Modal mu-calculi.
\newblock In Patrick Blackburn, Johan {Van Benthem}, and Frank Wolter, editors, {\em Handbook of Modal Logic}, volume~3 of {\em Studies in logic and practical reasoning}, pages 721--756. Elsevier, 2007.
\newblock \href {https://doi.org/10.1016/s1570-2464(07)80015-2} {\path{doi:10.1016/s1570-2464(07)80015-2}}.

\bibitem{bradfield2018mu}
Julian~C. Bradfield and Igor Walukiewicz.
\newblock The mu-calculus and model checking.
\newblock In Edmund~M. Clarke, Thomas~A. Henzinger, Helmut Veith, and Roderick Bloem, editors, {\em Handbook of Model Checking}, pages 871--919. Springer, 2018.
\newblock \href {https://doi.org/10.1007/978-3-319-10575-8\_26} {\path{doi:10.1007/978-3-319-10575-8\_26}}.

\bibitem{mCRL2toolset}
Olav Bunte, Jan~Friso Groote, Jeroen J.~A. Keiren, Maurice Laveaux, Thomas Neele, Erik~P. de~Vink, Wieger Wesselink, Anton Wijs, and Tim A.~C. Willemse.
\newblock The {mCRL2} toolset for analysing concurrent systems.
\newblock In Tom{\'a}{\v{s}} Vojnar and Lijun Zhang, editors, {\em Tools and Algorithms for the Construction and Analysis of Systems - 25th International Conference, {TACAS} 2019, Held as Part of the European Joint Conferences on Theory and Practice of Software, {ETAPS} 2019, Prague, Czech Republic, April 6-11, 2019, Proceedings, Part {II}}, volume 11428 of {\em Lecture Notes in Computer Science}, pages 21--39. Springer, 2019.
\newblock \href {https://doi.org/10.1007/978-3-030-17465-1\_2} {\path{doi:10.1007/978-3-030-17465-1\_2}}.

\bibitem{clarke1995efficient}
Edmund~M. Clarke, Orna Grumberg, Kenneth~L. McMillan, and Xudong Zhao.
\newblock Efficient generation of counterexamples and witnesses in symbolic model checking.
\newblock In Bryan Preas, editor, {\em Proceedings of the 32st Conference on Design Automation, San Francisco, California, USA, Moscone Center, June 12-16, 1995}, pages 427--432. {ACM} Press, 1995.
\newblock \href {https://doi.org/10.1145/217474.217565} {\path{doi:10.1145/217474.217565}}.

\bibitem{cobleigh2006user}
Rachel~L. Cobleigh, George~S. Avrunin, and Lori~A. Clarke.
\newblock User guidance for creating precise and accessible property specifications.
\newblock In Michal Young and Premkumar~T. Devanbu, editors, {\em Proceedings of the 14th {ACM} {SIGSOFT} International Symposium on Foundations of Software Engineering, {FSE} 2006, Portland, Oregon, USA, November 5-11, 2006}, pages 208--218. {ACM}, 2006.
\newblock \href {https://doi.org/10.1145/1181775.1181801} {\path{doi:10.1145/1181775.1181801}}.

\bibitem{cranen2011linear}
Sjoerd Cranen, Jan~Friso Groote, and Michel~A. Reniers.
\newblock A linear translation from {CTL$^\star$} to the first-order modal $\mu$-calculus.
\newblock {\em Theor. Comput. Sci.}, 412(28):3129--3139, 2011.
\newblock \href {https://doi.org/10.1016/j.tcs.2011.02.034} {\path{doi:10.1016/j.tcs.2011.02.034}}.

\bibitem{dijkstra1962over}
Edsger~W Dijkstra.
\newblock Over de sequentialiteit van procesbeschrijvingen ({EWD}-35). {EW} dijkstra archive.
\newblock {\em Center for American History, University of Texas at Austin}, 1962.
\newblock URL: \url{https://www.cs.utexas.edu/~EWD/ewd00xx/EWD35.PDF}.

\bibitem{dijkstra65}
Edsger~W. Dijkstra.
\newblock Solution of a problem in concurrent programming control.
\newblock {\em Commun. {ACM}}, 8(9):569, 1965.
\newblock \href {https://doi.org/10.1145/365559.365617} {\path{doi:10.1145/365559.365617}}.

\bibitem{dwyer1999patterns}
Matthew~B. Dwyer, George~S. Avrunin, and James~C. Corbett.
\newblock Patterns in property specifications for finite-state verification.
\newblock In Barry~W. Boehm, David Garlan, and Jeff Kramer, editors, {\em Proceedings of the 1999 International Conference on Software Engineering, ICSE' 99, Los Angeles, CA, USA, May 16-22, 1999}, pages 411--420. {ACM}, 1999.
\newblock \href {https://doi.org/10.1145/302405.302672} {\path{doi:10.1145/302405.302672}}.

\bibitem{emerson1982using}
E.~Allen Emerson and Edmund~M. Clarke.
\newblock Using branching time temporal logic to synthesize synchronization skeletons.
\newblock {\em Sci. Comput. Program.}, 2(3):241--266, 1982.
\newblock \href {https://doi.org/10.1016/0167-6423(83)90017-5} {\path{doi:10.1016/0167-6423(83)90017-5}}.

\bibitem{fisher1979pdl}
Michael~J. Fischer and Richard~E. Ladner.
\newblock Propositional dynamic logic of regular programs.
\newblock {\em J. Comput. Syst. Sci.}, 18(2):194--211, 1979.
\newblock \href {https://doi.org/10.1016/0022-0000(79)90046-1} {\path{doi:10.1016/0022-0000(79)90046-1}}.

\bibitem{garavel2013cadp}
Hubert Garavel, Fr{\'{e}}d{\'{e}}ric Lang, Radu Mateescu, and Wendelin Serwe.
\newblock {CADP} 2011: a toolbox for the construction and analysis of distributed processes.
\newblock {\em Int. J. Softw. Tools Technol. Transf.}, 15(2):89--107, 2013.
\newblock \href {https://doi.org/10.1007/s10009-012-0244-z} {\path{doi:10.1007/s10009-012-0244-z}}.

\bibitem{glabbeek2019justness}
Rob J.~van Glabbeek.
\newblock Justness - {A} completeness criterion for capturing liveness properties (extended abstract).
\newblock In Miko{\l}aj oja{\'{n}}czyk and Alex Simpson, editors, {\em Foundations of Software Science and Computation Structures - 22nd International Conference, {FOSSACS} 2019, Held as Part of the European Joint Conferences on Theory and Practice of Software, {ETAPS} 2019, Prague, Czech Republic, April 6-11, 2019, Proceedings}, volume 11425 of {\em Lecture Notes in Computer Science}, pages 505--522. Springer, 2019.
\newblock \href {https://doi.org/10.1007/978-3-030-17127-8\_29} {\path{doi:10.1007/978-3-030-17127-8\_29}}.

\bibitem{glabbeek2020reactive}
Rob J.~van Glabbeek.
\newblock Reactive temporal logic.
\newblock In Ornela Dardha and Jurriaan Rot, editors, {\em Proceedings Combined 27th International Workshop on Expressiveness in Concurrency and 17th Workshop on Structural Operational Semantics, {EXPRESS/SOS} 2020, and 17th Workshop on Structural Operational SemanticsOnline, 31 August 2020}, volume 322 of {\em {EPTCS}}, pages 51--68. Open Publishing Association, 2020.
\newblock \href {https://doi.org/10.4204/EPTCS.322.6} {\path{doi:10.4204/EPTCS.322.6}}.

\bibitem{glabbeek2023modelling}
Rob J.~van Glabbeek.
\newblock Modelling mutual exclusion in a process algebra with time-outs.
\newblock {\em Inf. Comput.}, 294:105079, 2023.
\newblock \href {https://doi.org/10.1016/j.ic.2023.105079} {\path{doi:10.1016/j.ic.2023.105079}}.

\bibitem{glabbeek2015ccs}
Rob J.~van Glabbeek and Peter H{\"o}fner.
\newblock {CCS}: It’s not fair! fair schedulers cannot be implemented in {CCS}-like languages even under progress and certain fairness assumptions.
\newblock {\em Acta Informatica}, 52(2-3):175--205, 2015.
\newblock \href {https://doi.org/10.1007/s00236-015-0221-6} {\path{doi:10.1007/s00236-015-0221-6}}.

\bibitem{glabbeek2019progress}
Rob J.~van Glabbeek and Peter H{\"o}fner.
\newblock Progress, {J}ustness, and {F}airness.
\newblock {\em {ACM} Comput. Surv.}, 52(4):69:1--69:38, 2019.
\newblock \href {https://doi.org/10.1145/3329125} {\path{doi:10.1145/3329125}}.

\bibitem{groote2021tutorial}
Jan~Friso Groote and Jeroen J.~A. Keiren.
\newblock Tutorial: designing distributed software in {mCRL2}.
\newblock In Kirstin Peters and Tim A.~C. Willemse, editors, {\em Formal Techniques for Distributed Objects, Components, and Systems - 41st {IFIP} {WG} 6.1 International Conference, {FORTE} 2021, Held as Part of the 16th International Federated Conference on Distributed Computing Techniques, DisCoTec 2021, Valletta, Malta, June 14-18, 2021, Proceedings}, volume 12719 of {\em Lecture Notes in Computer Science}, pages 226--243. Springer, 2021.
\newblock \href {https://doi.org/10.1007/978-3-030-78089-0\_15} {\path{doi:10.1007/978-3-030-78089-0\_15}}.

\bibitem{mCRL2language}
Jan~Friso Groote and Mohammad~Reza Mousavi.
\newblock {\em {Modeling and Analysis of Communicating Systems}}.
\newblock {MIT} Press, 08 2014.
\newblock URL: \url{https://mitpress.mit.edu/books/modeling-and-analysis-communicating-systems}.

\bibitem{grunske2008specification}
Lars Grunske.
\newblock Specification patterns for probabilistic quality properties.
\newblock In Wilhelm Sch{\"{a}}fer, Matthew~B. Dwyer, and Volker Gruhn, editors, {\em 30th International Conference on Software Engineering {(ICSE} 2008), Leipzig, Germany, May 10-18, 2008}, pages 31--40. {ACM}, 2008.
\newblock \href {https://doi.org/10.1145/1368088.1368094} {\path{doi:10.1145/1368088.1368094}}.

\bibitem{konrad2005real}
Sascha Konrad and Betty H.~C. Cheng.
\newblock Real-time specification patterns.
\newblock In Gruia{-}Catalin Roman, William~G. Griswold, and Bashar Nuseibeh, editors, {\em 27th International Conference on Software Engineering {(ICSE} 2005), 15-21 May 2005, St. Louis, Missouri, {USA}}, pages 372--381. {ACM}, 2005.
\newblock \href {https://doi.org/10.1145/1062455.1062526} {\path{doi:10.1145/1062455.1062526}}.

\bibitem{kozen1983results}
Dexter Kozen.
\newblock Results on the propositional $\mu$-calculus.
\newblock {\em Theor. Comput. Sci.}, 27(3):333--354, 1983.
\newblock \href {https://doi.org/10.1016/0304-3975(82)90125-6} {\path{doi:10.1016/0304-3975(82)90125-6}}.

\bibitem{lamport2000fairness}
Leslie Lamport.
\newblock Fairness and hyperfairness.
\newblock {\em Distributed Comput.}, 13(4):239--245, 2000.
\newblock \href {https://doi.org/10.1007/PL00008921} {\path{doi:10.1007/PL00008921}}.

\bibitem{CADPRAFMC}
Radu Mateescu.
\newblock {Property Pattern Mappings for RAFMC}, 2019.
\newblock {Available at: \url{https://cadp.inria.fr/resources/evaluator/rafmc.html} (Accessed: 26 January 2024)}.

\bibitem{pnueli1977temporal}
Amir Pnueli.
\newblock The temporal logic of programs.
\newblock In {\em 18th Annual Symposium on Foundations of Computer Science, Providence, Rhode Island, USA, 31 October - 1 November 1977}, pages 46--57. {IEEE} Computer Society, 1977.
\newblock \href {https://doi.org/10.1109/SFCS.1977.32} {\path{doi:10.1109/SFCS.1977.32}}.

\bibitem{pol2008multi-core}
Jaco van~de Pol and Michael Weber.
\newblock A multi-core solver for parity games.
\newblock {\em Electronic Notes in Theoretical Computer Science}, 220(2):19--34, 2008.
\newblock Proceedings of the 7th International Workshop on Parallel and Distributed Methods in verifiCation (PDMC 2008).
\newblock \href {https://doi.org/10.1016/j.entcs.2008.11.011} {\path{doi:10.1016/j.entcs.2008.11.011}}.

\bibitem{remenska2016bringing}
Daniela Remenska.
\newblock {\em Bringing Model Checking Closer To Practical Software Engineering}.
\newblock PhD thesis, Vrije U., Amsterdam, 2016.
\newblock {PhD Thesis}, available at: \url{https://hdl.handle.net/1871/53958}.

\bibitem{spronck2023fairness}
Myrthe S.~C. Spronck.
\newblock Fairness assumptions in the modal $\mu$-calculus, 2023.
\newblock {Master's} thesis, Eindhoven University of Technology, available at \url{https://research.tue.nl/en/studentTheses/fairness-assumptions-in-the-modal-%C2%B5-calculus}.

\bibitem{stomp1989mu}
Frank~A. Stomp, Willem-Paul de~Roever, and Rob~T. Gerth.
\newblock The $\mu$-calculus as an assertion-language for fairness arguments.
\newblock {\em Inf. Comput.}, 82(3):278--322, 1989.
\newblock \href {https://doi.org/10.1016/0890-5401(89)90004-7} {\path{doi:10.1016/0890-5401(89)90004-7}}.

\end{thebibliography}

\altfull{

\newpage
\appendix

\section{Property Patterns}\label{app:psp}
Here we recall the behaviours and scopes presented in \cite{dwyer1999patterns}.
The original presentation is not restricted to a particular logic and the patterns allow behaviour and scopes to be defined based on both states and actions.
We give the definitions specifically with respect to occurrences of actions, since those are the properties we consider in this paper.
We use $S_a$ (``after''), $S_b$ (``before''), $S_q$ (``query'') and $S_r$ (``required''/``response'') as placeholder names for property-specific sets of actions. We use $k$ for an arbitrary natural number.

The following behaviours are given:
\begin{itemize}
    \item \textbf{Absence}: no action in $S_r$ may occur.
    \item \textbf{Existence}: some action in $S_r$ must occur.
    \begin{itemize}
        \item \textbf{Existence at least/at most/exactly}: there must be at least/at most/exactly $k$ occurrences of actions in $S_r$. The \pattern{existence} pattern is an instantiation of \pattern{existence at least} with $k = 1$.
    \end{itemize}
    \item \textbf{Universality}: only actions in $S_r$ occur.
    \item \textbf{Precedence}: an occurrence of an action in $S_r$ must always be preceded by an occurrence of an action in $S_q$.
    \begin{itemize}
        \item \textbf{Chain precedence}: if actions from the sets $S_{r_0}$, $S_{r_1}$, $\ldots$, $S_{r_n}$ occur in that order (potentially with other actions in-between), then they must have been preceded by occurrences of actions from the sets $S_{q_0}$, $S_{q_1}$, $\ldots$, $S_{q_m}$, in that order.
    \end{itemize}
    \item \textbf{Response}: an occurrence of an action in $S_q$ must be followed by the occurrence of an action from $S_r$.
    \begin{itemize}
        \item \textbf{Chain response}: if actions from the sets $S_{q_0}$, $S_{q_1}$, $\ldots$, $S_{q_n}$ occur in that order (potentially with other actions in between), then they must be followed by occurrences of actions from the sets $S_{r_0}$, $S_{r_1}$, $\ldots$, $S_{r_m}$, in that order.
    \end{itemize}
\end{itemize}
All of the behaviours only need to hold within the chosen scope.
The following scopes are given:
\begin{itemize}
    \item \textbf{Global}: the full path.
    \item \textbf{Before}: the prefix of the path before the first occurrence of an action in $S_b$. If no action in $S_b$ occurs on the path, then the behaviour does not need to be satisfied anywhere.
    \begin{itemize}
        \item \textbf{Until}:  same as \pattern{before}, except that if no action in $S_b$ occurs, then the behaviour needs to hold on the full path.
    \end{itemize}
    \item \textbf{After}: the suffix of the path after the first occurrence of an action in $S_a$. If no action in $S_a$ occurs on the path, then the behaviour does not need to be satisfied anywhere.
    \item \textbf{Between}: every subpath of the path that starts after an occurrence of an action in $S_a$ and ends before the first following occurrence of an action in $S_b$. If there is an occurrence of an action in $S_a$ that is not eventually followed by an action in $S_b$, the behaviour does not need to be satisfied after that $S_a$. This combines \pattern{after} and \pattern{before}, but unlike the default \pattern{after} scope considers any occurrence of $S_a$, not merely the first.
    \begin{itemize}
        \item \textbf{After-until}: same as \pattern{between}, except that if there is an occurrence of an action in $S_a$ that is not eventually followed by an action in $S_b$, the behaviour still needs to be satisfied after that occurrence of $S_a$. This combines \pattern{after} and \pattern{until}.
    \end{itemize}
\end{itemize}
The \pattern{until} scope does not appear in \cite{dwyer1999patterns}, but \pattern{after-until} does.
We include the \pattern{until} scope from \cite{remenska2016bringing}, there called \pattern{before-variant}, because it can be seen as a simpler form of \pattern{after-until}.

We only consider liveness properties, so we must ask which combinations of behaviour and scope result in liveness properties.
To make this judgement, we need a formal definition of what makes a property a safety or liveness property.
For the purposes of this paper, since all the properties we consider are defined on occurrences of actions, we can use the following definition of a property:
\begin{definition}
    A \emph{property} is a set of sequences of actions. A path $\pi$ \emph{satisfies} the property if its sequence of actions is in the set, otherwise it \emph{violates} the property.
\end{definition}
We adapt the formal definitions of safety and liveness properties from \cite{lamport1985basic} and \cite{alpern2985defining} respectively to this definition of properties.
\begin{definition}\label{def:safety}
    A property $P$ is a \emph{safety} property if, and only if, every infinite sequence of actions not in $P$ has a finite prefix that is not in $P$.
\end{definition}
The consequence of this is that an infinite path that violates a safety property always has a finite prefix that violates it as well.
\begin{definition}\label{def:liveness}
    A property $P$ is a \emph{liveness} property if, and only if, for every finite sequence of actions $\mathit{f}$ there exists some infinite sequence of actions $\mathit{f}'$ such that $\mathit{f}\mathit{f}'$ is in $P$.
\end{definition}
In terms of paths, this means that for every finite path $\pi$ that violates a liveness property, there exists an infinite path $\pi'$ of which $\pi$ is a prefix that satisfies the property. Note that it is not required for an LTS that admits $\pi$ to also admit $\pi'$, only that such an extension could be made.
We now discuss which patterns form liveness properties.

First, we note that the \pattern{before} and \pattern{between} scopes will turn every behaviour into a safety property: whenever an action in $S_b$ occurs the behaviour should be satisfied before that occurrence, hence every path that violates the property will have a finite prefix, ending with the first occurrence of an action in $S_b$, that also violates the property.
The \pattern{global}, \pattern{until}, \pattern{after} and \pattern{after-until} scopes remain relevant.
Under these four scopes, the behaviours \pattern{absence}, \pattern{existence at most}, \pattern{universality}, \pattern{precedence} and \pattern{chain precedence} will always result in safety properties as well.
For each of these behaviours, some actions may not occur under certain circumstances (be it at all, after there have already been a number of occurrence of those actions, or when some other actions have not yet occurred), therefore every violating path has a finite prefix, ending with the occurrence of such an action, that also violates the path.
This leaves us with \pattern{existence}, \pattern{existence at least}, \pattern{existence exactly}, \pattern{response} and \pattern{chain response}.

We further drop \pattern{existence exactly} since it is merely a conjunction of \pattern{existence at least} and \pattern{existence at most}.
Both parts of the pattern can be expressed separately, so a separate formula for \pattern{existence exactly} is superfluous.

We could apply a similar argument to the \pattern{until} scope: it is merely a combination of \pattern{global} and \pattern{before}. 
Saying that an action in $S_r$ has to occur until $S_b$ (\pattern{existence until}), for instance, is the same as saying that an action in $S_r$ has to occur at all (\pattern{global existence}) and that if there is an occurrence of $S_b$, there must be an occurrence of $S_r$ before it (\pattern{existence before}).
This extends to \pattern{after-until} as well, although it requires a bit more care than simply combining \pattern{after} and \pattern{between}, since \pattern{after} always applies to the first occurrence of an action in $S_a$, whereas \pattern{after-until} refers to every occurrence. 
This could be achieved with minor modifications to the patterns.
However, it turns out we can relatively easily incorporate the \pattern{until} and \pattern{after-until} scopes into our formulae, so we include them for convenience.

\section{Representing \pattern{Chain Response}}\label{app:chain-resp}
We here illustrate how the \pattern{chain response} behaviour can be represented using our template formulae by combining several \pattern{response} formulae.

Consider, for example, sequences of two sets each: if an occurrence of $S_{q_0}$ is eventually followed by an occurrence of $S_{q_1}$, then there must subsequently be an occurrence of $S_{r_0}$ followed by $S_{r_1}$. 
There are two possible violating paths here: either $S_{q_0}, S_{q_1}$ is not followed by $S_{r_0}$, or $S_{q_0},S_{q_1},S_{r_0}$ is not followed by $S_{r_1}$.
These two violations cannot be slotted directly into the form of a single $(\dpre,\ddis,\den)$-violating path. 
Instead, we have two different $(\dpre,\ddis,\den)$-violating paths.
In general, if we have chain response with $S_{q_0}$ to $S_{q_n}$ and $S_{r_0}$ to $S_{r_m}$ then we get $m$ different violating paths. Specifically, for all $0 \leq i \leq m$, we get a violating path that consists of a sequence $S_{q_0}$, $\ldots$, $S_{q_n}$, $S_{r_0}$, $\ldots$, $S_{r_{i-1}}$ that may not be followed by an occurrence of $S_{r_i}$. For convenience, we write such violating paths as a sequence $S_0$, $S_1$, $\ldots$, $S_n$ that may not be followed by $S_{n+1}$. 

Such violating paths can be expressed through $\dpreb = \clos{\comp{\den}} \co S_0 \co \clos{\comp{\den \cup S_1}} \co S_1 \co \clos{\comp{\den \cup S_2}} \co S_2 \ldots \co \clos{\comp{\den\cup S_n}} \co S_n$, and $\ddis = S_{n+1}$. Each violating path must be given its own formula, where \autoref{tab:scopes} is still used for the scope, and all resulting formulae placed in conjunction. This way \pattern{chain-response} can be represented.

\begin{example}
    Say we want to express \pattern{chain response} with the scope \pattern{after-until} under WFA, and we take the chain that an occurrence of an action in $S_{q_0}$, if followed by an action in $S_{q_1}$, needs to be followed by an action in $S_{r_0}$ and then by an action in $S_{r_1}$.
    The violating paths are occurrences of an actions in $S_{q_0}, S_{q_1}$ not followed by an action $S_{r_0}$, and occurrences of actions in $S_{q_0}, S_{q_1}, S_{r_0}$ with no subsequent occurrence of an action in $S_{r_1}$.
    Both must be after the first occurrence of an action in $S_a$, and before the next occurrence of an action in $S_b$.
    The formula we need is then:
    \begin{align*}
        &\neg(\diam{\clos{\allact} \co S_a \co \clos{\comp{S_b}} \co S_{q_0} \co \clos{\comp{S_b \cup S_{q_1}}} \co S_{q_1}}\\
        &\quad\nu X.(\bigwedge_{a \in \nonblock}(\diam{a}\tp \imps\diam{\clos{\comp{S_{r_0}}}}(\diam{S_b}\tp \lor (\boxm{a}\fp \land X) \lor \diam{a \setminus S_{r_0}}X))))\\
        & \land\\
        &\neg(\diam{\clos{\allact} \co S_a \co \clos{\comp{S_b}} \co S_{q_0} \co \clos{\comp{S_b \cup S_{q_1}}} \co S_{q_1} \co \clos{\comp{S_b \cup S_{r_0}}} \co S_{r_0}}\\
        & \quad\nu X.(\bigwedge_{a \in \nonblock}(\diam{a}\tp \imps\diam{\clos{\comp{S_{r_1}}}}(\diam{S_b}\tp \lor (\boxm{a}\fp \land X) \lor \diam{a \setminus S_{r_1}}X))))\\
    \end{align*}
\end{example}

\section{Proofs of Feasibility}\label{app:props}
In \autoref{subsec:gen} we claimed WFA, WHFA and JA are feasible with respect to finite LTSs. 
In the proof of the SFA and SHFA formulae, we will need feasibility of SFA and SHFA as well.
In this appendix, we give those proofs.

All our proofs assume a fixed LTSC $M = (\states,\initstate,\actionset,\transrel,\conc)$, although the $\conc$ is only relevant for JA.
We also refer to an arbitrary environment $\env$, and set of blocking actions $\block \subseteq \actionset$.
When we refer to an arbitrary state or transition in a path in our proofs, it should be understood that we are referring to specific occurrences of those states and transitions unless explicitly stated otherwise.

Recall \autoref{def:feas}:
\deffeas*

\begin{proposition}\label{prop:wfa-feasible}
    $\block$-weak fairness of actions is feasible.
\end{proposition}
\begin{proof}
    It is proven in \cite[Theorem 6.1]{glabbeek2019progress} that if only countably many actions are enabled in each state of a transition system, then weak fairness of actions with $\block = \emptyset$ is feasible. 
    We have assumed a finite set $\actionset$, hence this theorem applies in our case.
    This means that every finite path can be extended to a path that is $\emptyset$-weakly fair.
    A path that satisfies $\emptyset$-WFA also satisfies $\block$-WFA for arbitrary $\block$, since $\emptyset$-WFA requires all actions in $\actionset$ to occur in suffixes that they are perpetually enabled in, and $\actionset$ is a superset of $\nonblock$. 
    We conclude that $\block$-weak fairness of actions is feasible.
\end{proof}

\begin{proposition}\label{prop:sfa-feasible}
    $\block$-strong fairness of actions is feasible.
\end{proposition}
\begin{proof}
     It is proven in \cite[Theorem 6.1]{glabbeek2019progress} that if only countably many actions are enabled in each state of a transition system, then strong fairness of actions with $\block = \emptyset$ is feasible. We have assumed a finite set $\actionset$, hence this theorem applies.
     Similarly to WFA, as argued in \autoref{prop:wfa-feasible}, $\emptyset$-SFA implies $\block$-SFA for arbitrary $\block$ because $\block$-SFA requires only actions in $\nonblock$  to occur when they are relentlessly enabled.
     We conclude $\block$-strong fairness of actions is feasible.
\end{proof}

For the two forms of hyperfairness, we first prove a supporting lemma.
\begin{lemma}\label{lem:whfa-more-feasible}
    Every finite path $\pi$ can be extended to a path $\pi'$ that satisfies weak $\block$-hyperfairness of actions, such that all occurrences of blocking actions in $\pi'$ are part of $\pi$.
\end{lemma}
\begin{proof}
    Let $\pi$ be an arbitrary finite path.
    We prove that $\pi$ can be extended to path $\pi'$ that satisfies weak $\block$-hyperfairness of actions, such that there are no occurrences of blocking actions in the extension. We do this through construction of the path $\pi'$. 
    We will construct $\pi'$ in steps. Let $\pi_i$ with $i \geq 0$ be the path constructed in the $i$'th iteration, with $\pi_0 = \pi$. 
    Let $s_i$ be the last state of $\pi_i$.
    For this construction, we use a queue $Q$ containing non-blocking actions. At the start of the construction, $Q$ is initialised with exactly one copy all non-blocking actions $\block$-reachable from $s_0$, the final state of $\pi$, in some arbitrary order.
    The construction has the following invariants: $Q$ contains exactly one copy of every non-blocking action $\block$-reachable in the final state of the path constructed so far. It may contain zero or one copies of non-blocking actions not $\block$-reachable from this state. It contains no blocking actions. Additionally, the only occurrences of blocking actions in the path constructed thus far are in $\pi$.
    
    At each step $i > 0$, we do the following: first, we determine if $Q$ is empty or not. If it is empty, we take $\pi_i = \pi_{i-1}$ and the construction terminates. If $Q$ is not empty, we pop the head $a$ from $Q$. If $a$ is not $\block$-reachable from $s_{i-1}$, then we let $\pi_{i} = \pi_{i-1}$ and go to step $i + 1$. The invariants are maintained because we only removed $a$ from $Q$, and $a$ was not $\block$-reachable from $s_{i-1} = s_i$. 
    If $a$ is $\block$-reachable from $s_{i-1}$, then there exists some path $\pi_i'$ consisting of only non-blocking actions starting in $s_{i-1}$ and ending in a state $s_{i-1}'$ such that some transition $t_a$ with $\action{t_a} = a$ is enabled in $s_{i-1}'$. Let $\pi_{i} = \pi_{i-1} \co \pi_{i-1}' t_a \target{t_a}$ and append $a$ back to the end of $Q$. Then continue to step $i + 1$. 
    The invariants are maintained in this case as well. This is because every action $\block$-reachable in $s_{i}$ must also have been $\block$-reachable from $s_{i-1}$, since $s_{i}$ is $\block$-reachable from $s_{i-1}$. Since $Q$ at the start of this step contains the same actions as at the end of this step, and by the invariant it contained all actions at the start that are $\block$-reachable from $s_{i-1}$, it also contains all actions that are $\block$-reachable from $s_{i}$ at the end of the step. Finally, the segment we added did not contain any blocking actions, and $a$ itself is non-blocking because all actions in $Q$ are non-blocking.

    There are two potential outcomes to this construction: either $Q$ becomes empty and the construction terminates, or $Q$ never becomes empty and the construction continues infinitely. We prove that in either case, the path $\pi'$ that is ultimately constructed is weakly $\block$-hyperfair of actions and does not contain occurrences of blocking actions beyond those already present in $\pi$.
    \begin{itemize}
        \item If $Q$ becomes empty and the construction terminates, then the final path $\pi'$ is $\pi_i$ for the $i$ on which $Q$ was determined to be empty. The final state of $\pi'$, $s'$ is then a state in which no non-blocking actions are $\block$-reachable. Hence, there are no non-blocking actions perpetually $\block$-reachable on any suffix of $\pi'$ and so $\pi'$ is trivially $\block$-WHFA. 
        \item If $Q$ never becomes empty then the construction continues forever. Let $\pi'$ be the infinite path $\pi_{\infty}$. Let $\pi''$ be an arbitrary suffix of $\pi'$, and let $a$ be an arbitrary action in $\nonblock$ that is perpetually $\block$-reachable on $\pi''$. We prove $a$ occurs in $\pi''$. Consider that if $a$ is perpetually $\block$-reachable on $\pi''$, then it is $\block$-reachable in every state of $\pi''$.
        Consider also that, since $\pi''$ is a suffix of the infinite path $\pi'$, $\pi''$ is also infinite. 
        In our construction, we add only a finite number of steps to the path in every iteration. Therefore, $\pi''$ was created as a part of $\pi'$ over infinitely many iterations.
        Since $a$ is enabled in every state of $\pi''$, by the invariants $a$ must be in $Q$ at the start of all iterations of the construction that contributed to $\pi''$, with possible exception of the first.
        Since $Q$ is a queue and $\actionset$ is finite, $a$ will be at the head of the queue during the construction of $\pi''$ infinitely many times. 
        Whenever $a$ was at the head of the queue during the construction, a finite number of steps were added to the path that ended with a transion labelled with $a$. Hence, $a$ occurs in $\pi''$, and so $\pi'$ is $\block$-WHFA.
    \end{itemize}
    In both cases, that no new occurrences of blocking actions are added to the path comes directly from the invariants.
\end{proof}

\begin{proposition}\label{prop:whfa-feasible}
    Weak $\block$-hyperfairness of actions is feasible.
\end{proposition}
\begin{proof}
    This follows from \autoref{lem:whfa-more-feasible}, which is a stronger property.
\end{proof}

\begin{proposition}\label{prop:shfa-feasible}
    Strong $\block$-hyperfairness of actions in feasible.
\end{proposition}
\begin{proof}
    We prove that every finite path $\pi$ can be extended to a path $\pi'$ that satisfies strong $\block$-hyperfairness of actions. 
    Let $\pi$ be an arbitrary finite path, then by \autoref{lem:whfa-more-feasible} we know that there exists a path $\pi'$ that extends $\pi$ and satisfies weak $\block$-hyperfairness of actions, and has no occurrences of blocking actions save those already present in $\pi$.
    We will use $\pi'$ to witness that there exists a strongly $\block$-hyperfair extension of $\pi$, by proving $\pi'$ satisfies strong $\block$-hyperfairness as well as weak $\block$-hyperfairness.
    
    Towards a contradiction, assume that $\pi'$ does not satisfy strong $\block$-hyperfairness of actions. Then $\pi'$ must have a suffix $\pi''$ such that there is an action $a \in \nonblock$ that is relentlessly $\block$-reachable in $\pi''$ and yet does not occur in $\pi''$.
    If $a$ is relentlessly $\block$-reachable in $\pi''$, it is also relentlessly $\block$-reachable in every suffix of $\pi''$.
    Let $\pi'''$ be a suffix of $\pi''$ such that $\pi'''$ does not contain any occurrences of blocking actions. That such a suffix exists follows from $\pi'$ only having occurrences of blocking actions in the finite prefix $\pi$.
    We now have the path $\pi'''$ on which $a$ is relentlessly $\block$-reachable and that does not contain occurrences of blocking actions.
    Let $s$ be an arbitrary state on $\pi'''$. Since $a$ is relentlessly $\block$-reachable, there must be a state $s'$ on $\pi'''$ past $s$ such that $a$ is $\block$-reachable from $s'$.
    And since there are no occurrences of blocking actions on $\pi'''$, $a$ is also $\block$-reachable from $s$. Hence, $a$ is $\block$-reachable from every state of $\pi'''$ and is therefore perpetually $\block$-reachable on $\pi'''$.    
    
    We constructed $\pi'''$ as a suffix of $\pi''$ which is a suffix of $\pi'$, so $\pi'''$ is a suffix of $\pi'$ as well.
    We know that $\pi'$ satisfies weak $\block$-hyperfairness, so since $a$ is perpetually $\block$-reachable on $\pi'''$, a suffix of $\pi'$, $a$ also occurs in $\pi'''$.
    Since $\pi'''$ is a suffix of $\pi''$, we know that $a$ occurs on $\pi'$.
    However, we assumed previously that $a$ does not occur on $\pi''$.
    We have reached a contradiction and therefore conclude that $\pi'$ satisfies strong $\block$-hyperfairness as well as weak $\block$-hyperfairness.
\end{proof}

\begin{proposition}\label{prop:ja-feasible}
    $\block$-justness of actions is feasible.
\end{proposition}
\begin{proof}
    Let $\pi$ be an arbitrary finite path.
    We prove $\pi$ can be extended to a path $\pi'$ satisfying $\block$-justness of actions. We do this through construction of such a path $\pi'$. We do this in steps, where $\pi_i$ with $i \geq 0$ represents the path constructed in step $i$. Let $\pi_0 = \pi$. Let $s_i$ be the last state of $\pi_i$ for all $i \geq 0$. 
    For this construction we use a queue $Q$.
    
    The initial contents of $Q$ are determined by $\pi_0$: it contains exactly one copy of every non-blocking action that is enabled in some state of $\pi_0$ but has not been subsequently eliminated. The order of these actions is arbitrary.
    The construction has the following invariant: $Q$ contains exactly one copy of every non-blocking action that is enabled in some state of the path constructed so far, but has not been subsequently eliminated. Trivially, this invariant holds at initialisation.
    
    The construction proceeds as follows:
    in step $i$, with $i > 0$, we construct $\pi_i$ from $\pi_{i-1}$ using $Q$. 
    At this point, $Q$ contains exactly one copy of every non-blocking action that was enabled in some state of $\pi_{i-1}$ but has not subsequently been eliminated. 
    If $Q$ is empty, let $\pi_i = \pi_{i-1}$ and the construction terminates.
    Otherwise, we pop the head of $Q$, let this action be $a$. The action $a$ must have been enabled in some state $s_a$ of $\pi_{i-1}$ such that the subpath $\pi_a$ of $\pi_{i-1}$ from $s_a$ to $s_{i-1}$ does not contain an occurrence of an action that eliminates $a$. 
    By the second property of concurrency relations on actions, $a$ must still be enabled in $s_{i-1}$. Let $t_a$ be a transition enabled in $s_{i-1}$ with $\action{t_a} = a$, let $\pi_i = \pi_{i-1} t_a \target{t_a}$. We modify $Q$ in two steps: firstly, every action $b$ that is in $Q$ such that $b \nconc a$ is removed from $Q$. Secondly, every non-blocking action that is enabled in $\target{t_a} = s_i$ that is not yet in $Q$ gets appended to $Q$ in some arbitrary order. 
    At this point, the invariant is again satisfied: by removing all actions that are eliminated by $a$ from $Q$, we ensure that $Q$ no longer contains those actions that were not eliminated in $\pi_{i-1}$  but are eliminated in $\pi_i$. By afterwards adding those actions that are enabled in $s_i$, we include those actions that are newly enabled without being eliminated yet.
    We proceed to the next iteration.

    This construction either terminates after finitely many steps, or continues forever, the latter case results in an infinite path. We show that in either case, the constructed path $\pi'$ satisfies $\block$-justness of actions.
    \begin{itemize}
        \item If the construction terminates during step $i$, then $\pi' = \pi_i$. By the invariant, $Q$ contains exactly those non-blocking actions that are enabled in $\pi'$ without being subsequently eliminated, and $Q$ must be empty because the construction terminated. Let $s'$ be the final state of $\pi'$. If there are non-blocking actions enabled in $s'$, then those actions are enabled in $\pi'$ without being subsequently eliminated, since there are no further transitions in $\pi'$. 
        Since no such actions can exist, we know $s'$ is a $\block$-locked state.
        If there were a non-blocking action enabled on $\pi'$ that has not subsequently been eliminated through the occurrence of an interfering action, then by the second property of concurrency relations on actions that non-blocking action should still be enabled on $s'$. 
        Since $s'$ is a $\block$-locked state, this is impossible and hence $\pi'$ is $\block$-JA.
        \item If the construction never terminates, then we construct an infinite path $\pi_{\infty} = \pi'$.  Let $s$ be some arbitrary state of $\pi'$ and let $a$ be an arbitrary non-blocking action that is enabled in $s$. We prove that $a$ is eliminated in the suffix $\pi_a'$ of $\pi'$ starting in $s$. 
        Towards a contradiction, assume that $a$ is not eliminated in $\pi'_a$.
        Let $i > 0$ be the first iteration of the construction such that $s_{i-1}$ is in $\pi_a'$. 
        Since $s_{i-1}$ is in $\pi_a'$, it either is $s$ or comes after $s$, and since $a$ is not eliminated in $\pi_a'$ it must be the case, by the second property of concurrency relations on actions, that $a$ is enabled in $s_{i-1}$.
        Hence, by the invariant, $Q$ must have contained $a$ at the start of iteration $i$. 
        Since $Q$ contains at most one copy of every non-blocking action and $\actionset$ is finite, there are finitely many actions before $a$ in the queue. Every iteration, at least one action gets removed from $Q$ and new actions get appended. Hence, $a$ is either removed early or eventually becomes the head of the queue. If $a$ is removed early, this is because an action occurred that eliminates $a$, hence $a$ is eliminated in $\pi_a'$. If $a$ becomes the head of the queue, then we ensure $a$ itself occurs. By the first property of concurrency relations on actions, $a$ eliminates itself. In this case too, $a$ is eliminated in $\pi_a'$. 
        This contradicts our assumption that $a$ was not eliminated.
        We conclude that all non-blocking actions that are enabled in some state of $\pi'$ are subsequently eliminated. Hence, $\pi'$ satisfies $\block$-justness of actions.
    \end{itemize}
    In either case, the finite path $\pi$ can be extended to a path $\pi'$ that satisfies $\block$-JA. We conclude $\block$-justness of actions is feasible.
\end{proof}

\section{Correctness of Formulae}\label{app:proofs}
In this appendix, we provide the correctness proofs for the presented formulae.
First, there is a supporting proposition we use repeatedly throughout the different proofs.

All our proofs assume a fixed LTSC $M = (\states,\initstate,\actionset,\transrel,\conc)$, although the $\conc$ is only relevant for JA.
We also refer to an arbitrary environment $\env$ and set of blocking actions $\block \subseteq \actionset$.
We define the length of a finite path to be the number of transitions occurring in it. A path of length $0$ contains only a single state and is called the empty path. 

\subsection{Supporting Proposition}
The following proposition gives the semantics of a least fixed point formula that occurs in several of our presented formulae.
\begin{restatable}{proposition}{proplfpbasic}\label{prop:lfp-basic}
    For all states $s \in \states$, formal variables $Y$, modal $\mu$-calculus formulae $\phi_1$ and $\phi_2$ that do not depend on $Y$, and set of actions $\alpha$, it is the case that
    $s$ is in $\llbracket \lfpb{Y}{\phi_1 \land (\phi_2 \lor \diam{\alpha}Y)} \rrbracket$ if, and only if, $s$ admits a finite path $\pi$ satisfying the following requirements:
    \begin{enumerate}
        \item all actions occurring in $\pi$ are in $\alpha$, and
        \item all states in $\pi$ are in $\llbracket \phi_1 \rrbracket$, and
        \item the final state of $\pi$ is in $\llbracket \phi_2 \rrbracket$.
    \end{enumerate}
\end{restatable}

We first need to prove a supporting lemma. To prove this lemma, we use an alternate characterisation of the semantics of least fixpoints to the one we presented in \autoref{subsec:mucalc}. We only give the definitions that we require for our proofs.
The following presentation can be found in \cite{bradfield2001modal,bradfield200712}, amongst others.

Let $Y$ be an arbitrary formal variable, $\phi$ be an arbitrary modal $\mu$-calculus formula and $\env$ an arbitrary environment. 
Let $T$ be the transformer associated with $\mu Y.\phi$, defined as
$$ T(\mathcal{F}) = \{s \in \states \mid s \in \llbracket \phi \rrbracket_{\env [Y := \mathcal{F}]}\}$$
And define
\begin{align*}
    T^0(\mathcal{F}) &= \mathcal{F}\\
 T^{i+1}(\mathcal{F}) &= T(T^i(\mathcal{F}))
\end{align*}
Then we can calculate the semantics of $\lfp{Y}{\phi}$ under $\env$ as:
\begin{align*}
    \llbracket \lfp{Y}{\phi}\rrbracket_{\env} &= \bigcup_{0 \leq i \leq |\states|} T^i(\emptyset)
\end{align*}
Note that this definition only works for finite systems, since it uses $|\states|$. A version exists for infinite systems, but is not relevant here.
We call $T^i(\emptyset)$ the $i$'th approximation of $\phi$.

For the subsequent lemmas, we fix formal variable $Y$.

\begin{lemma}\label{lem:lfp-approx}
     For all environments $\env$, states $s \in \states$, modal $\mu$-calculus formulae $\phi_1$ and $\phi_2$ that do not depend on $Y$, sets of actions $\alpha$, and natural numbers $0 \leq i \leq |\states|$, it holds that:
     $s$ is in the $i$'th approximation of $\lfpb{Y}{\phi_1 \land (\phi_2 \lor \diam{\alpha}Y)} $ under $\env$ if, and only if, $s$ admits a finite path $\pi$ meeting the following conditions:
     \begin{enumerate}
         \item $\pi$ has length at most $i - 1$, and
         \item only actions in $\alpha$ occur in $\pi$, and
         \item all states in $\pi$ are in $\llbracket \phi_1 \rrbracket_{\env}$, and
         \item the final state of $\pi$ is in $\llbracket \phi_2 \rrbracket_{\env}$.
     \end{enumerate}
\end{lemma}
\begin{proof}
    Let $T$ be the transformer of $\lfpb{Y}{\phi_1 \land (\phi_2 \lor \diam{\alpha}Y)}$.
    We prove that $s$ is in $T^i(\emptyset)$ if, and only if, $s$ admits a finite path $\pi$ of length at most $i - 1$ on which only actions in $\alpha$ occur, all states are in $\llbracket \phi_1 \rrbracket_{\env}$, and which ends in a state $s' \in \llbracket \phi_2 \rrbracket_{\env}$.
    We do this by induction on $i$.

    For the first \emph{base}, take $i = 0$. Note that $s$ is in the $i$'th approximation if $s \in T^0(\emptyset)$. However, $T^0(\emptyset) = \emptyset$, so $s$ cannot be in the $0$'th approximation.
    Indeed, we cannot have a path of length at most $-1$. So in both directions of the bi-implication, the left side of the implication does not hold.

    For the second \emph{base}, take $i = 1$. We prove the bi-implication.
    \begin{itemize}
        \item First, assume $s$ is in the first approximation. Then $s \in T^1(\emptyset) = \{s \in \states \mid s \in \llbracket \phi_1 \land (\phi_2 \lor \diam{\alpha}Y)\rrbracket_{\env [Y := \emptyset]}\}$. Hence, $s \in \llbracket \phi_1 \land (\phi_2 \lor \diam{\alpha}Y) \rrbracket_{\env[Y := \emptyset]}$.
        Through the semantics of the modal $\mu$-calculus, and using that $\phi_1$ and $\phi_2$ do not depend on $Y$, this becomes $s \in (\llbracket \phi_1 \rrbracket_{\env} \cap (\llbracket \phi_2 \rrbracket_{\env} \cup \llbracket \diam{\alpha}Y \rrbracket_{\env[Y := \emptyset]})$. Therefore, $s \in \llbracket \phi_1 \rrbracket_{\env}$, and $s$ is in $\llbracket \phi_2 \rrbracket_{\env}$ or $ s \in \{s \in \states \mid \exists_{s' \in \states}.s \xrightarrow{\alpha} s' \land s' \in \emptyset\}$. It is not possible for a state $s'$ to exist that is in $\emptyset$, hence $s \in \llbracket \phi_2 \rrbracket_{\env}$. 
        Let $\pi$ be the empty path from $s$. No actions occur on $\pi$, so trivially all occurring actions are in $\alpha$. Since $s$ satisfies both $\phi_1$ and $\phi_2$ under $\env$, the other conditions are met as well.
        We conclude that $s$ admits a finite path of length at most 0 on which only actions in $\alpha$ occur, all states are in $\llbracket \phi_1 \rrbracket_{\env}$, and which ends in a state in $\llbracket \phi_2 \rrbracket_{\env}$.
        \item Second, assume $s$ admits a path of length at most $0$ on which only actions on $\alpha$ occur, all states are in $\llbracket \phi_1 \rrbracket_{\env}$, and which ends in a state in $\llbracket \phi_2 \rrbracket_{\env}$. The only path starting in $s$ of length at most $0$ is the path consisting of only $s$. Hence, $s \in \llbracket \phi_1 \rrbracket_{\env}$ and $s \in \llbracket \phi_2 \rrbracket_{\env}$. Since $\phi_1$ and $\phi_2$ do not depend on $Y$, we also have $s \in \llbracket \phi_1 \rrbracket_{\env [Y := \emptyset]}$ and $s \in \llbracket \phi_2 \rrbracket_{\env[Y := \emptyset]}$. 
        If $s \in \llbracket \phi_2 \rrbracket_{\env[Y := \emptyset]}$, it is also in the superset $\llbracket \phi_2 \rrbracket_{\env[Y := \emptyset]} \cup \llbracket \diam{\alpha}Y\rrbracket_{\env[Y := \emptyset]}$. We conclude that $s \in  \{s \in \states \mid s \in \llbracket \phi_1 \land ( \phi_2 \lor \diam{\alpha}Y)\rrbracket_{\env [Y := \emptyset]}\} = T^1(\emptyset)$ and hence $s$ is in the first approximation.
    \end{itemize}

    The \emph{induction hypothesis} we use is that a state $s'$ is in the $k$'th approximation of $\mu Y. (\phi_1 \land (\phi_2 \lor \diam{\alpha}Y))$ if, and only if, $s'$ admits a finite path of length at most $k - 1$ on which only actions in $\alpha$ occur, all states are in $\llbracket \phi_1 \rrbracket_{\env}$, and which ends in a state in $\llbracket \phi_2 \rrbracket_{\env}$. This is for all $k \geq 1$.
    
    For the \emph{step} case, we prove the claim for $k + 1$.
    Let $S$ be the set of states that admit finite paths of length at most $k - 1$ on which only actions in $\alpha$ occur, all states are in $\llbracket \phi_1 \rrbracket_{\env}$ and which end in a state in $\llbracket \phi_2 \rrbracket_{\env}$. By the induction hypothesis, $S = T^k(\emptyset)$.
    Since this lemma is a bi-implication, we prove both directions separately.
    \begin{itemize}
        \item We assume $s \in T^{k+1}(\emptyset)$. We need to prove $s$ admits a path $\pi$ that is of length at most $k + 1 - 1 = k$, on which only actions in $\alpha$ occur, all states are in $\llbracket\phi_1 \rrbracket_{\env}$, and which ends in a state in $\llbracket \phi_2 \rrbracket_{\env}$.
        We have $s \in T^{k+1}(\emptyset) = T(T^k(\emptyset)) = T(S)$
        Hence, $s \in \{s \in \states \mid s \in \llbracket \phi_1 \land (\phi_2 \lor \diam{\alpha}Y) \rrbracket_{\env [Y := S]}\}$. This reduces to $s \in \llbracket \phi_1 \rrbracket_{\env}$, and $s \in \llbracket \phi_2 \rrbracket_{\env} \lor s \in \{s \in \states \mid \exists_{s' \in \states}.s \xrightarrow{\alpha} s' \land s' \in S\}$, because $\phi_1$ and $\phi_2$ do not depend on $Y$.
        We do a case distinction on whether $s \in \llbracket \phi_2 \rrbracket_{\env}$.
        \begin{itemize}
            \item If $s \in \llbracket \phi_2 \rrbracket_{\env}$, then the path $\pi$ consisting of only $s$ is a path of length 0 on which only actions in $\alpha$ occur, and all states are in $\llbracket \phi_1 \rrbracket_{\env}$, and which ends in a state in $\llbracket \phi_2 \rrbracket_{\env}$. Since we assumed $k \geq 1$, we know $0 \leq k$, hence $s$ admits a path meeting the requirements of length at most $k$.
            \item If $s \not\in \llbracket \phi_2 \rrbracket_{\env}$, then $s \in \{s \in \states \mid \exists_{s' \in \states}.s \xrightarrow{\alpha} s' \land s' \in S\}$. Hence, there exists a state $s'$ such that that there exists an $\alpha$-transition $t$ from $s$ to $s'$ and $s'$ is in $S$.
            Since $s' \in S$, we know $s'$ admits a path $\pi'$ of length at most $k-1$, on which only actions in $\alpha$ occur, all states are in $\llbracket \phi_1 \rrbracket_{\env}$, and which ends in a state satisfying $\llbracket \phi_2 \rrbracket_{\env}$.
            Let $\pi = s t \pi'$. Since $\pi'$ has length at most $k - 1$ and we added one transition, $\pi$ has length at most $k$. Additionally, $t$ is an $\alpha$-transition, as are all transitions in $\pi'$, so all transitions in $\pi$ are labelled with actions in $\alpha$. Since $s \in \llbracket \phi_1 \rrbracket_{\env}$ and all states in $\pi'$ are as well, all states on $\pi$ meet this requirement. Finally,  since $\pi'$ ends in a state satisfying $\llbracket \phi_2 \rrbracket_{\env}$, so does $\pi$. Hence, $\pi$ is a witness that $s$ admits a path meeting all requirements.
        \end{itemize}
        In both cases $s$ admits such a path $\pi$ of length at most $k$.

        \item We assume $s$ admits a path $\pi$ of length at most $k$ such that all transitions on $\pi$ are labelled with actions in $\alpha$, all states are in $\llbracket \phi_1 \rrbracket_{\env}$, and $\pi$ ends in a state satisfying $\llbracket \phi_2 \rrbracket_{\env}$.
        We prove $s \in T^{k+1}(\emptyset) = T(T^k(\emptyset)) = T(S) = \{s \in \states \mid s \in \llbracket \phi_1 \land (\phi_2 \lor \diam{\alpha}Y) \rrbracket_{\env [Y := S]}\}$.
        We do a case distinction on whether the length of $\pi$ is zero.
        \begin{itemize}
            \item If the length of $\pi$ is zero, then $\pi = s$ and $s \in \llbracket \phi_2 \rrbracket_{\env}$. Additionally, since all states on $\pi$ are in $\llbracket \phi_1 \rrbracket_{\env}$, so is $s$. Since $\phi_1$ and $\phi_2$ do not depend on $Y$, we also have $s \in \llbracket \phi_1 \rrbracket_{\env[Y := S]}$ and $s \in \llbracket \phi_2 \rrbracket_{\env[Y := S]}$, and hence also $s \in \llbracket \phi_1 \land (\phi_2 \lor \diam{\alpha}Y)\rrbracket_{\env [Y := S]}$. Thus, $s \in T(S) = T^{k+1}(\emptyset)$.
            \item If the length of $\pi$ is greater than zero, then there is at least one transition in $\pi$. Let $t$ be the first transition of $\pi$. Since there are only $\alpha$-transitions in $\pi$, $t$ is an $\alpha$-transition. Let $s'$ be the target of $t$, and let $\pi'$ be the suffix of $\pi$ starting in $s'$. Then since the length of $\pi$ is at most $k$, the length of $\pi'$ is at most $k -1$. Hence, $\pi'$ witnesses that $s'$ admits a path of length at most $k-1$ on which only $\alpha$-transitions occur, all states are in $\llbracket \phi_1 \rrbracket_{\env}$, and which ends in a state satisfying $\llbracket \phi_2 \rrbracket_{\env}$. Hence, $s \in S$.
            So $s$ admits an $\alpha$-transition, namely $t$, to a state in $S$, namely $s'$. Therefore $s \in \llbracket \diam{\alpha}Y \rrbracket_{\env [ Y := S]}$ and hence also $s \in \llbracket \phi_1 \land (\phi_2 \lor \diam{\alpha}Y)\rrbracket_{\env[Y := S]}$. We conclude that $s \in T(S) = T^{k+1}(\emptyset)$.
        \end{itemize}
        In both cases we demonstrate that $s$ is in the $k+1$'th approximation.
    \end{itemize}
    We have proven both sides of the bi-implication that $s$ is in the $k+1$'th approximation of $\mu Y. (\phi_1 \land (\phi_2 \lor \diam{\alpha}Y))$ if, and only if, $s$ admits a finite path of length at most $k$ on which only $\alpha$ actions occur, all states are in $\llbracket \phi_1 \rrbracket_{\env}$, and which ends in a state satisfying $\llbracket \phi_2 \rrbracket_{\env}$.
    This proves the step case.

    By induction, we have proven the claim holds for all $i \geq 0$. Therefore it also holds for all $0 \leq i \leq |\states|$. We conclude the lemma holds.
\end{proof}

We can now prove the main claim, \autoref{prop:lfp-basic}:
\proplfpbasic*
\begin{proof}
    This claim is a bi-implication, so we prove both directions.
    \begin{itemize}
        \item If $s$ is in the semantics of $\lfpb{Y}{\phi_1 \land (\phi_2 \lor \diam{\alpha}Y)}$, then $s$ is in the least fixed point of the transformer $T$ matching this formula. Hence, there are one or more natural numbers $0 \leq i \leq |\states|$ such that $s \in T^i(\emptyset)$. Let $i$ be the smallest such number, then by \autoref{lem:lfp-approx}, $s$ admits a path of length at most $i - 1$ that meets all three conditions. This path witnesses that $s$ indeed admits a path meeting all three conditions.
        \item Assume $s$ admits at least one finite path that satisfies all three conditions. Let $\pi$ be the shortest such path that $s$ admits. Let $k$ be the length of $\pi$. We first prove that $0 \leq k + 1 \leq |\states|$. Trivially, a path has length at least $0$, so $0 \leq k + 1$. Towards a contradiction, assume $k + 1 > |\states|$. A path of length $j$ contains $j + 1$ individual instances of states: the initial state of the path and the target of every transition on the path. The path $\pi$ has length $k$, and so contains at least $k + 1$ individual occurrences of states, and $k + 1 > |\states|$. Hence, $\pi$ contains strictly more than $|\states|$ individual instances of states. Considering there are exactly $|\states|$ states in the LTS, by the pigeonhole principle there must be at least one state $s'$ that is visited at least twice on $\pi$. Let $\pi_1$ be the prefix of $\pi$ up until the first occurrence of $s'$, and let $\pi_2$ be the suffix of $\pi$ starting in the last occurrence of $s'$. Now let $\pi' = \pi_1 \co \pi_2$. This is a valid path, since $\pi_1$ ends in state $s'$ and $\pi_2$ starts in this state. $\pi'$ contains a subset of the actions and states of $\pi$, and has the same final state as $\pi$, so it satisfies all three conditions. Since we chose $s'$ to be a state that occurred more than once on $\pi$, $\pi'$ contains at least one transition less than $\pi$ and hence $\pi \neq \pi'$. In fact, $\pi'$ is shorter than $\pi$. However, we chose $\pi$ to be the shortest path that meets all three conditions starting in $s$. We have reached a contradiction.
        Hence, we conclude that $k + 1 \leq |\states|$. 
        This means that $s$ admits a path of length at most $k$ with $0 \leq k + 1 \leq |\states|$ that meets the three conditions. By \autoref{lem:lfp-approx}, this means that $s$ is in the $k+1$'th approximation of $\lfpb{Y}{\phi_1 \land(\phi_2 \lor \diam{\alpha}Y)}$. Since the semantics of the formula are the union of all approximations, we conclude $s$ is in $\llbracket \lfpb{Y}{\phi_1 \land (\phi_2 \lor \diam{\alpha}Y)} \rrbracket$.
    \end{itemize}
    We have proven both directions of the bi-implication.
\end{proof}

\subsection{Proof of Progress Formula}\label{app:unfair-proof}

We prove \autoref{thm:unfair}:
\thmunfair*

\autoref{mucalc:unfair} is:
\begin{equation*}
    \neg\diam{\dpre}\nu X. (\diam{\den}\tp \lor \boxm{\nonblock}\fp \lor \diam{\comp{\ddis}}X)
\end{equation*}

We fix arbitrary $\block$, $\dpre$, $\ddis$ and $\den$ for this proof.
We first prove the formula without the $\neg\diam{\dpre}$ at the start.
For this, let $S_P$ be the set of states that admit $\block$-progressing paths that are $(\emptyseq,\ddis,\den)$-violating.
In other words, these are the states that admit paths that are $\block$-progressing and that are $\ddis$-free up until an occurrence of an action in $\den$.

We first prove that $S_P$ is a fixed point of $\nu X.(\diam{\den}\tp \lor \boxm{\nonblock}\fp \lor \diam{\comp{\ddis}}X)$, and then that it is the greatest fixed point.
\begin{lemma}\label{lem:unfair-fix}
    $S_P$ is a fixed point of the transformer $T_P$ defined by:
    $$T_P(\mathcal{F}) = \bigcap_{a \in \nonblock}\{s \in \states \mid s\in \llbracket \diam{\den}\tp \lor \boxm{\nonblock}\fp \lor \diam{\comp{\ddis}}X \rrbracket_{\update{\env}{X}{\mathcal{F}}}\}$$
\end{lemma}
\begin{proof}
    To prove $S_P$ is a fixed point of $T_P$, we prove $T_P(S_P) = S_P$. We do this through mutual set inclusion.
    \begin{itemize}
        \item Let $s$ be an arbitrary state in $T_P(S_P)$. We therefore know that $s \in \llbracket \diam{\den}\tp \lor \boxm{\nonblock}\fp \lor \diam{\comp{\ddis}}X\rrbracket_{\update{\env}{X}{S_P}}$. We do a case distinction on which of those conditions $s$ satisfies.
        \begin{itemize}
            \item If $s$ satisfies $\diam{\den}\tp$, then there is a transition $t$ enabled in $s$ that is labelled with an action in $\den$. Let $\pi = s t \target{t}$. This is a path that is $\ddis$-free up until the first occurrence of an action in $\den$. We now extend $\pi$ arbitrarily, either until a $\block$-locked state is reached or infinitely. This is always possible: as long we are not in a $\block$-locked state there is always a non-blocking action enabled that can be appended to the path we are constructing. This way, a $\block$-progressing path that is $(\emptyseq,\ddis,\den)$-violating is constructed.

            \item If $s$ satisfies $\boxm{\nonblock}\fp$, then $s$ is a $\block$-locked state. Hence, the empty path is a path that $s$ admits that is $\block$-progressing and on which, trivially, no actions in $\ddis$ occur. Hence, $s$ admits a $\block$-progressing path that is $(\emptyseq,\ddis,\den)$-violating.

            \item If $s$ satisfies $\diam{\comp{\ddis}}X$ with $X = S_P$, then $s$ admits a transition labelled with an action not in $\ddis$ to a state in $S_P$.
            Let $t$ be such a transition and $s' = \target{t}$. Then since $s' \in S_P$, $s'$ admits a path $\pi'$ that is $\block$-progressing and $\ddis$-free up until the first occurrence of an action in $\den$. Let $\pi = st\pi'$: this path too is $\block$-progressing and $\ddis$-free up until the first occurrence of an action in $\den$. Hence, $\pi$ witnesses that $s$ admits a $\block$-progressing path that is $(\emptyseq,\ddis,\den)$-violating.
        \end{itemize}
        Therefore $s \in T_P(S_P) \imps s \in S_P$.
        \item Let $s$ be an arbitrary state in $S_P$, then $s$ admits a path $\pi$ that is $\block$-progressing and $\ddis$-free up until the first occurrence of $\den$. We prove $s \in \llbracket \diam{\den}\tp \lor \boxm{\nonblock}\fp \lor \diam{\comp{\ddis}}X\rrbracket_{\update{\env}{X}{S_P}}$.
        First, we do a case distinction on whether $\pi$ is the empty path.
        \begin{itemize}
            \item If it is, then since $\pi$ is $\block$-progressing, $s$ must be a $\block$-locked state. Hence $s \in \llbracket \boxm{\nonblock}\fp \rrbracket$.
            \item If $\pi$ is not the empty path, we do a case distinction on whether the first transition of $\pi$ is labelled with an action in $\den$.
            \begin{itemize}
                \item If it is, then $s \in \llbracket \diam{\den}\tp \rrbracket$.
                \item If it is not, then since $\pi$ is $\ddis$-free up until the first occurrence of an action in $\den$, and the first transition of $\pi$ is not labelled with an action in $\den$, this transition must be labelled with an action not in $\ddis$. Let $t$ be this transition and $s'$ the target of $t$. Let $\pi'$ be the suffix of $\pi$ starting in $s'$. Then because $\pi$ is $\block$-progressing and $\ddis$-free up until the first occurrence of an action in $\den$, $\pi'$ meets those requirements as well. Hence, $s' \in S_P$. We conclude that $s$ admits a non-$\ddis$ transition to a state in $S_P$, and so $s \in \llbracket \diam{\comp{\ddis}}X\rrbracket_{\update{\env}{X}{S_P}}$.
            \end{itemize}
        \end{itemize}
        In every case, $s$ is in one of the options of $\llbracket \diam{\den}\tp \lor \boxm{\nonblock}\fp \lor \diam{\comp{\ddis}}X\rrbracket_{\update{\env}{X}{S_P}}$. Hence, $s \in T_P(S_P)$.
    \end{itemize}
    By mutual set inclusion, we conclude that $T_P(S_P) = S_P$.
\end{proof}

\begin{lemma}\label{lem:unfair-greatest}
    $S_P$ is the greatest fixed point of the transformer defined in \autoref{lem:unfair-fix}.
\end{lemma}
\begin{proof}
    We prove $S_P$ is the greatest fixed point by showing that for every $\mathcal{F} \subseteq \states$ such that $T_P(\mathcal{F}) = \mathcal{F}$, $\mathcal{F} \subseteq S_P$. To this end, let $\mathcal{F}$ be an arbitrary fixed point of $T_P$ and let $s$ be an arbitrary state in $\mathcal{F}$.
    We show that $s \in S_P$ by constructing a path $\pi$ from $s$ that is $\block$-progressing and $\ddis$-free up until the first occurrence of an action in $\den$.

    We initialise the construction with $\pi = s$. Throughout the construction, let $s'$ be the last state of the path $\pi$ constructed thus far.
    The invariants of this construction are that $s' \in \mathcal{F}$, and the $\pi$ constructed thus far is $\ddis$-free. This trivially holds initially.

    We know that $s' \in \mathcal{F}$ and $\mathcal{F} = T_P(\mathcal{F})$. Hence, $s' \in \llbracket \diam{\den}\tp \lor \boxm{\nonblock}\fp \lor \diam{\comp{\ddis}}X\rrbracket_{\update{\env}{X}{\mathcal{F}}}$.
    We do a case distinction on which of the three 
    cases $s'$ satisfies.
    \begin{itemize}
        \item If $s' \in \llbracket \diam{\den}\tp \rrbracket$ then $s'$ admits a transition labelled with an action in $\den$. Similar to the argument given in the proof of \autoref{lem:unfair-fix}, we can use this transition as the first transition of an otherwise arbitrary path $\pi'$ that is $\block$-progressing and, because its first transition is labelled with an action in $\den$, $\ddis$-free up until the first occurrence of an action in $\den$. We append this path $\pi'$ to the $\pi$ constructed thus far and terminate the construction. The resulting path $\pi$ is $\block$-progressing and $\ddis$-free up until the first occurrence of an action in $\den$, since before this iteration of the construction $\pi$ was $\ddis$-free. Hence, it witnesses $s \in S_P$.
        \item If $s' \in \llbracket \boxm{\nonblock}\fp \rrbracket$ then $s'$ is a $\block$-locked state and hence we can terminate with the $\pi$ constructed thus far. This path is $\block$-progressing since it ends in a $\block$-locked state and by the invariant it is $\ddis$-free. Hence, it witnesses $s \in S_P$.
        \item If $s' \in \llbracket \diam{\comp{\ddis}}X\rrbracket_{\update{\env}{X}{\mathcal{F}}}$ then $s'$ admits a non-$\ddis$-transition to a state in $\mathcal{F}$. We append this transition and its subsequent target to $\pi$ and continue the construction.
        We know that the new last state of $\pi$ must be in $\mathcal{F}$, and that the path is still $\ddis$-free.
    \end{itemize}
    We have argued in the termination cases that the resulting path witnesses $s \in S_P$.
    If the construction never terminates, then we have an infinite path, which is therefore $\block$-progressing. By the invariant, the path will also be $\ddis$-free, so it witnesses $s \in S_P$.

    Since such a witness can always be constructed, $\mathcal{F} \subseteq S_P$ and so $S_P$ is the greatest fixed point of $T_P$.
\end{proof}

\begin{corollary}\label{cor:unfair}
        The set of states characterised by $\nu X.(\diam{\den}\tp \lor \boxm{\nonblock}\fp \lor \diam{\comp{\ddis}}X)$ is exactly the set of states that admit $\block$-progressing paths that are $(\emptyseq,\ddis,\den)$-violating.
\end{corollary}

All that remains is to prove that the $\dpre$ aspect of the $(\dpre, \ddis, \den)$-violating path is properly characterised in the full formula.
\begin{lemma}\label{lem:violate-unfair}
    For all environments $\env$ and states $s \in \states$, it holds that $s \in \llbracket \diam{\dpre}\nu X.(\diam{\den}\tp \lor \boxm{\nonblock}\fp \lor \diam{\comp{\ddis}}X) \rrbracket_{\env}$ if, and only if, $s$ admits a path that is $\block$-progressing and $(\dpre, \ddis, \den)$-violating.
\end{lemma}
\begin{proof}
    It follows directly from the definition of the diamond operator that this formula characterises those states that admit a path $\pi$ that has a prefix matching $\dpre$ that ends in some state $s'$ in $\llbracket\nu X.(\diam{\den}\tp \lor \boxm{\nonblock}\fp \lor \diam{\comp{\ddis}}X) \rrbracket_{\env}$. By \autoref{cor:unfair}, $s'$ admits a $\block$-progressing path $\pi'$ that is $(\emptyseq,\ddis,\den)$-violating. Hence, $\pi \co \pi'$ witnesses the admission of a $\block$-progressing and $(\dpre, \ddis, \den)$-violating path.
\end{proof}
\autoref{mucalc:unfair} is the negation of this formula, hence it expresses that a state does not admit such a path.
\autoref{thm:unfair} follows directly.

\subsection{Proof of WFA, WHFA and JA Formulae}\label{app:gen-proof}

In \autoref{subsec:gen}, we defined finitely realisable predicates as follows:
\defgen*

We gave a general formula for feasible, finitely-realisable completeness criteria, 
\autoref{mucalc:gen}:
\begin{equation*}
    \neg\diam{\dpre}\nu X. (\bigwedge_{a \in \nonblock}(\formon{a} \imps 
    \diam{\clos{\comp{\ddis}}}(\diam{\den}\tp \lor (\formoff{a} \land X) \lor 
    \diam{\actelim{a} \setminus \ddis}X)))
\end{equation*}

We will prove the correctness of this formula before proving that it can indeed be used to express WFA, WHFA and JA.
However, before proving anything on the general formula, we first prove a few properties that can be derived from those given in \autoref{def:gen}.
\begin{proposition}\label{prop:gen-finite}
    Every $\block$-progressing, finite path satisfies every finitely realisable path predicate, for all $\block\subseteq\actionset$.   
\end{proposition}
\begin{proof}
    Let $P$ be an arbitrary finitely realisable predicate on paths and let $\pi$ be an arbitrary $\block$-progressing, finite path.
    We assume towards a contradiction that $\pi$ does not satisfy $P$.
    Then there must exist a state $s$ on $\pi$ and an action $a \in \nonblock$ such that $s$ satisfies $ \formon{a}$, and the suffix of $\pi$ starting in $s$, $\pi'$, does not contain an occurrence of an action in $\actelim{a}$ nor a state that satisfies $\formoff{a}$.
    Since $\pi$ is finite, so is $\pi'$. Let $s_f$ be the final state of $\pi'$.
    By the persistent property, $s_f$ satisfies $\formon{a}$.
    However, by the locking property that means $s_f$ cannot be a $\block$-locked state.
    Since $s_f$ is also the final state of $\pi$, this means $\pi$ is not a $\block$-progressing path.
    This contradicts the assumption on $\pi$; we conclude that $\pi$ satisfies $P$.
\end{proof}

\begin{proposition}\label{prop:gen-to-progress}
    Every path that satisfies a finitely realisable path predicate $P$ is $\block$-progressing.
\end{proposition}
\begin{proof}
    Let $P$ be an arbitrary finitely realisable predicate on paths and let $\pi$ be an arbitrary path.
    If $\pi$ is infinite, then it is trivially $\block$-progressing.
    Hence, we assume $\pi$ is finite. 
    Let $s$ be the final state of $\pi$.
    Towards a contradiction, assume $s$ is not $\block$-locked.
    Then by the locking property, there must exist an action $a \in \nonblock$ such that $s \in \llbracket \formon{a}\rrbracket$.
    Since $\pi$ satisfies $P$, it must be the case that in the suffix $\pi'$ of $\pi$ starting in $s$, an action in $\actelim{a}$ occurs or a state occurs that satisfies $\formoff{a}$, by the invariant property.
    The suffix of $\pi$ starting in its final state $s$ consists of only the state $s$ itself, so trivially there are no action occurrences. So it must be the case that $s \in \llbracket \formoff{a} \rrbracket$.
    However, by the exclusive property this is impossible.
    We have reached a contradiction and conclude that $s$ is a $\block$-locked state, and hence $\pi$ is $\block$-progressing.
\end{proof}

\begin{proposition}\label{prop:gen-prepend}
    If a path $\pi$ satisfies finitely realisable path predicate $P$, then every path of which $\pi$ is a suffix also satisfies $P$.
\end{proposition}
\begin{proof}
    Let $P$ be an arbitrary finitely realisable predicate on paths and let $\pi$ be an arbitrary path that satisfies $P$, with initial state $s_0$.
    Let $\pi'$ be an arbitrary finite path that ends in $s_0$.
    We prove that $\pi'' = \pi' \co \pi$ satisfies $P$.
    Towards a contradiction, assume $\pi''$ does not satisfy $P$.
    Then there must exist a state $s$ on $\pi''$ such that $s$ satisfies $\formon{a}$ for some $a \in \nonblock$, yet the suffix of $\pi''$ starting in $s$, $\pi_s$, contains no occurrence of an action in $\actelim{a}$ nor a state that satisfies $\formoff{a}$.
    If $s$ is on $\pi$, then since $\pi$ satisfies $P$ this situation is impossible.
    Hence, we must assume $s$ is on $\pi'$.
    Let $\pi_s'$ be the subpath of $\pi''$ from $s$ to $s_0$, this is a prefix of $\pi_s$.
    Since $\pi_s$ contains no occurrence of an action in $\actelim{a}$ or a state in $\formoff{a}$, the same holds for $\pi_s'$.
    By the persistent property, that means $s_0$ satisfies $\formon{a}$.
    Since $s_0$ is on $\pi$, and $\pi$ satisfies $P$, this means that $\pi$ must contain an occurrence of an action in $\actelim{a}$ or a state in which $\formoff{a}$ is satisfied.
    Since $\pi$ is a suffix of $\pi_s$ as well, $\pi_s$ contains such an occurrence too.
    This contradicts that $\pi_s$ does not contain such an occurrence.
    Hence, we conclude $\pi''$ satisfies $P$.    
\end{proof}

\begin{proposition}\label{prop:gen-suffix}
    If a path $\pi$ satisfies a finitely realisable path predicate $P$, then every suffix of $\pi$ also satisfies $P$.
\end{proposition}
\begin{proof}
    This follows directly from the invariant property: a path $\pi$ satisfies an finitely realisable predicate on paths $P$ exactly when every state on $\pi$ meets the condition that all non-blocking actions for which the state satisfies $\formon{a}$ are eliminated (through $\actelim{a}$ or $\formoff{a}$) in the suffix of the path starting in that state.
    If this condition holds on all states of a path, it also holds on all states of a suffix of that path.
\end{proof}

We now prove \autoref{thm:gen}:
\thmgen*

For the proof of this theorem, we fix $\block$, $\dpre$, $\ddis$ and $\den$. We also fix a feasible, finitely realisable path predicate $P$.
To characterise the semantics of the formula, we first split it into multiple smaller subformulae.
\begin{align*}
    &\mathit{violate_G} &&= \diam{\dpre}\mathit{invariant_G}\\
    &\mathit{invariant_G} &&= \nu X. (\bigwedge_{a \in \nonblock}(\formon{a} \imps \mathit{eliminate_G}(a)))\\
    &\mathit{eliminate_G}(a) &&= \diam{\clos{\comp{\ddis}}}(\diam{\den}\tp \lor (\formoff{a} \land X) \lor \diam{\actelim{a} \setminus \ddis}X )
\end{align*}
We have that \autoref{mucalc:gen} $= \neg\mathit{violate_G}$.

We can now characterise the semantics of parts of the formula.
We start with $\mathit{eliminate_G}$.
\begin{lemma}\label{lem:sat}
    For all environments $\env$, states $s \in \states$, actions $a \in \nonblock$ and sets $\mathcal{F} \subseteq \states$, it holds that $s \in \llbracket \mathit{eliminate_G(a)\rrbracket_{\env [X := \mathcal{F}]}}$ if, and only if, $s$ admits a finite path $\pi$ with final state $s_{\mathit{final}}$ that satisfies the following conditions:
    \begin{enumerate}
        \item $\pi$ is $\ddis$-free, and
        \item one of the following three holds:
        \begin{enumerate}
            \item at least one action in $\den$ is enabled in $s_{\mathit{final}}$, or
            \item $s_{\mathit{final}} \in \mathcal{F}$ and $s_{\mathit{final}}$ satisfies $\formoff{a}$, or
            \item $s_{\mathit{final}} \in \mathcal{F}$ and the last transition in $\pi$, $t_{\mathit{final}}$, is labelled with an action in $\actelim{a}\setminus\ddis$.
        \end{enumerate}
    \end{enumerate}
\end{lemma}
\begin{proof}
    Note that the formula $\mathit{eliminate_G}(a)$ is equivalent to $\mu Y. (\diam{\den}\tp \lor (\formoff{a} \land X) \lor \diam{\actelim{a}\setminus\ddis}X \lor \diam{\comp{\ddis}}Y)$, since the regular formula hides the least fixpoint operator.
    Since $\formoff{a}$ is a closed formula,$\diam{\den}\tp \lor (\formoff{a} \land X) \lor \diam{\actelim{a}\setminus \ddis}X$ is a formula that does not depend on $Y$. We can therefore apply \autoref{prop:lfp-basic} (with $\phi_1 = \tp$) to conclude that $s \in \llbracket \mathit{eliminate_G}(a) \rrbracket_{\update{\env}{X}{\mathcal{F}}}$ if, and only if, $s$ admits a finite path $\pi'$ meeting the following conditions:
    \begin{enumerate}
    \setcounter{enumi}{2}
        \item all actions occurring in $\pi'$ are in $\comp{\ddis}$, and
        \item the final state of $\pi'$, $s'_{\mathit{final}}$ is in $\llbracket \diam{\den}\tp \lor (\formoff{a} \land X) \lor \diam{\actelim{a} \setminus \ddis}X \rrbracket_{\update{\env}{X}{\mathcal{F}}}$. 
    \end{enumerate}
    Condition 3 can be restated as $\pi'$ being $\ddis$-free.
    Condition 4 can be subdivided into three options: at least one action in $\den$ is enabled in $s'_{\mathit{final}}$ (4a), or $s'_{\mathit{final}}\in \mathcal{F}$ and  $s'_{\mathit{final}}$ satisfies $\formoff{a}$ (4b), or $s'_{\mathit{final}}$ admits a transition to a state in $\mathcal{F}$ that is labelled with an action in $\actelim{a} \setminus \ddis$ (4c).
    The correspondence between 1 and 3, as well as 2a and 4a, and 2b and 4b, is direct. In the a and b cases, in both directions, $\pi$ and $\pi'$ describe paths meeting the exact same conditions.
    It remains to argue that if $s$ admits a path $\pi$ meeting conditions 1 and 2c, it also admits a path $\pi'$ meeting conditions 3 and 4c, and vice versa. 
    Here too, the correspondence is rather direct: $\pi$ contains at least one transition, its last transition, which is labelled with an action in $\actelim{a} \setminus \ddis$ to a state in $\mathcal{F}$. A path $\pi'$ can be constructed from $\pi$ by dropping this last transition. The other way around, 4c witnesses such a transition exists, so by appending this transition to $\pi'$ a path $\pi$ can be constructed.
    Hence, it follows from \autoref{prop:lfp-basic} that $s$ indeed admits a path meeting conditions 1 and 2 exactly when it is in the semantics of $\mathit{eliminate_G}(a)$ when $X = \mathcal{F}$.
\end{proof}

We proceed to characterise the semantics of $\mathit{invariant_G}$. 
We will prove that it exactly describes those states that admit paths that are $\block$-progressing, $(\emptyseq, \ddis,\den)$-violating and satisfy $P$. 
Formally, define set $S_G$ to be exactly those states in $M$ that admit a path $\pi$ meeting the following conditions:
\begin{itemize}
    \item $\pi$ satisfies $P$, and
    \item $\pi$ is $\block$-progressing, and
    \item $\pi$ satisfies one of the following conditions:
    \begin{itemize}
        \item $\pi$ is $\ddis$-free, or
        \item $\pi$ contains an occurrence of an action in $\den$, and the prefix of $\pi$ before the first occurrence of an action in $\den$ is $\ddis$-free.
    \end{itemize}
\end{itemize}
As before, we frequently refer to the last condition as $\pi$ being $\ddis$-free up until the first occurrence of an action in $\den$.

We prove that $\llbracket \mathit{invariant_G} \rrbracket_{\env} = S_G$. We do this by proving $S_G$ is the greatest fixed point of the transformer characterising $\mathit{invariant_G}$. To do this, we first prove it is a fixed point of this transformer and then that every fixed point of this transformer is a subset of $S_G$.
\begin{lemma}\label{lem:inv-fix}
    $S_G$ is a fixed point of the transformer $T_G$ defined by:
    \begin{multline*}
        T_G(\mathcal{F}) = \bigcap_{a \in \nonblock}\left\{s \in \states \mid s \in \llbracket\formon{a}\rrbracket_{\env [X := \mathcal{F}]} \imps 
        s \in \llbracket \mathit{eliminate_G}(a)\rrbracket_{\update{\env}{X}{\mathcal{F}}}\right\}
    \end{multline*}
\end{lemma}
\begin{proof}
    To prove $S_G$ is a fixed point of $T_G$, we need to prove $T_G(S_G) = S_G$. We do this through mutual set inclusion.
    \begin{itemize}
        \item Let $s$ be an arbitrary element of $T_G(S_G)$, we prove $s \in S_G$. From $s \in T_G(S_G)$ we know that for all actions $a \in \nonblock$, $s \in \llbracket \formon{a} \rrbracket_{\env[X := S_G]}$ implies $s \in \llbracket \mathit{eliminate_G}(a) \rrbracket_{\env[X:= S_G]}$. We do a case distinction on whether $s$ satisfies $\formon{a}$ for some $a \in \nonblock$.
        \begin{itemize}
            \item If there are no non-blocking actions $a$ such that $s$ satisfies $\formon{a}$, then by the locking property it must be the case that $s$ is a $\block$-locked state. Hence, trivially $s$ admits a $\block$-progressing path on which no actions in $\ddis$ ever occur, namely the empty path. By \autoref{prop:gen-finite}, it must be the case that a finite, $\block$-progressing path satisfies $P$. Hence, $s \in S_G$. 
            \item If there is such an action, let $a$ be an arbitrary non-blocking action such that $s$ satisfies $\formon{a}$. Because $s \in \llbracket \formon{a} \rrbracket_{\env[X := S_G]}$, $s \in \llbracket \mathit{eliminate_G}(a)\rrbracket_{\env [X := S_G]}$. By \autoref{lem:sat}, we conclude that $s$ admits a finite path $\pi$ with final state $s_{\mathit{final}}$ meeting the following requirements:
            \begin{enumerate}
                \item $\pi$ is $\ddis$-free, and
                \item one of the following three holds:
                \begin{enumerate}
                    \item at least one action in $\den$ is enabled in $s_{\mathit{final}}$, or
                    \item $s_{\mathit{final}} \in S_G$ and $s_{\mathit{final}}$ satisfies $\formoff{a}$, or
                    \item $s_{\mathit{final}} \in S_G$ and the last transition in $\pi$, $t_{\mathit{final}}$ is labelled with an action in $\actelim{a}\setminus \ddis$.
                \end{enumerate}
            \end{enumerate}
            We use this information to construct a path $\pi_w$ that witnesses $s \in S_G$. For this, we do a case distinction on whether 2a holds.
            \begin{itemize}
                \item If $\pi$ satisfies 1 and 2a, then $\pi$ is a path on which no actions in $\ddis$ occur until a state is reached where at least one action in $\den$ is enabled. Let $\pi'$ be $\pi$ extended with an arbitrary transition labelled with an action in $\den$ and the subsequent target state. Then $\pi'$ is a finite path on which there are no occurrences of actions in $\ddis$ before the first occurrence of an action in $\den$. This path is either already $\block$-progressing or not. In the former case, it is finite and therefore trivially satisfies $P$ (\autoref{prop:gen-finite}). We can then take $\pi' = \pi_w$. Alternatively, it is not $\block$-progressing. In this case, then because $P$ is feasible we can extend it to a path $\pi_w$ that satisfies $P$, and therefore by \autoref{prop:gen-to-progress} is also $\block$-progressing. On $\pi_w$, it is still the case that there are no occurrences of actions in $\ddis$ until the first occurrence of an action in $\den$. Hence, we have provided a witness for $s \in S_G$.
                \item If $\pi$ does not satisfy 2a, then it satisfies 1 and 2b, or 1 and 2c. In either case, $s$ admits a finite, $\ddis$-free path to a state $s_{\mathit{final}}$. Since $s_{\mathit{final}}$ is in $S_G$, we know it admits a $\block$-progressing path $\pi'$ that satisfies $P$ and is $\ddis$-free up until the first occurrence of an action in $\den$. Let $\pi_w = \pi \co \pi'$. Since $\pi$ is $\ddis$-free, we know that prepending $\pi$ to $\pi'$ does not violate that the resulting path is $\ddis$-free until the first occurrence of an action in $\den$. Additionally, by \autoref{prop:gen-prepend}, $\pi_w$ satisfies $P$ because $\pi'$ does. Similarly, $\pi_w$ is $\block$-progressing because $\pi'$ is. Hence, $\pi_w$ witnesses $s \in S_G$.
            \end{itemize}
        \end{itemize}
        We have proven in every scenario that $s \in S_G$.

        \item Assume $s \in S_G$, we prove $s \in T_G(S_G)$. To prove $s \in T_G(S_G)$, we need to show that for all non-blocking actions $a$ such that $s$ satisfies $\formon{a}$, $s \in \llbracket \mathit{eliminate_G}(a) \rrbracket_{\env [X := S_G]}$. Towards this end, we prove $s \in \llbracket \mathit{eliminate_G}(a) \rrbracket_{\env [X := S_G]}$, where $a$ is an arbitrary non-blocking action such that $\formon{a}$ is satisfied in $s$. By \autoref{lem:sat}, it is sufficient to prove that $s$ admits a finite path $\pi$ with final state $s_{\mathit{final}}$ satisfying the the same two conditions as given previously as (1) and (2).
        We know from $s \in S_G$ that $s$ admits a $\block$-progressing path $\pi_G$ that satisfies $P$ and is $\ddis$-free up until the first occurrence of an action in $\den$. We use $\pi_G$ to construct $\pi$.

        First, we do a case distinction on whether there is an occurrence of an action in $\den$ in $\pi_G$.
        \begin{itemize}
            \item If there is, then let $t_{\mathit{en}}$ be the first transition of $\pi_G$ that is labelled with an action in $\den$. Let $s_{\mathit{final}}$ be the source state of $t_{\mathit{en}}$. Then we take $\pi$ to be the prefix of $\pi_G$ ending in $s_{\mathit{final}}$. Since $\pi_G$ is $\ddis$-free up until the first occurrence of an action in $\den$, and $t_{\mathit{en}}$ is the first occurrence of action action in $\den$, we know $\pi$ is $\ddis$-free, so satisfies condition 1. Additionally, by construction an action in $\den$ is enabled in $s_{\mathit{final}}$, the final state of $\pi$, because $t_{\mathit{en}}$ is enabled. Therefore, $\pi$ satisfies condition 2a. We have constructed a path $\pi$ meeting conditions 1 and 2.
            \item If there is no occurrence of an action in $\den$ in $\pi_G$, then $\pi_G$ is fully $\ddis$-free. Recall that we assumed that $s$ satisfies $\formon{a}$ for the non-blocking action $a$. Since $\pi_G$ satisfies $P$, we know by the invariant property that $\pi_G$ contains an occurrence of some action in $\actelim{a}$ or a state that satisfies $\formoff{a}$. We do a case distinction on whether some state of $\pi_G$ satisfies $\formoff{a}$.
            \begin{itemize}
                \item If $\pi_G$ contains a state satisfying $\formoff{a}$, let $s_{\mathit{final}}$ be the first such state and let $\pi$ be the prefix of $\pi_G$ ending in $s_{\mathit{final}}$. Then $\pi$ is $\ddis$-free because $\pi_G$ is, and ends in a state satisfying $\formoff{a}$. To show $\pi$ satisfies conditions 1 and 2b, it remains to show that $s_{\mathit{final}}$ is in $S_G$. By \autoref{prop:gen-suffix}, the suffix $\pi_G'$ of $\pi_G$ starting in $s_{\mathit{final}}$ satisfies $P$. Additionally, $\pi_G'$ is $\block$-progressing and $\ddis$-free because $\pi_G$ is. Hence, $\pi_G'$ witnesses $s_{\mathit{final}} \in S_G$. The $\pi$ we have constructed satisfies conditions 1 and 2.
                \item If no state of $\pi_G$ satisfies $\formoff{a}$, then $\pi_G$ must contain an occurrence of some action in $\actelim{a}$. Then we let $t$ be the first transition that occurs in $\pi_G$ that is labelled with an action in $\actelim{a}$. Let $\action{t} = b$. Note that because $\pi_G$ is $\ddis$-free, $b \in \actelim{a}\setminus\ddis$. Let $\pi$ be the prefix of $\pi_G$ such that $t$ is the final transition of $\pi$. Then $\pi$ is $\ddis$-free (condition 1) and the last transition of $\pi$ is labelled with an action in $\actelim{a}\setminus\ddis$. To show this $\pi$ satisfies both 1 and 2c, it remains to show that the final state of $\pi$, $s_{\mathit{final}}$ is in $S_G$. The same argument applies here as in the previous case. Hence, $\pi$ satisfies conditions 1 and 2.
            \end{itemize}
        \end{itemize}
        We have proven $s \in \llbracket \mathit{eliminate_G}(a)\rrbracket_{\env [X := S_G]}$ for all $a$ such that $s$ satisfies $\formon{a}$, and hence $s \in T_G(S_G)$.
    \end{itemize}
    By mutual set inclusion, $S_G$ is a fixed point of $T_G$.
\end{proof}

Next, we need to prove $S_G$ is the greatest fixed point of the transformer.
For this, we first need a supporting lemma.
\begin{lemma}\label{lem:inner}
    For all states $s$ in a fixed point $\mathcal{F}$ of $T_G$ as defined in \autoref{lem:inv-fix}, if there is no action in $\den$ that is reachable from $s$ without doing an action in $\ddis$ and there exists at least one action $a \in \nonblock$ such that $s$ satisfies $\formon{a}$, then there exists a finite path $\pi$ from $s$ to some state $s'$ meeting all of the following conditions:
    \begin{enumerate}
        \item $s' \in \mathcal{F}$, and
        \item $\pi$ has length at least one, and
        \item $\pi$ is $\ddis$-free, and
        \item for all actions $a \in \nonblock$ such that $s$ satisfies $\formon{a}$, there is a state on $\pi$ that satisfies $\formoff{a}$ or there is a transition on $\pi$ labelled with an action in $\actelim{a}$.
    \end{enumerate}
\end{lemma}
\begin{proof}
    Let $\mathcal{F}$ be an arbitrary fixed point of $T_G$ and let $s$ be an arbitrary state in $\mathcal{F}$ such that there is at least one non-blocking action for which $s$ satisfies $\formonsym$ and it is impossible to reach a state in which an action in $\den$ is enabled from $s$ without doing an action in $\ddis$.
    Let $L$ be the set of non-blocking actions for which $s$ satisfies $\formonsym$. By the definition of $s$, we have $|L| \geq 1$.

    We prove the more general claim that for all subsets $L' \subseteq L$ that contain at least one element, $s$ admits a path $\pi$ meeting conditions 1 through 3 as well as the following rephrasing of condition 4 with respect to $L'$: for all actions $a \in L'$, there is a state on $\pi$ that satisfies $\formoff{a}$ or there is a transition on $\pi$ labelled with an action in $\actelim{a}$. We prove this by induction on the size of $L'$. Since $L$ contains at least one element and is a subset of itself, this also proves the lemma.

    \textit{Base case}: for $|L'| = 1$, let $a$ be the one non-blocking action in $L'$. Since $L' \subseteq L$, $s$ satisfies $\formon{a}$. Observe that since $s \in \mathcal{F}$ and $\mathcal{F} = T_G(\mathcal{F})$, $s \in \llbracket \formon{a}\rrbracket_{\update{\env}{X}{\mathcal{F}}} \imps s \in \llbracket \mathit{eliminate_G}(a) \rrbracket_{\update{\env}{X}{\mathcal{F}}}$. By \autoref{lem:sat}, this allows us to conclude that $s$ admits a finite path $\pi$ to some state $s'$ such that the following holds:
    \begin{enumerate}
    \setcounter{enumi}{4}
        \item $\pi$ is $\ddis$-free, and
        \item one of the following three holds:
        \begin{enumerate}
            \item at least one action in $\den$ is enabled in $s'$, or
            \item $s' \in \mathcal{F}$ and $s'$ satisfies $\formoff{a}$, or
            \item $s' \in \mathcal{F}$ and the last transition of $\pi$, $t$, is labelled with an action in $\actelim{a}\setminus\ddis$.
        \end{enumerate}
    \end{enumerate}
    We prove that this $\pi$ also satisfies conditions 1 through 4.
    We assumed that it is impossible to reach a state in which an action in $\den$ is enabled from $s$ without doing an action in $\ddis$. Since $\pi$ is $\ddis$-free, it is therefore impossible that an action in $\den$ is enabled in $s'$. We can conclude that condition 6a cannot hold.
    Hence, 6b or 6c must hold. 
    In either case, $s' \in \mathcal{F}$ so condition 1 is satisfied.
    If condition 6b holds, then $s'$ satisfies $\formoff{a}$. We know that $s$ satisfies $\formon{a}$ and by the exclusive property, this means $s$ and $s'$ cannot be the same state so condition 2 holds.
    If 6c holds instead, then since $\pi$ contains at least one transition it has length at least one, so condition 2 also holds.
    Condition 3 follows directly from condition 5.
    It remains to show that condition 4 is satisfied, regardless of whether 6b or 6c holds. This follows from $a$ being the only action in $L'$ and, either through 6b or 6c, there being some occurrence of a state satisfying $\formoff{a}$ or an occurrence of an action in $\actelim{a}$.
    We conclude that this $\pi$ meets conditions 1 through 4.

    \textit{Step case}: let $|L'| = n + 1$ for $n \geq 1$ and assume that the claim holds when $|L'| = n$ (\textit{induction hypothesis}). We prove that $s$ has a path $\pi$ to a state $s'$ meeting all four conditions.
    Let $a$ be an arbitrary action in $L'$. Now apply the induction hypothesis to $s$ with respect to the set $L'' = L' \setminus \{a\}$, which has size $n$. This gives us a path $\pi'$ to a state $s''$ meeting all four conditions, condition 4 specifically with respect to $L''$.
    We use $\pi'$ to construct $\pi$. 
    We do a case distinction on whether $s''$ satisfies $\formon{a}$.
    \begin{itemize}
        \item If it does, then since $s'' \in \mathcal{F}$ and $\mathcal{F}$ is a fixed point of $T_G$, $s'' \in \llbracket \mathit{eliminate_G}(a)\rrbracket_{\update{\env}{X}{\mathcal{F}}}$. We can then use the same argument as given in the base case and apply \autoref{lem:sat} to find a path $\pi''$ from $s''$ to some state $s' \in \mathcal{F}$ that is $\ddis$-free and on which an action in $\actelim{a}$ occurs or on which a state exists that satisfies $\formoff{a}$. Let $\pi = \pi' \co \pi''$; this $\pi$ satisfies all four conditions. That $s' \in \mathcal{F}$ follows from the application of \autoref{lem:sat}; we know from the application of the induction hypothesis that $\pi'$ has length at least one, so $\pi$ does as well; both $\pi'$ and $\pi''$ are $\ddis$-free; and finally for all actions in $L'$ except $a$ there is a state in which $\formoffsym$ is satisfied for this action, or an action in the associated $\actelimsym$ occurs in $\pi'$, the action $a$ itself is eliminated in $\pi''$. Hence, we eliminate all actions in $L'$.
        \item If $s''$ does not satisfy $\formoff{a}$, then by the persistent property, it must be the case that $\pi'$ already contains a state that satisfies $\formoff{a}$ or an occurrence of an action in $\actelim{a}$. Hence, $\pi'$ already satisfies all four conditions for the whole set $L'$ and so is a candidate for $\pi$.
    \end{itemize}
    In either case, there exists a path $\pi$ meeting all four conditions.
    This proves the claim on arbitrary subsets of $L$ containing at least one element. As stated, $L$ is a subset of itself and contains at least one element, and so the lemma follows directly.
\end{proof}

\begin{lemma}\label{lem:inv-greatest}
    $S_G$ is the greatest fixed point of the transformer $T_G$ as defined in \autoref{lem:inv-fix}.
\end{lemma}
\begin{proof}
    Let $\mathcal{F} \subseteq \states$ be an arbitrary fixed point of $T_G$, meaning we have $\mathcal{F} = T_G(\mathcal{F})$. We prove that $\mathcal{F} \subseteq S_G$.
    To this end, let $s$ be an arbitrary state in $\mathcal{F}$. 
    We prove $s \in S_G$ by constructing a path $\pi$ from $s$ that is $\block$-progressing, satisfies $P$ and on which no actions in $\ddis$ occur up until the first occurrence of an action in $\den$.

    We do a case distinction on whether there exists some non-blocking action $a$ such hat $s$ satisfies $\formon{a}$.
    If no such action exists, then by the locking property, $s$ is a $\block$-locked state. In this case, the empty path is a $\block$-progressing from $s$ which, by \autoref{prop:gen-finite}, satisfies $P$ and on which trivially no actions in $\ddis$ occur. This path witnesses $s \in S_G$.
    
    We proceed under the assumption that there exists some non-blocking action $a$ such that $s$ satisfies $\formon{a}$.
    Let $a$ be an arbitrary such action.
    Since $s \in \mathcal{F}$ and $\mathcal{F}$ is a fixed point of $T_G$, we conclude that $s \in \llbracket \mathit{eliminate_G}(a)\rrbracket_{\update{\env}{X}{\mathcal{F}}}$. 
    By \autoref{lem:sat} we know that $s$ admits a finite path $\pi'$ to a state $s_{\mathit{final}}$ such that the following holds:
    \begin{enumerate}
        \item $\pi'$ is $\ddis$-free, and
        \item one of the following holds:
        \begin{itemize}
            \item at least one action in $\den$ is enabled in $s_{\mathit{final}}$, or
            \item $s_{\mathit{final}} \in \mathcal{F}$ and $s_{\mathit{final}}$ satisfies $\formoff{a}$, or
            \item $s_{\mathit{final}} \in \mathcal{F}$ and the last transition in $\pi'$, $t_{\mathit{final}}$ is labelled with an action in $\actelim{a} \setminus \ddis$.
        \end{itemize}
    \end{enumerate}
    We do a case distinction on whether $s$ admits some path $\pi'$ that meets conditions 1 and 2a.
    If it does, then this $\pi'$ is a finite, $\ddis$-free path to a state in which an action in $\den$ is enabled. We extend $\pi'$ by appending a transition labelled with an action in $\den$ to it together with the associated target state, creating a finite path $\pi''$ that is $\ddis$-free up until the first occurrence of an action in $\den$. By feasibility of $P$, we can extend $\pi''$ to a path $\pi$ that satisfies $P$ and, by \autoref{prop:gen-to-progress} is also $\block$-progressing. It still holds that this path is $\ddis$-free up until the first occurrence of an action in $\den$. Hence, this $\pi$ witnesses $s \in S_G$.
    
    If $s$ does not admit a path that satisfies conditions 1 and 2a, then it is impossible to reach a state in which an action in $\den$ is enabled from $s$ without doing $\ddis$-steps.
    We can then apply \autoref{lem:inner} to conclude that $s$ admits a finite path $\pi'$ to a state $s'$ that meets the following requirements:
    \begin{enumerate}
        \item $s' \in \mathcal{F}$, and
        \item $\pi'$ has length at least one, and
        \item $\pi'$ is $\ddis$-free, and
        \item for all action $a \in \nonblock$ such that $s$ satisfies $\formon{a}$, there is a state on $\pi'$ that satisfies $\formoff{a}$ or there is a transition on $\pi'$ labelled with an action in $\actelim{a}$.
    \end{enumerate}
    We can use this to construct $\pi$ from $s$.
    We start with $\pi = \pi'$. Subsequently, we consider the last state of the path constructed thus far, $s'$: if there are no actions $a \in \nonblock$ such that $s'$ satisfies $\formon{a}$, then by the locking property the constructed path is $\block$-progressing and finite, and hence by \autoref{prop:gen-finite} it satisfies $P$. Additionally, since it is constructed from $\ddis$-free segments it is $\ddis$-free and hence witnesses $s \in S_G$. If there are non-blocking actions for which $s'$ satisfies $\formonsym$, then we can apply \autoref{lem:inner} again to find a new $\pi'$ and append it to the path constructed thus far. The procedure then repeats, potentially infinitely. Note that \autoref{lem:inner} can always be applied because the last state of the constructed path is always in $\mathcal{F}$, and we have established already that it is impossible to reach a state in which an action in $\den$ is enabled from $s$ without doing an action in $\ddis$: this remains true for all states reached through appending $\ddis$-free segments to $\pi$.

    If the construction ever terminates then it is because we reached a state where no actions are ``on''. Hence, as argued above, we have constructed a path that witnesses $s \in S_G$.
    If the procedure repeats infinitely then an infinite path is constructed, this is because every segment we add has length at least one.
    If an infinite path $\pi$ is constructed this way, then this path is trivially $\block$-progressing. It is also $\ddis$-free, since it is constructed from $\ddis$-free segments. It remains to argue $\pi$ satisfies $P$. To this end, let $s_*$ be an arbitrary state of $\pi$ and let $a$ be an arbitrary non-blocking action such that $s_*$ satisfies $\formon{a}$. Let $\pi_*$ be the suffix of $\pi$ starting in $s_*$. Towards a contradiction, assume that $\pi_*$ contains neither a state in which $\formoff{a}$ is satisfied, nor a transition labelled with an action in $\actelim{a}$. Then by the persistent property, every state of $\pi_*$ must satisfy $\formon{a}$. In our construction, we always append finite paths to the infinite path we are constructing. Hence, there exist infinitely many states of $\pi_*$ to which \autoref{lem:inner} has been applied to find the next path to append to $\pi$. Let $s_*'$ be such a state. Since $\formon{a}$ must be satisfied in $s_*'$, it must be the case that the path we added to $\pi$ after applying \autoref{lem:inner} to $s_*'$ contained a state that satisfies $\formoff{a}$ or a transition labelled with an action in $\actelim{a}$. This contradicts our assumption that neither such a state nor such a transition exists in $\pi_*$. Hence, we conclude that the infinite path $\pi$ satisfies $P$ and thus witnesses $s \in S_G$. 

    In every case, $s \in S_G$ and thus $\mathcal{F} \subseteq S_G$, therefore $S_G$ is the greatest fixed point of $T_G$.
\end{proof}
Since the semantics of $\mathit{invariant_G}$ are exactly the greatest fixed point of $T_G$, we can conclude the following from the definition of $S_G$.
\begin{corollary}\label{cor:inv}
    The set of states characterised by $\mathit{invariant_G}$ is exactly the set of states that admit $\block$-progressing, $(\emptyseq,\ddis,\den)$-violating paths that satisfy $P$.
\end{corollary}

All that remains is to prepend the $\dpre$ part of the formula.
\begin{lemma}\label{lem:violate}
    For all environments $\env$ and states $s \in \states$, it holds that $s \in \llbracket \mathit{violate_G} \rrbracket_{\env}$ if, and only if, $s$ admits a path that is $\block$-progressing, satisfies $P$ and is $(\dpre, \ddis, \den)$-violating.
\end{lemma}
\begin{proof}
    It follows directly from the definition of the diamond operator that $\mathit{violate_G}$ characterises those states that admit a path $\pi$ that has a prefix matching $\dpre$ that ends in some state $s'$ in $\llbracket \mathit{invariant_G} \rrbracket_{\env}$. By \autoref{cor:inv}, $s'$ admits a $\block$-progressing path $\pi'$ that is $\ddis$-free up until the first occurrence of an action in $\den$ and satisfies $P$. By \autoref{prop:gen-prepend}, prepending $\pi$ to $\pi'$ results in a path that still satisfies $P$. It is also still $\block$-progressing, and by construction, it has a prefix matching $\dpre$, namely $\pi$, after which it is $\ddis$-free up until the first occurrence of an action in $\den$. Hence, it is $(\dpre, \ddis, \den)$-violating.
\end{proof}
\autoref{mucalc:gen} is the negation of $\mathit{violate}_G$, hence it expresses that a state does not admit such a path.
\autoref{thm:gen} follows directly.

In \autoref{subsec:gen}, we described how $\formonsym$, $\formoffsym$ and $\actelimsym$ should be defined for WFA, WHFA and JA. We here prove that these assignments indeed satisfy the conditions on finitely realisable predicates, and then use \autoref{thm:gen} to prove the formulae correct.
\subsubsection{WFA}\label{app:wfa-proof}
We prove \autoref{cor:wfa}.
\corwfa*

\begin{proof}
We argue this assignment to $\formonsym$, $\formoffsym$ and $\actelimsym$ satisfies the four properties of finitely realisable predicates.

\begin{enumerate}
    \item The invariant property: let $\pi$ be an arbitrary path. According to $\block$-WFA, $\pi$ is complete if for all suffixes of $\pi$, all non-blocking actions that are perpetually enabled on the suffix also occur in that suffix.
    We prove this is equivalent to the condition that a path $\pi$ is complete if for all states $s$ in $\pi$ and all non-blocking actions $a$, if $s$ satisfies $\diam{a}\tp$ then the suffix of $\pi$ starting in $s$ contains a state satisfying $\boxm{a}\fp$ or an occurrence of $a$. Refer to this condition as $P$.
    Let $\pi'$ be an arbitrary suffix of $\pi$ and let $s$ be the first state of $\pi'$. Let $a$ be an arbitrary non-blocking action. If $a$ is not enabled in $s$, then WFA does not require $a$ to occur in $\pi'$ since it is not perpetually enabled, and $P$ does not require $a$ to occur in $\pi'$ because $\diam{a}\tp$ is not satisfied in $s$.
    If $a$ is enabled in $s$, we consider if $a$ is perpetually enabled on $\pi'$. If it is not, WFA does not require $a$ to occur in $\pi'$. Also $P$ does not require this, since within finitely many steps a state will occur on $\pi'$ in which $\boxm{a}\fp$ is satisfied.
    If $a$ is perpetually enabled in $\pi'$, then WFA requires $a$ to occur on $\pi'$. Since there are no states on $\pi'$ that satisfy $\boxm{a}\fp$, $P$ also requires an occurrence of $a$ in $\pi'$.
    Hence, in all cases, WFA and $P$ give the exact same requirements on action occurrences in $\pi'$.

    \item The locking property: if a state is $\block$-locked, no non-blocking actions are enabled in that state. This is identical to the definition that $\diam{a}\tp$ is violated for all $a \in \nonblock$.

    \item The exclusive property: $\diam{a}\tp = \neg\boxm{a}\fp$ for all $a \in \actionset$, so this follows directly.

    \item The persistent property: this follows from $\diam{a}\tp = \neg\boxm{a}\fp$ as well. If there is no state on a path that satisfies $\boxm{a}\fp$ for some $a \in \nonblock$, then the final state of that path satisfies $\neg\boxm{a}\fp$ and hence satisfies $\diam{a}\tp$.
\end{enumerate}
Hence, $\block$-WFA is a finitely realisable predicate on paths.

That $\block$-WFA is feasible was proven as \autoref{prop:wfa-feasible}.
We apply \autoref{thm:gen} to conclude the theorem holds.
\end{proof}

\subsubsection{WHFA}\label{app:whfa-proof}

We prove \autoref{cor:whfa}.
\corwhfa*

\begin{proof}
We argue this assignment to $\formonsym$, $\formoffsym$ and $\actelimsym$ satisfies the four properties of finitely realisable predicates.

\begin{enumerate}
    \item The invariant property: let $\pi$ be an arbitrary path. According to $\block$-WHFA, $\pi$ is complete if for all suffixes of $\pi$, all non-blocking actions that are perpetually $\block$-reachable on the suffix also occur in that suffix.
    We prove this is equivalent to the condition that a path $\pi$ is complete if for all states $s$ in $\pi$ and all non-blocking actions $a$, if $s$ satisfies $\diam{\clos{\nonblock}\co a}\tp$ then the suffix of $\pi$ starting in $s$ contains a state satisfying $\boxm{\clos{\nonblock}\co a}\fp$ or an occurrence of $a$. Refer to this condition as $P$.
    Let $\pi'$ be an arbitrary suffix of $\pi$ and let $s$ be the first state of $\pi'$. Let $a$ be an arbitrary non-blocking action. If $a$ is not $\block$-reachable from $s$, then WHFA does not require $a$ to occur in $\pi'$ since it is not perpetually reachable, and $P$ does not require $a$ to occur in $\pi'$ because $\diam{\clos{\nonblock}\co a}\tp$ is not satisfied in $s$.
    If $a$ is $\block$-reachable from $s$, we consider if $a$ is perpetually $\block$-reachable on $\pi'$. If it is not, WHFA does not require $a$ to occur in $\pi'$. Also $P$ does not require this, since within finitely many steps a state will occur on $\pi'$ in which $\boxm{\clos{\nonblock} \co a}\fp$ is satisfied.
    If $a$ is perpetually $\block$-reachable in $\pi'$, then WHFA requires $a$ to occur on $\pi'$. Since there are no states on $\pi'$ that satisfy $\boxm{\clos{\nonblock} \co a}\fp$, $P$ also requires an occurrence of $a$ in $\pi'$.
    Hence, in all cases, WHFA and $P$ give the exact same requirements on action occurrences in $\pi'$.

    \item The locking property: if a state is $\block$-locked, no non-blocking actions are enabled in that state. In that case, no non-blocking actions can be $\block$-reachable since there must be an occurrence of a blocking action before a state can be reached in which a non-blocking action is enabled. In the other direction, if a state has no $\block$-reachable non-blocking actions, then no non-blocking actions are enabled since every enabled action is also reachable. Hence, the two conditions coincide.

    \item The exclusive property: $\diam{\clos{\nonblock} \co a}\tp = \neg\boxm{\clos{\nonblock} \co a}\fp$ for all $a \in \actionset$, so this follows directly.

    \item The persistent property: this follows from $\diam{\clos{\nonblock} \co a}\tp = \neg\boxm{\clos{\nonblock} \co a}\fp$ as well. If there is no state on a path that satisfies $\boxm{\clos{\nonblock} \co a}\fp$ for some $a \in \nonblock$, then the final state of that path satisfies $\neg\boxm{\clos{\nonblock} \co a}\fp$ and hence satisfies $\diam{\clos{\nonblock} \co a}\tp$.
\end{enumerate}
Hence, $\block$-WHFA is a finitely realisable predicate on paths.

That $\block$-WHFA is feasible was proven as \autoref{prop:whfa-feasible}.
We apply \autoref{thm:gen} to conclude the theorem holds.
\end{proof}

\subsubsection{JA}\label{app:ja-proof}

Finally, we prove \autoref{cor:ja}.
\corja*

\begin{proof}
We use $\elim{a}$ as shorthand for $\{b \in \actionset \mid a \nconc b\}$ from here on.
We argue this assignment to $\formonsym$, $\formoffsym$ and $\actelimsym$ satisfies the four properties of finitely realisable predicates.

\begin{enumerate}
    \item The invariant property: we need to show that an arbitrary path $\pi$ satisfies JA exactly when it satisfies the condition that for all states $s$ of $\pi$ and all $a \in \nonblock$, if $s$ satisfies $\diam{a}\tp$ then in the suffix of $\pi$ starting in $s$, there must be a state that satisfies $\fp$ or an occurrence of an action in $\elim{a}$. We refer to this latter condition as $P$.
    Justness requires that for every state in $\pi$, every action that is enabled must be eliminated in the suffix of $\pi$ starting in $s$.
    The correspondence between $P$ and $JA$ follows directly from the observation that there can never be a state that satisfies $\fp$, hence $P$ also requires that after a non-blocking action is enabled, it must subsequently be eliminated.

    \item The locking property: if a state is $\block$-locked, no non-blocking actions are enabled in that state. This is identical to the definition that $\diam{a}\tp$ is violated for all $a \in \nonblock$.

    \item The exclusive property: $\fp$ is violated in every state, hence this property is trivially satisfied.

    \item The persistent property: this follows directly from the second requirement on concurrency relations on actions.
\end{enumerate}
Hence, $\block$-JA is a finitely realisable predicate on paths.

That $\block$-JA is feasible was proven as \autoref{prop:ja-feasible}.
We apply \autoref{thm:gen} to conclude the theorem holds.
\end{proof}

\subsection{Proof of SFA and SHFA Formulae}\label{app:sgen-proof}
We here prove \autoref{thm:sfa} and \autoref{thm:shfa}.
We prove both by first proving a more generic theorem.
For this, we introduce an assumption that generalises both strong fairness and strong hyperfairness, which we name $P_S$.
$P_S$ depends on a mapping $\formdissym$ from non-blocking actions to closed modal $\mu$-calculus formulae.
The assumption $P_S$ is then as follows: a path $\pi$ satisfies $P_S$ if for every non-blocking action $a$ , it holds that $a$ occurs infinitely often on $\pi$ or there is a suffix of $\pi$ on which $\formdis{a}$ is perpetually satisfied.
In line with previous definitions, ``perpetually satisfied'' means satisfied in every state.
\begin{theorem}\label{thm:sgen}
    An LTS satisfies \autoref{mucalc:sgen} if, and only if, its initial state does not admit a $\block$-progressing path that is $(\dpre,\ddis,\den)$-violating, and satisfies $P_S$,
    for a given choice of $\formdissym$ such that: $P_S$ is feasible, and a state is $\block$-locked if, and only if, it satisfies $\formdis{b}$ for all $b \in \nonblock$.
\end{theorem}
Here \autoref{mucalc:sgen} is as follows:
\begin{equation}\label{mucalc:sgen}
        \neg\diam{\dpre\co\clos{\comp{\ddis}}} ( \diam{\den}\tp \lor \boxm{\nonblock}\fp \lor
        \bigvee_{\emptyset \neq F \subseteq \nonblock}\nu X. ( \bigwedge_{a \in F}\mu W. ( ( \bigwedge_{b \in \nonblock\setminus F}\formdis{b} ) \land ( \diam{a \setminus\ddis}X \lor \diam{\comp{\ddis}}W ) ) ) )
\end{equation}

We once again fix $\block, \dpre,\ddis$ en $\den$.
We also fix an arbitrary mapping $\formdissym$ such that $P_S$ is feasible, and a state is $\block$-locked if, and only if, it satisfies $\formdis{b}$ for all $b \in \nonblock$. 
We first have to prove some supporting propositions.
\begin{proposition}\label{prop:sgen-to-progress}
    Every path $\pi$ that satisfies $P_S$ is also $\block$-progressing
\end{proposition}
\begin{proof}
    Let $\pi$ be an arbitrary path. If $\pi$ is infinite, it is trivially $\block$-progressing. We hence assume it is finite.
    Towards a contradiction, assume $\pi$ is not $\block$-progressing. Let $s$ be the final state of $\pi$, it must by our assumption not be a $\block$-locked state. 
    Then by our restriction on $\formdissym$, there must exist an action $b \in \nonblock$ such that $s$ does not satisfy $\formdis{b}$.
    Since $\pi$ is finite, no actions can occur infinitely often, so $b$ does not occur infinitely often.
    Since $\pi$ satisfies $P_S$, it must be the case that there is a suffix of $\pi$ on which $\formdis{b}$ is perpetually satisfied.
    But it is not satisfied in the final state of $\pi$, which is in every suffix of $\pi$.
    Hence, we have a contradiction and we conclude that $s$ is a $\block$-locked state, and so $\pi$ satisfies $\block$-progress.
\end{proof}
\begin{proposition}\label{prop:sgen-finite}
    Every $\block$-progressing and finite path satisfies $P_S$.
\end{proposition}
\begin{proof}
    Let $\pi$ be an arbitrary $\block$-progressing and finite path. Since $\pi$ is finite and $\block$-progressing, it has a final state $s$ and all actions enabled in $s$ are blocking. 
    Since $s$ is a $\block$-locked state, by our choice of $\formdissym$ we know $\formdis{b}$ is satisfied in $s$ for all $b \in \nonblock$. 
    The suffix of $\pi$ consisting only of $s$ is a suffix in which $\formdis{b}$ is perpetually satisfied for all $b \in \nonblock$, so $\pi$ satisfies $P_S$.
\end{proof}

\begin{proposition}\label{prop:sgen-prepend}
    Let $\pi$ be an path that satisfies $P_S$, then every path $\pi'$ of which $\pi$ is a suffix also satisfies $P_S$.
\end{proposition}
\begin{proof}
    Let $\pi'$ be an arbitrary path which contains a suffix $\pi$ such that $\pi$ satisfies $P_S$. We prove $\pi'$ satisfies $P_S$ as well. 
    To this end, we prove that every non-blocking action $b \in \nonblock$ occurs infinitely often in $\pi'$ or $\pi'$ has a suffix on which  $\formdis{b}$ is perpetually satisfied.
    Let $b$ be an arbitrary non-blocking action in $\nonblock$.
    Then since $\pi$ satisfies $P_S$, $b$ occurs infinitely often in $\pi$ or $\pi$ has a suffix on which $\formdis{b}$ is perpetually satisfied.
    Both qualities are inherited by any path of which $\pi$ is a suffix, and hence $\pi'$ also satisfies at least one of the two conditions for $b$.
    We conclude that $\pi'$ satisfies $P_S$.
\end{proof}

\begin{proposition}\label{prop:sgen-suffix}
    Let $\pi$ be path that satisfies $P_S$, then every suffix $\pi'$ of $\pi$ also satisfies $P_S$.
\end{proposition}
\begin{proof}
    Let $\pi$ be an arbitrary path that satisfies $P_S$ and let $\pi'$ be an arbitrary suffix of that path. We prove $\pi'$ satisfies $P_S$.
    To this end, let $b$ be an arbitrary non-blocking action in $\nonblock$.
    We show that $b$ occurs infinitely often in $\pi'$ or that $\pi'$ has a suffix on which $\formdis{b}$ is perpetually satisfied.
    Since $\pi$ satisfies $P_S$, we know that at least one of the two cases is true for $\pi$.
    \begin{itemize}
        \item If $b$ occurs infinitely often in $\pi$, then it also occurs infinitely often in $\pi'$.
        \item If $\pi$ has a suffix $\pi''$ on which $\formdis{b}$ is perpetually satisfied, then $\pi''$ is a suffix of $\pi'$ as well, or a suffix of $\pi''$ is also a suffix of $\pi'$. Either way, $\pi'$ also has a suffix on which $\formdis{b}$ is perpetually satisfied.
    \end{itemize}
    We conclude that $\pi'$ satisfies $P_S$.
\end{proof}

We now prove the correctness of the formula.
We first split the formula into multiple sub-formulae.

\begin{align*}
    \mathit{violate_S} &= \diam{\dpre}\mathit{finite_S}\\
    \mathit{finite_S} &= \diam{\clos{\comp{\ddis}}}(\diam{\den}\tp \lor \boxm{\nonblock}\fp \lor \mathit{disjunct_S})\\
    \mathit{disjuct_S} &= \bigvee_{\emptyset \neq F \subseteq \nonblock}\mathit{invariant_S}(F)\\
    \mathit{invariant_S}(F) &= \nu X. ( \bigwedge_{a \in F}\mathit{eliminate_S}(F, a) ) \\
    \mathit{eliminate_S}(F, a) &= \mu W. ( ( \bigwedge_{b \in \nonblock\setminus F} \formdis{b} ) \land (\diam{a \setminus \ddis}X \lor \diam{\comp{\ddis}}W) )
\end{align*}
We have that \autoref{mucalc:sgen} $= \neg\mathit{violate_S}$.

We characterise the semantics of the formula bottom-up, starting with $\mathit{eliminate_S}$.
\begin{lemma}\label{lem:sgen-sat}
    For all environments $\env$, states $s \in \states$, sets $F \subseteq \nonblock, \mathcal{F} \subseteq \states$ and non-blocking actions $a \in F$, it is the case that $s \in \llbracket \mathit{eliminate_S}(F, a) \rrbracket_{\env [X := \mathcal{F}]}$ if, and only if, $s$ admits a finite path $\pi$ with final state $s_{\mathit{final}}$ satisfying the following requirements:
    \begin{enumerate}
        \item $\pi$ is $\ddis$-free, and
        \item all states of $\pi$, with the possible exception of $s_{\mathit{final}}$, satisfy $\formdis{b}$ for all $b \in \nonblock\setminus F$, and
        \item $s_{\mathit{final}}$ is in $\mathcal{F}$, and
        \item the final transition of $\pi$, $t_{\mathit{final}}$ is labelled with $a$.
    \end{enumerate}
\end{lemma}
\begin{proof}
    Let $s$ be an arbitrary state, $F \subseteq \nonblock$ and $\mathcal{F} \subseteq \states$ arbitrary sets and $a$ an arbitrary action in $F$, $a$ is by definition of $F$ non-blocking. We prove $s$ is in the semantics of $\mathit{eliminate_{S}}(F, a)$ under environment $\env[X := \mathcal{F}]$ if, and only if, it admits a path $\pi$ meeting conditions 1 through 4. 

    Note that $\formdis{b}$ does not depend on $Y$ regardless of the choice of $b$, since we require $\formdissym$ to map to closed modal $\mu$-calculus formulae. It is also that case that $\diam{a \setminus \ddis}X$ does not depend on $Y$. Hence, we can apply \autoref{prop:lfp-basic} to conclude that $s \in \llbracket \mathit{eliminate_{S}}(F, a) \rrbracket_{\env[X := \mathcal{F}]}$ if, and only if, $s$ admits a finite path $\pi'$ meeting the following conditions:
    \begin{enumerate}
    \setcounter{enumi}{4}
        \item all actions occurring in $\pi'$ are in $\comp{\ddis}$, and
        \item all states in $\pi'$ satisfy $\left\llbracket \bigwedge_{b \in \nonblock\setminus F} \formdis{b} \right\rrbracket_{\env [X := \mathcal{F}]}$, and
        \item the final state of $\pi'$ is in $\llbracket \diam{a \setminus \ddis}X \rrbracket_{\env [X := \mathcal{F}]}$.
    \end{enumerate}
    We prove the claim by showing that if $s$ admits a path $\pi$ satisfying conditions 1 through 4, it also admits a path $\pi'$ satisfying conditions 5 through 7, and vice versa.
    \begin{itemize}
        \item If $s$ admits a path $\pi$ satisfying conditions 1 through 4, then we know $\pi$'s last transition is labelled with $a$. Let $\pi'$ be the prefix of $\pi$ missing only the last transition and final state. Then because $\pi$ is $\ddis$-free (by 1), so is $\pi'$ (satisfying 5). Additionally, because all states on $\pi$ with the possible exception of $s_{\mathit{final}}$ are in $\left\llbracket \bigwedge_{b \in \nonblock \setminus F} \formdis{b} \right\rrbracket_{\env [X := \mathcal{F}]}$ (by 2), and $\pi'$ has all states of $\pi$ except $s_{\mathit{final}}$, all states in $\pi'$ are in this set as well (satisfying 6). 
        Finally, the final state of $\pi'$, $s_{\mathit{final}}'$, admits a transition labelled with $a$ to $s_{\mathit{final}}$, namely $t_{\mathit{final}}$ (by 4), and $s_{\mathit{final}} \in \mathcal{F}$ (by 3). Hence, $s_{\mathit{final}}' \in \llbracket \diam{a}X\rrbracket_{\env [X := \mathcal{F}]}$. Since $a$ occurs in $\pi$ (by 4) and $\pi$ is $\ddis$-free (by 1), we know $a \not\in \ddis$. Therefore, $s_{\mathit{final}}' \in \llbracket \diam{a \setminus \ddis}X\rrbracket_{\env [X := \mathcal{F}]}$ (satisfying 7). We have shown $s$ admits a path $\pi'$ satisfying conditions 5 through 7.
        \item If $s$ admits a path $\pi'$ satisfying conditions 5 through 7, then we know that the final state of $\pi'$, $s_{\mathit{final}}'$,  admits an $a$ transition to a state in $\mathcal{F}$, and that $a \not\in \ddis$ (by 7). Let $t_{\mathit{final}}$ be such a transition, and let $s_{\mathit{final}}$ be the target of $t_{\mathit{final}}$. We take $\pi = \pi't_{\mathit{final}} s_{\mathit{final}}$. Since $\pi'$ is $\ddis$-free (by 5) and the only new transition is also not labelled with an action in $\ddis$, $\pi$ is $\ddis$-free (satisfying 1). Additionally, all states on $\pi'$ satisfy $\formdis{b}$ for all $b \in \nonblock \setminus F$ (by 6) and so this also holds for all states on $\pi$ with the possible exception of $s_{\mathit{final}}$ (satisfying 2). Finally, by construction $s_{\mathit{final}} \in \mathcal{F}$ and $\action{t_{\mathit{final}}} = a$ (satisfying 3 and 4). We conclude that $s$ admits a path $\pi$ meeting conditions 1 through 4.
    \end{itemize}
    By \autoref{prop:lfp-basic}, we have proven that $s \in \llbracket \mathit{eliminate_{S}}(F, a) \rrbracket_{\env [X := \mathcal{F}]}$ if, and only if, $s$ admits a path satisfying conditions 1 through 4.    
\end{proof}

Next, we want to establish the semantics of $\mathit{invariant_{S}}(F)$. 
We here fix an arbitrary $F \subseteq \nonblock$ such that $F$ is not the empty set.
Let $S_{\mathit{inf}(F)}$ be the set of all states that admit infinite, $\ddis$-free paths along which all actions in $F$ occur infinitely often and $\formdis{b}$ is perpetually satisfied for all $b \in \nonblock \setminus F$.
We prove that $S_{\mathit{inf}(F)}$ is a fixed point of the transformer characterised by $\mathit{invariant_{S}}(F)$.
We will subsequently prove that it is the greatest fixed point.
\begin{lemma}\label{lem:sgen-inv-fix}
    $S_{\mathit{inf}(F)}$ is a fixed point of the transformer $T_{\mathit{inf}(F)}$ defined by:
    $$T_{\mathit{inf}(F)}(\mathcal{F}) = \bigcap_{a \in F} \left\{ s \in \states \mid s \in \left\llbracket \mathit{eliminate_{S}}(F, a) \right\rrbracket_{\env [X := \mathcal{F}]}\right\}$$
    for arbitrary environment $\env$.
\end{lemma}
\begin{proof}
We prove $S_{\mathit{inf}(F)}$ is a fixed point of $T_{\mathit{inf}(F)}$ through mutual set inclusion of $S_{\mathit{inf}(F)}$ and $T_{\mathit{inf}(F)}(S_{\mathit{inf}(F)})$. 
\begin{itemize}
    \item Let $s$ be an arbitrary state in $S_{\mathit{inf}(F)}$, we prove that $s \in T_{\mathit{inf}(F)}(S_{\mathit{inf}(F)})$. Since $s \in S_{\mathit{inf}(F)}$, we know that $s$ admits an infinite, $\ddis$-free path $\pi$ such that all actions in $F$ occur infinitely often in $\pi$ and $\formdis{b}$ is perpetually satisfied in $\pi$ for all $b \in \nonblock \setminus F$. To prove $s \in T_{\mathit{inf}(F)}(S_{\mathit{inf}(F)})$, we need to show that for all $a \in F$, it is the case that $s \in \llbracket \mathit{eliminate_{S}}(F, a) \rrbracket_{\env[X := S_{\mathit{inf}(F)}]}$. To this end, let $a$ be an arbitrary action in $F$. We now such an action exists since we assumed $F \neq \emptyset$. We prove $s \in \llbracket \mathit{eliminate_{S}}(F, a) \rrbracket_{\env[X := S_{\mathit{inf}(F)}]}$. By \autoref{lem:sgen-sat}, it suffices to prove that $s$ admits a finite path $\pi'$ with final state $s_{\mathit{final}}'$ and final transition $t_{\mathit{final}}'$ such that $\pi'$ is $\ddis$-free, $\formdis{b}$ is satisfied for all $b \in \nonblock \setminus F$ in all states of $\pi'$ (with the possible exception of $s_{\mathit{final}}'$), $s_{\mathit{final}}' \in S_{\mathit{inf}(S)}$ and $\action{t_{\mathit{final}}'} = a$.
    Since $a \in F$, $a$ is non-blocking and occurs infinitely often in $\pi$. Let $t_{\mathit{final}}'$ be the first transition of $\pi$ labelled with $a$, and let $s_{\mathit{final}}' = \target{t_{\mathit{final}}'}$. Let $\pi'$ be the prefix of $\pi$ ending in $s_{\mathit{final}}'$. Note that since $\pi$ is $\ddis$-free, so is $\pi'$. Additionally, since $\formdis{b}$ is satisfied for all $b \in \nonblock \setminus F$ in all states of $\pi$, this holds for all states of $\pi'$ as well. By construction, $\action{t_{\mathit{final}}'} = a$. It remains to prove that $s_{\mathit{final}}' \in S_{\mathit{inf}(F)}$.
    Let $\pi''$ be the suffix of $\pi$ starting in $s_{\mathit{final}}'$. This suffix inherits many properties from $\pi$, particularly that it is infinite, $\ddis$-free, that all actions in $F$ occur infinitely often and that $\formdis{b}$ is perpetually satisfied for all non-blocking actions not in $F$. Hence, $\pi''$ witnesses $s_{\mathit{final}}' \in S_{\mathit{inf}(F)}$. We conclude that $s \in \llbracket \mathit{eliminate_{S}}(F, a) \rrbracket_{\env [X := S_{\mathit{inf}(F)}]}$ for all $a \in F$ and hence $s \in T_{\mathit{inf}(F)}(S_{\mathit{inf}(F)})$.

    \item Let $s$ be an arbitrary state in $T_{\mathit{inf}(F)}(S_{\mathit{inf}(F)})$, we prove $s \in S_{\mathit{inf}(F)}$. From $s\in T_{\mathit{inf}(F)}(S_{\mathit{inf}(F)})$, we conclude that for all $a \in F$, $s \in \llbracket \mathit{eliminate_{S}}(F, a)\rrbracket_{\env[X := S_{\mathit{inf}(F)}]}$. To prove $s \in S_{\mathit{inf}(F)}$, we construct a path $\pi$ from $s$ that is infinite, $\ddis$-free, and on which all actions in $F$ occur infinitely often and $\formdis{b}$ is perpetually satisfied for all $b \in \nonblock \setminus F$. To this end, let $a$ be an arbitrary action in $F$, which is guaranteed to exist because $F \neq \emptyset$. We will use $s \in \llbracket \mathit{eliminate_{S}}(F, a)\rrbracket_{\env[X := S_{\mathit{inf}(F)}}$ to construct $\pi$. By \autoref{lem:sgen-sat}, $s$ admits a finite path $\pi'$ with a final state $s_{\mathit{final}}'$ and a final transition $t_{\mathit{final}}'$ such that $\pi'$ is $\ddis$-free, $\formdissym$ is satisfied for all non-blocking actions not in $F$ in all states of $\pi'$ with the possible exception of $s_{\mathit{final}}'$, $s_{\mathit{final}}' \in S_{\mathit{inf}(F)}$ and $\action{t_{\mathit{final}}'} = a$. From $s_{\mathit{final}}' \in S_{\mathit{inf}(F)}$, we know that $s_{\mathit{final}}'$ admits an infinite, $\ddis$-free path $\pi''$ such that all actions in $F$ occur infinitely often in $\pi''$ and $\formdissym$ is perpetually satisfied for all non-blocking actions not in $F$ in $\pi''$. Note that since $s_{\mathit{final}}'$ is on $\pi''$, $\formdissym$ is satisfied for all non-blocking actions not in $F$ in $s_{\mathit{final}}'$, even though we did not get this fact directly from \autoref{lem:sgen-sat}. Let $\pi = \pi' \co \pi''$. Since $\pi''$ is infinite and all actions in $F$ occur infinitely often in $\pi''$, both qualities hold for $\pi$ as well. Additionally, since both $\pi'$ and $\pi''$ are $\ddis$-free, so is $\pi$. It remains to show that $\formoffsym$ is perpetually satisfied for all non-blocking actions not in $F$ on $\pi$. We got from \autoref{lem:sgen-sat} that all states on $\pi'$ satisfy this condition, except possibly $s_{\mathit{final}}'$. But from $s_{\mathit{final}}' \in S_{\mathit{inf}(F)}$, we got that this holds for $s_{\mathit{final}}'$ as well. It also holds for all states on $\pi''$. Hence, the condition is satisfied for all states of $\pi$. Using $\pi$ as a witness, we conclude $s \in S_{\mathit{inf}(F)}$.
\end{itemize}
    We conclude that $S_{\mathit{inf}(F)} = T_{\mathit{inf}(F)}(S_{\mathit{inf}(F)})$.
\end{proof}

\begin{lemma}\label{lem:sgen-inv-greatest}
    The set $S_{\mathit{inf}(F)}$ is the greatest fixed point of the transformer defined in \autoref{lem:sgen-inv-fix}.
\end{lemma}
\begin{proof}
    For this proof, we use \autoref{lem:sgen-inv-fix}. Let $\mathcal{F}$ be an arbitrary fixed point of $T_{\mathit{inf}(F)}$, we prove $\mathcal{F} \subseteq S_{\mathit{inf}(F)}$. 
    To this end, let $s$ be an arbitrary state in $\mathcal{F}$, we prove $s \in S_{\mathit{inf}(F)}$ by constructing a path $\pi$ from $s$ that is infinite, $\ddis$-free, on which all actions in $F$ occur infinitely often and $\formdis{b}$ is perpetually satisfied for all $b \in \nonblock \setminus F$.

    For this construction, let $Q$ be a queue of actions.
    This construction consists of an infinite loop. We start the loop with $\pi = s$ and $Q$ containing exactly one copy of every action in $F$ in arbitrary order. Note that $s \in \mathcal{F}$. The loop as several invariants
    \begin{itemize}
        \item the final state of $\pi$ is in $\mathcal{F}$.
        \item $Q$ contains exactly one copy of every action in $F$ and nothing else.
        \item $\pi$ constructed so far is $\ddis$-free.
    \end{itemize}

    In every iteration of the loop, let $s'$ be the final state of $\pi$ constructed so far. Let $a$ be the action at the head of $Q$.  Since $F \neq \emptyset$, and $Q$ always contains a copy of every action in $F$, such an $a$ always exists. We pop $a$ from $Q$ and immediately re-insert it at the end of the queue. An invariant of this construction is that $s'$ is always in $\mathcal{F}$. Because of this, and since $\mathcal{F}$ is a fixed point of $T_{\mathit{inf}(F)}$, we have that $s' \in T_{\mathit{inf}(F)}(\mathcal{F})$. From this we conclude that for all actions $a' \in F$, $s' \in \llbracket \mathit{eliminate_{S}}(F, a')\rrbracket_{\env [X := \mathcal{F}]}$. The action $a$ came from $Q$, and hence is an action in $F$. Therefore, $s' \in \llbracket \mathit{eliminate_{S}}(F, a)\rrbracket_{\env[X:= \mathcal{F}]}$. By \autoref{lem:sgen-sat}, $s'$ therefore admits a path $\pi'$ meeting the following conditions:
    \begin{enumerate}
        \item $\pi'$ is $\ddis$-free, and
        \item all states of $\pi'$, with the possible exception of its final state, satisfy $\formdis{b}$ for all $b \in \nonblock \setminus F$, and
        \item the final state of $\pi'$ is in $\mathcal{F}$, and
        \item the final transition of $\pi'$ is labelled with $a$.
    \end{enumerate}
    We append $\pi'$ to $\pi$, and then start the loop again.
    Note that the new final state of $\pi$ is in $\mathcal{F}$ by construction, and we the only modification of $Q$ is removing and then re-inserting $a$, so $Q$ still contains exactly one copy of every action in $F$ and nothing else.
    Since we added a $\ddis$-free segment to $\pi$, $\pi$ is still $\ddis$-free. Hence, all the invariants still hold.

    This loop will never terminate.
    From this, and the invariants, it immediately follows that the path $\pi$ we construct is infinite and $\ddis$-free. It remains to argue that all actions in $F$ occur infinitely often in $\pi$ and that $\formdissym$ is perpetually satisfied for all non-blocking actions not in $F$.
    Since $Q$ is a queue, and we always pop the top action and then insert it back at the end, and $F \subseteq \nonblock \subseteq \actionset$ and $\actionset$ is finite, the construction of $\pi$ ensures that every action that is initially in $Q$ is selected infinitely often to be $a$. We then add a segment to $\pi$ where $a$ occurs. Hence, every action that is initially in $Q$ occurs infinitely often in $\pi$. We initialised $Q$ with all the actions in $F$, hence every action in $F$ occurs infinitely often in $\pi$.
    At first glance, it may seem that this construction does not guarantee the second condition. After all, every path segment we add only ensures that all states except the last state satisfy $\formdis{b}$ for all $b \in \nonblock \setminus F$. Note however, that these segments do guarantee that the final state is in the fixed point $\mathcal{F}$ again, and therefore that the final state once again admits a segment that meets the condition that all states, with the possible exception of the last, satisfy this condition. Additionally, the added path always contains at least one transition, so the initial state of the segment is not the final state. Hence, whenever we add a segment $\pi'$, the fact that we can add a next segment guarantees that $\formdissym$ is satisfied for all non-blocking actions not in $F$ in the final state of $\pi'$. And since we infinitely often can add a next segment, we conclude that every state on $\pi$ satisfies the condition.

    We have constructed a path $\pi$ witnessing $s \in S_{\mathit{inf}(F)}$, hence $S_{\mathit{inf}(F)}$ is the greatest fixed point of $T_{\mathit{inf}(F)}$.
\end{proof}
Since the semantics of $\mathit{invariant_{S}}(F)$ are exactly the greatest fixed point of $T_{\mathit{inf}(F)}$, we can conclude the following.
\begin{corollary}\label{cor:sgen-inv}
    The set of states characterised by $\mathit{invariant_{S}}(F)$ is exactly the set of states that admit infinite, $\ddis$-free paths on which all actions in $F$ occur infinitely often and $\formdis{b}$ is perpetually satisfied for all $b \in \nonblock \setminus F$. This holds true whenever $F \subseteq \nonblock$ and $F \neq \emptyset$.
\end{corollary}

We now cover the semantics of $\mathit{disjunct_S}$.
\begin{lemma}\label{lem:sgen-disjunct}
    For all environments $\env$ and  $s \in \states$, $s \in \llbracket \mathit{disjunct_{S}} \rrbracket_{\env}$ if, and only if, there exists a subset $F$ of $\nonblock$ with $F \neq \emptyset$, such that $s$ admits an infinite, $\ddis$-free path on which all actions in $F$ occur infinitely often and $\formdis{b}$ is perpetually satisfied for all $b \in \nonblock \setminus F$.
\end{lemma}
\begin{proof}
    This follows trivially from the definition of $\mathit{disjunct_{S}}$, \autoref{cor:sgen-inv} and the definition of the modal $\mu$-calculus semantics.
\end{proof}
We next cover $\mathit{finite_S}$.

\begin{lemma}\label{lem:sgen-finite}
    For all environments $\env$ and states $s \in \states$, $s \in \llbracket \mathit{finite_{S}} \rrbracket_{\env}$ if, and only if, $s$ admits a complete path $\pi$ that is $(\emptyseq,\ddis,\den)$-violating and satisfied $P_S$.
\end{lemma}
\begin{proof}
    Note that $\mathit{finite_S}$ is equivalent to $\mu Y.(\diam{\den}\tp \lor \boxm{\nonblock}\fp \lor \mathit{disjunct_S} \lor \diam{\comp{\ddis}}Y)$, we merely wrote it more compactly using regular formulae.

    We wish to apply \autoref{prop:lfp-basic} to determine the semantics of this formula. 
    To be able to do this, we observe that $\diam{\den}\tp \lor \boxm{\nonblock}\fp \lor \mathit{disjunct_{S}}$ does not depend on $Y$. Hence, we can apply \autoref{prop:lfp-basic} (with $\phi_1 = \tp$) to conclude that $s$ is in $\llbracket \mathit{finite_{S}} \rrbracket_{\env}$ if, and only if, $s$ admits a finite path $\pi'$ that is $\ddis$-free and ends in a state satisfying $\llbracket \diam{\den}\tp \lor \boxm{\nonblock}\fp \lor \mathit{disjunct_{S}} \rrbracket_{\env}$.
    We use this information to prove the bi-implication.

    \begin{itemize}
        \item Assume $s \in \llbracket \mathit{finite_{S}} \rrbracket_{\env}$, then by \autoref{prop:lfp-basic}, we obtain a path $\pi'$ starting in $s$ that is $\ddis$-free and ends in a state $s'$ such that $ s' \in \llbracket \diam{\den}\tp \rrbracket_{\env}$, or $s' \in \llbracket \boxm{\nonblock}\fp \rrbracket_{\env}$, or $s' \in \llbracket \mathit{disjunct_{S}} \rrbracket_{\env}$. We construct the path $\pi$ using this information. For this, we do a case distinction on which set $s'$ is in.
        \begin{itemize}
            \item If $s' \in \llbracket \diam{\den}\tp \rrbracket_{\env}$, then $\pi'$ is a finite path with an action in $\den$ enabled in the final state. We extend $\pi'$ with a transition labelled with an action in $\den$. The resulting path can be extended to a path $\pi$ that satisfies $P_S$ because we assumed $P_S$ to be feasible and we know this path is also $\block$-progressing by \autoref{prop:sgen-to-progress}. Note that $\pi'$ is $\ddis$-free, and we immediately appended a transition labelled with an action in $\den$ to it, before we did the extension to a $\block$-progressing and complete path. Hence, the extended path $\pi$ is $\ddis$-free up until an occurrence of $\den$, and hence is also $\ddis$-free up until the first occurrence of such an action. We conclude that $\pi$ is $\block$-progressing, $\ddis$-free up until the first occurrence of an action in $\den$, and satisfies $P_S$.
            
            \item If $s' \in \llbracket \boxm{\nonblock}\fp \rrbracket_{\env}$, then $\pi'$ is a $\block$-progressing path ending in a blocking state. Since $\pi'$ is finite and $\block$-progressing, it satisfies $P_S$ according to \autoref{prop:sgen-finite}. Additionally, we know that $\pi'$ is entirely $\ddis$-free. We conclude that $s$ admits a path that is $\block$-progressing, entirely $\ddis$-free and satisfies $P_S$.
            
            \item If $s' \in \llbracket \mathit{disjunct_{S}}\rrbracket_{\env}$ then by \autoref{lem:sgen-disjunct}, there exists a subset $F$ of $\nonblock$ with $F \neq \emptyset$ such that $s'$ admits an infinite, $\ddis$-free path $\pi''$ on which all actions in $F$ occur infinitely often and $\formdis{b}$ is perpetually satisfied for all $b \in \nonblock \setminus F$. Let $F$ be a subset meeting this condition, and let $\pi''$ be the path we obtain from $F$. Then, let $\pi = \pi'\co \pi''$. Since $\pi''$ is infinite, so is $\pi$, hence $\pi$ is $\block$-progressing. Since both $\pi'$ and $\pi''$ are $\ddis$-free, so is $\pi$. It remains to show that $\pi$ satisfies $P_S$.

            To do this, we must show that for every $a \in \nonblock$, $a$ occurs infinitely often in $\pi$ or $\pi$ has a suffix on which $\formdis{a}$ is perpetually satisfied. 
            Let $a$ be an arbitrary non-blocking action. If $a \in F$, then since all actions in $F$ occur infinitely often on $\pi''$, $a$ occurs infinitely often on $\pi$. If $a \not\in F$, then $a \in \nonblock \setminus F$ and hence $\formdis{a}$ is perpetually satisfied on $\pi''$. Therefore, $\pi$ contains a suffix on which $\formdis{a}$ is perpetually satisfied.

            We have constructed a path $\pi$ that is $\block$-progressing, $\ddis$-free and satisfying $P_S$.
        \end{itemize}
        In all three cases, we can construct a $\block$-progressing path that is $(\emptyseq,\ddis,\den)$-violating and satisfies $P_S$.
        
        \item Assume that $s$ admits a path $\pi$ that is $\block$-progressing, $(\emptyseq, \ddis, \den)$-violating and satisfying $P_S$. We prove $s \in \llbracket \mathit{finite_{S}} \rrbracket_{\env}$. By \autoref{prop:lfp-basic}, it is sufficient to show that $s$ admits a finite path $\pi'$ that is $\ddis$-free and ends in a state $s'$ that is in $\llbracket \diam{\den}\tp \rrbracket_{\env}$, or in $\llbracket \boxm{\nonblock}\fp \rrbracket_{\env}$, or in $\llbracket \mathit{disjunct_{S}}\rrbracket_{\env}$.
        We do a case distinction on whether there are any occurrences of actions in $\den$ in $\pi$.
        \begin{itemize}
            \item If there are, then let $t_{\mathit{en}}$ be the first transition in $\pi$ that is labelled with an action in $\den$. Let $\pi'$ be the prefix of $\pi$ ending in $\source{t_{\mathit{en}}}$. Since $\pi$ is $\ddis$-free up until the first occurrence of an action in $\den$, $\pi'$ is $\ddis$-free. Additionally, because $t_{\mathit{en}}$ is enabled in the final state of $\pi'$, this state is in $\llbracket \diam{\den}\tp\rrbracket_{\env}$. We can therefore conclude $s \in \llbracket \mathit{finite_{S}} \rrbracket_{\env}$.
            
            \item If there are not, then $\pi$ is entirely $\ddis$-free. We do a second case distinction on whether $\pi$ is finite or not.
            \begin{itemize}
                \item If $\pi$ is finite, then the final state of $\pi$ is a $\block$-locked state. Hence, $\pi$ is an entirely $\ddis$-free path ending in a state in $\llbracket \boxm{\nonblock}\fp \rrbracket_{\env}$. We conclude that $s \in \llbracket \mathit{finite_{S}}\rrbracket_{\env}$.
                
                \item If $\pi$ is infinite we need a more extensive argument. Recall that $\pi$ satisfies $P_S$.
                Hence all non-blocking actions $a$ occur infinitely often on $\pi$ or there is a suffix of $\pi$ on which $\formdis{a}$ is perpetually satisfied.
                Let $F$ be the set of all non-blocking actions that occur infinitely often on $\pi$. 
                If $F = \emptyset$, then it must be the case that $\pi$ has a suffix $\pi''$ on which $\formdis{a}$ is perpetually satisfied for all $a \in \nonblock$. Consider that we assumed that $P_S$ is such that a state is $\block$-locked if, and only if, it satisfies $\formdis{b}$ for all $b \in \nonblock$. Therefore, every state on $\pi''$ is a $\block$-locked state.
                Any such state is in $\llbracket \boxm{\nonblock}\fp\rrbracket_{\env}$. Let $\pi'$ be the finite prefix of $\pi$ ending in the first state of $\pi''$, then $\pi'$ is $\ddis$-free because $\pi$ is and ends in a state in $\llbracket \boxm{\nonblock}\fp\rrbracket_{\env}$. Since $s$ admits $\pi'$, $s \in \llbracket \mathit{finite}_S\rrbracket_{\env}$.

                If $F \neq \emptyset$, then $s$ admits a path $\pi$ that is infinite, $\ddis$-free and on which all actions in $F$ occur infinitely often and that has a suffix, $\pi''$ on which $\formoff{b}$ is perpetually satisfied for all $b \in \nonblock \setminus F$. Let $s'$ be the first state of $\pi''$. 
                Since $s'$ admits $\pi''$, it admits an infinite, $\ddis$-free path on which all actions in $F$ occur infinitely often and $\formoff{b}$ is perpetually satisfied for all $b \in \nonblock \setminus F$. Hence, by \autoref{lem:sgen-disjunct}, $s' \in \llbracket \mathit{disjunct_{S}}\rrbracket_{\env}$. 
                Let $\pi'$ be the path such that $\pi = \pi' \co \pi''$. Then $\pi'$ is a finite path from $s$ that is $\ddis$-free and which ends in a state that is in $\llbracket \mathit{disjunct_{S}}\rrbracket_{\env}$. Therefore, $s \in \llbracket \mathit{finite_{S}}\rrbracket_{\env}$.
            \end{itemize}
        \end{itemize}
        We have proven that in all cases, $s \in \llbracket{finite_{S}}\rrbracket_{\env}$.
    \end{itemize}
    We have proven both sides of the bi-implication.
\end{proof}

The final two steps of the general proof are largely trivial.
\begin{lemma}\label{lem:violate-sgen}
    For all environments $\env$ and states $s \in \states$, it holds that $s \in \llbracket \mathit{violate_S} \rrbracket_{\env}$ if, and only if, $s$ admits a path that is $\block$-progressing, satisfies $P_S$ and is $(\dpre, \ddis, \den)$-violating.
\end{lemma}
\begin{proof}
    It follows directly from the definition of the diamond operator that $\mathit{violate_S}$ characterises those states that admit a path $\pi$ that has a prefix matching $\dpre$ that ends in some state $s'$ in $\llbracket \mathit{finite_S} \rrbracket_{\env}$. By \autoref{lem:sgen-finite}, $s'$ admits a $\block$-progressing path $\pi'$ that satisfies $P_S$ and is $(\emptyseq,\ddis,\den)$-violating. By \autoref{prop:sgen-prepend}, prepending $\pi$ to $\pi'$ results in a path that still satisfies $P_S$. It is also still $\block$-progressing, and by construction, it has a prefix matching $\dpre$, namely $\pi$, after which it is $\ddis$-free up until the first occurrence of an action in $\den$. Hence, it is $(\dpre, \ddis, \den)$-violating.
\end{proof}
\autoref{mucalc:sgen} is the negation of $\mathit{violate}_S$, hence it expresses that a state does not admit such a path.
\autoref{thm:sgen} follows directly.

To be able to apply \autoref{thm:sgen} to SFA and SHFA, we need to prove the following: firstly, that both assumptions are feasible, this has already been proven as \autoref{prop:sfa-feasible} and \autoref{prop:shfa-feasible}. Secondly, that there exist choices of $\formdissym$ such that a state is $\block$-locked if, and only if, $\formdissym$ is satisfied for all non-blocking actions and such that the resulting $P_S$ is equivalent to SFA or SHFA respectively.

\subsubsection{SFA}\label{app:sfa-proof}
We first prove \autoref{thm:sfa}.
\thmsfa*

Where \autoref{mucalc:sfa} is:
\begin{equation*}
        \neg\diam{\dpre\co\clos{\comp{\ddis}}} ( \diam{\den}\tp \lor \boxm{\nonblock}\fp \lor
        \bigvee_{\emptyset \neq F \subseteq \nonblock}\nu X. ( \bigwedge_{a \in F}\mu W. ( ( \bigwedge_{b \in \nonblock\setminus F}\boxm{b}\fp ) \land ( \diam{a \setminus\ddis}X \lor \diam{\comp{\ddis}}W ) ) ) )
\end{equation*}
\begin{proof}
    To get \autoref{mucalc:sfa} from \autoref{mucalc:sgen}, we take $\formdis{b} = \boxm{b}\fp$ for all $b \in \nonblock$. This is a mapping to closed modal $\mu$-calculus formulae, and hence a valid choice of $\formdissym$. By definition of a $\block$-locked state, every $\block$-locked state satisfies $\formdis{b}$ for all $b \in \nonblock$, and vice versa.
    We show that a path $\pi$ satisfies $\block$-strong fairness of actions (\autoref{def:sfa}) exactly when it satisfies $P_S$ interpreted on this $\formdissym$. Note that in both theorems, it is included that the path must be $\block$-progressing so we take that as given.
    A path $\pi$ satisfies $\block$-SFA if, and only if, for every suffix $\pi'$ of $\pi$, every action $a \in \nonblock$ that is relentlessly enabled in $\pi'$ occurs in $\pi'$.

    First note that if an action $a$ is relentlessly enabled on a suffix $\pi'$ of $\pi$, then $a$ is also relentlessly enabled on $\pi$. After all, every suffix of $\pi$ is either also a suffix of $\pi'$ or contains $\pi'$, and in either case contains a state in which $a$ is enabled, since every suffix of $\pi'$ contains such a state.
    Similarly, if an action $a$ is relentlessly enabled on $\pi$, it is also relentlessly enabled on an arbitrary suffix $\pi'$ of $\pi$, since every suffix of $\pi'$ is also a suffix of $\pi$.
    Therefore, $\block$-SFA can be rewritten as follows: a path $\pi$ satisfies $\block$-SFA if, and only if, every action $a \in \nonblock$ that is relentlessly enabled in $\pi$ occurs in every suffix $\pi'$ of $\pi$. 

    If a path is $\block$-progressing and finite, then no non-blocking action can be relentlessly enabled on that path, since no non-blocking actions are enabled in the final state.
    If a path is infinite, then if an action occurs in every suffix of the path then it must occur infinitely often.
    Therefore, $\block$-SFA can be rewritten as follows: a path $\pi$ satisfies $\block$-SFA if, and only if, every action $a \in \nonblock$ that is relentlessly enabled in $\pi$ occurs infinitely often on $\pi$.   

    We need to make one last step to connect $\block$-SFA to $P_S$: if an action is not relentlessly enabled on some path, then the path must contain a suffix on which that action is never enabled. Hence, every non-blocking action is either relentlessly enabled on some path, or the path has a suffix on which the action is perpetually disabled.
    In other words, a path $\pi$ satisfies $\block$-SFA if, and only if, every action $a \in \nonblock$ either occurs infinitely often on $\pi$ or $\pi$ has a suffix on which $a$ is perpetually disabled.
    This is almost exactly the definition of $P_S$ when $\formdis{b} = \boxm{b}\fp$ for all $b \in \nonblock$, with the only exception being that $\block$-SFA has an exclusive or, and $P_S$ has a non-exclusive or. However, if an action occurs infinitely often then it cannot become perpetually disabled, and vice-versa. So in this case, the choice of $\formdissym$ enforces an exclusive or even though this is not part of the definition of $P_S$.
    Hence, the two definitions coincide.

    We have shown that $P_S$ when taking $\formdis{b} = \boxm{b}\fp$ for all $b \in \nonblock$ is equivalent to $\block$-strong fairness of actions. 
    To be able to apply \autoref{thm:sgen} to \autoref{mucalc:sfa}, we still need to show that, for this choice of $\formdissym$, $P_S$ is feasible.
    We get this from feasibility of $\block$-strong fairness of actions, as proven in \autoref{prop:sfa-feasible}.
    We therefore conclude that \autoref{thm:sgen} applies to \autoref{mucalc:sfa} and hence we conclude that \autoref{thm:sfa} is true.
\end{proof}

\subsubsection{SHFA}\label{app:shfa-proof}
We also prove \autoref{thm:shfa}. 
\thmshfa*

\autoref{mucalc:shfa} is:
\begin{equation*}
        \neg\diam{\dpre\co\clos{\comp{\ddis}}} ( \diam{\den}\tp \lor \boxm{\nonblock}\fp \lor
        \bigvee_{\emptyset \neq F \subseteq \nonblock}\nu X. ( \bigwedge_{a \in F}\mu W. ( ( \bigwedge_{b \in \nonblock\setminus F}\boxm{\clos{\nonblock} \co b}\fp ) \land ( \diam{a \setminus\ddis}X \lor \diam{\comp{\ddis}}W ) ) ) )
\end{equation*}
\begin{proof}
    To get \autoref{mucalc:shfa} from \autoref{mucalc:sgen}, we take $\formdis{b} = \boxm{\clos{\nonblock}\co b}\fp$ for all $b \in \nonblock$. This is a mapping to closed modal $\mu$-calculus formulae. We also need to show that in every $\block$-locked state, $\formdis{b}$ is satisfied for all $b \in \nonblock$, and vice versa. In a $\block$-locked state, all enabled actions are blocking. Hence, in order to reach a state where an non-blocking action is enabled, a blocking action must first be taken. Therefore, in a $\block$-locked state all non-blocking actions $b$ are not $\block$-reachable, and so $\boxm{\clos{\nonblock}\co b}\fp$ is satisfied.
    The other way around, if $\boxm{\clos{\nonblock}\co b}\fp$ is satisfied for all non-blocking actions $b$, then no non-blocking actions can be enabled since otherwise they would be $\block$-reachable. Hence, all enabled actions must be blocking and it must be a $\block$-locked state.
    
    We show a path satisfies $\block$-strong hyperfairness of actions (\autoref{def:shfa}) exactly when it satisfies $P_S$ interpreted on this $\formdissym$. We can use that the path is $\block$-progressing, since both the theorem on $\block$-SHFA and the theorem of $P_S$ include this.
    A path satisfies $\block$-SHFA if, and only if, for every suffix $\pi'$ of $\pi$, every action $a \in \nonblock$ that is relentlessly $\block$-reachable in $\pi'$ occurs in $\pi'$.

    Using a similar argument as given in the proof of \autoref{thm:sfa}, we claim that if an action is relentlessly $\block$-reachable on a suffix of a path it is also relentlessly $\block$-reachable on that whole path: if every suffix of a path $\pi$ has a state where the action $a$ is $\block$-reachable, then so does every suffix $\pi''$ of $\pi'$, since $\pi''$ is also a suffix of $\pi$. The other way around, if every suffix $\pi''$ of $\pi'$, which is a suffix of $\pi$, contains a state where the action $a$ is $\block$-reachable, then since every suffix of $\pi$ either contains $\pi'$ or is a suffix of $\pi'$, every suffix of $\pi$ also contains such a state.
    It is also the case for $\block$-SHFA that if a path is $\block$-progressing and finite, no non-blocking action can be relentlessly $\block$-reachable on that path, since in a $\block$-locked state no non-blocking action is $\block$-reachable.
    Additionally, on infinite paths every action that occurs in every suffix of the path must occur infinitely often.
    
    Using these insights, we can rewrite $\block$-SHFA to the following: a path $\pi$ satisfies $\block$-SHFA if, and only if, every action $a \in \nonblock$ that is relentlessly $\block$-reachable in $\pi$ occurs infinitely often in $\pi$. 

    If an action is not relentlessly $\block$-reachable on a path, then the path must contain a suffix in which the action is perpetually non-$\block$-reachable, and vice-versa.
    Hence, $\block$-SHFA requires that for every non-blocking action $a$, the action either occurs infinitely often on a path or the path has a suffix in which the action is perpetually non-$\block$-reachable.
    Similar to the $\block$-SFA case, the only difference between this definition and $P_S$ interpreted on $\formdis{b} = \boxm{\clos{\nonblock}\co b}\fp$ for all $b \in \nonblock$, is that the $\block$-SHFA definition uses an exclusive or and the $P_S$ definition does not.
    However, for this choice of $\formdissym$, this does not matter.
    If a non-blocking action occurs infinitely often it must be infinitely often enabled, and hence relentlessly $\block$-reachable on the path, thus there is no suffix of the path where the action in perpetually unreachable.
    In the other direction, if a path has a suffix on which a non-blocking action is perpetually not $\block$-reachable, then the action is only finitely often $\block$-reachable on the path and therefore also only finitely often enabled. Then it can also only occur finitely often.

    The final step to proving \autoref{thm:shfa} is to establish that $\block$-SHFA is feasible. This was proven in \autoref{prop:shfa-feasible}.
    Hence, we can apply \autoref{thm:sgen} to \autoref{mucalc:shfa} to conclude \autoref{thm:shfa} is true.    
\end{proof}

}
{
\newpage
\appendix
\setcounter{theorem}{45}

\section{Proof Sketch}\label{app:illustration}
Full proofs are included in the appendices of the full version of this paper.
Here, we provide an outline of the proof of \autoref{thm:gen} by presenting all the supporting lemmas and propositions proven. We include brief sketches of how we prove these claims, but not the full proofs. 
This is done to illustrate the approach we have taken.
This appendix corresponds to Appendix D.3 in the full version.
We begin with restating \autoref{thm:gen}.

\thmgen*

The following propositions\footnote{We use the same numbering here as in the full version, hence the jump to 46.}
give properties of finitely realisable paths that can be derived from the 
 invariant, locking, exclusive and persistent properties. 
\begin{proposition}\label{prop:gen-finite}
    Every $\block$-progressing, finite path satisfies every finitely realisable path predicate, for all $\block\subseteq\actionset$.   
\end{proposition}
\begin{proposition}\label{prop:gen-to-progress}
    Every path that satisfies a finitely realisable path predicate $P$ is $\block$-progressing.
\end{proposition}
\begin{proposition}\label{prop:gen-prepend}
    If a path $\pi$ satisfies finitely realisable path predicate $P$, then every path of which $\pi$ is a suffix also satisfies $P$.
\end{proposition}
\begin{proposition}\label{prop:gen-suffix}
    If a path $\pi$ satisfies a finitely realisable path predicate $P$, then every suffix of $\pi$ also satisfies $P$.
\end{proposition}

\autoref{prop:gen-finite} follows from the invariant, locking, and persistent properties.
For \autoref{prop:gen-to-progress}, we use the invariant, locking, and exclusive properties,
 and for \autoref{prop:gen-prepend}, the invariant and persistent property are enough.
 Finally, \autoref{prop:gen-suffix} follows directly from the invariant property.

For the proof of the main theorem, we fix an LTS $M$, as well as $\block$, $\dpre$, $\ddis$ and $\den$. We also fix a feasible, finitely realisable path predicate $P$.
To characterise the semantics of the formula, we first split it into multiple smaller subformulae.
\begin{align*}
    &\mathit{violate_G} &&= \diam{\dpre}\mathit{invariant_G}\\
    &\mathit{invariant_G} &&= \nu X. (\bigwedge_{a \in \nonblock}(\formon{a} \imps \mathit{eliminate_G}(a)))\\
    &\mathit{eliminate_G}(a) &&= \diam{\clos{\comp{\ddis}}}(\diam{\den}\tp \lor (\formoff{a} \land X) \lor \diam{\actelim{a} \setminus \ddis}X )
\end{align*}
We have that \autoref{mucalc:gen} $= \neg\mathit{violate_G}$.

The proof proceeds by characterising the semantics of every subformulae. We define the length of a path to be the number of transitions. A path of length 0 is called the empty path.
\begin{lemma}\label{lem:sat}
    For all environments $\env$, states $s \in \states$, actions $a \in \nonblock$ and sets $\mathcal{F} \subseteq \states$, it holds that $s \in \llbracket \mathit{eliminate_G(a)\rrbracket_{\env [X := \mathcal{F}]}}$ if, and only if, $s$ admits a finite path $\pi$ with final state $s_{\mathit{final}}$ that satisfies the following conditions:
    \begin{enumerate}
        \item $\pi$ is $\ddis$-free, and
        \item one of the following three holds:
        \begin{enumerate}
            \item at least one action in $\den$ is enabled in $s_{\mathit{final}}$, or
            \item $s_{\mathit{final}} \in \mathcal{F}$ and $s_{\mathit{final}}$ satisfies $\formoff{a}$, or
            \item $s_{\mathit{final}} \in \mathcal{F}$ and the last transition in $\pi$, $t_{\mathit{final}}$, is labelled with an action in $\actelim{a}\setminus\ddis$.
        \end{enumerate}
    \end{enumerate}
\end{lemma}

In the full proof of \autoref{lem:sat}, we use another supporting proposition that characterises the semantics of a simple least fixpoint formula that generalises $\mathit{eliminate_G(a)}$, and then show in detail how the lemma follows.
For this overview, we only give an intuitive explanation: the $\diam{\clos{\comp{\ddis}}}$ part of $\mathit{eliminate_{G}}(a)$ gives us the finite, $\ddis$-free path $\pi$ that the lemma refers to. The conditions on the final state of this path follow from the rest of the formula: $\diam{\den}\tp$ considers the possibility that an action in $\den$ is enabled; $\formoff{a} \land X$ represents reaching a state in $\mathcal{F}$ where $\formoff{a}$ is satisfied (recall that \autoref{lem:sat} refers to the environment where $X$ is mapped to $\mathcal{F}$); and $\diam{\actelim{a} \setminus \ddis}X$ appends one extra transition to a state in $\mathcal{F}$, eliminating $a$. Since the total path has to be $\ddis$-free, the eliminating action may not be in $\ddis$.

The next step is $\mathit{invariant_G}$. 
This formula exactly describes those states that admit paths that are $\block$-progressing, $(\emptyseq, \ddis,\den)$-violating and satisfy $P$. Since $\emptyseq$ is the empty sequence, we are ignoring the $\dpre$-prefix for now. 
We define the set  $S_G$ to be exactly those states in $M$ that admit a path $\pi$ meeting the following conditions:
\begin{itemize}
    \item $\pi$ satisfies $P$, and
    \item $\pi$ is $\block$-progressing, and
    \item $\pi$ satisfies one of the following conditions:
    \begin{itemize}
        \item $\pi$ is $\ddis$-free, or
        \item $\pi$ contains an occurrence of an action in $\den$, and the prefix of $\pi$ before the first occurrence of an action in $\den$ is $\ddis$-free.
    \end{itemize}
\end{itemize}
Our goal is then to prove that $\llbracket \mathit{invariant_G} \rrbracket_{\env} = S_G$. This takes two steps: first we prove that $S_G$ is a fixed point of the transformer characterising $\mathit{invariant_G}$, and then that it is the \emph{greatest} fixed point. 
\begin{lemma}\label{lem:inv-fix}
    $S_G$ is a fixed point of the transformer $T_G$ defined by:
    \begin{multline*}
        T_G(\mathcal{F}) = \bigcap_{a \in \nonblock}\left\{s \in \states \mid s \in \llbracket\formon{a}\rrbracket_{\env [X := \mathcal{F}]} \imps 
        s \in \llbracket \mathit{eliminate_G}(a)\rrbracket_{\update{\env}{X}{\mathcal{F}}}\right\}
    \end{multline*}
\end{lemma}

Proving that $S_G$ is a fixed point means proving $T_G(S_G) = S_G$, which we do by mutual set inclusion.  
We briefly explain how we reach the conclusion that if a state $s$ is in $T_G(S_G)$, then it must also be in $S_G$. If there are no non-blocking actions $a$ such that $s \in \llbracket \formon{a}\rrbracket$, then the empty path witnesses that $s$ is in $S_G$; for this we use the locking property as well as \autoref{prop:gen-finite}. If there is an action $a \in \nonblock$ such that $s \in \llbracket \formon{a}\rrbracket$, then we know from $s \in T_G(S_G)$ that $s \in \llbracket \mathit{eliminate}_G\rrbracket_{\env[X := S_G]}$. Then \autoref{lem:sat} yields a finite path $\pi$ that is $\ddis$-free and on which a state in which $\den$ is enabled is reached, or $a$ is eliminated and a state in $S_G$ is reached. In the former case, we can use feasibility of $P$ and \autoref{prop:gen-to-progress} to find a path that satisfies $P$ and is $\block$-progressing, and that is also $\ddis$-free until the first occurrence of an action in $\den$. Hence, we find a path that witnesses $s \in S_G$. If $a$ is eliminated instead, then since $\pi$ reaches a state that is in $S_G$, we can create a path $\pi'' = \pi \co \pi'$, where $\pi'$ is a path from the final state of $\pi$ that is $\block$-progressing, satisfies $P$, and is $\ddis$-free up until the first occurrence of an action in $\den$. Using \autoref{prop:gen-prepend}, we can then show $\pi''$ witnesses $s \in S_G$.
The other part of the proof, that an arbitrary state in $S_G$ is also in $T_G(S_G)$, works very similarly, just in the other direction. We need \autoref{prop:gen-suffix} in that part of the proof in place of \autoref{prop:gen-prepend}.

To prove that $S_G$ is actually the greatest fixed point of $T_G$, we use the following supporting lemma:
\begin{lemma}\label{lem:inner}
    For all states $s$ in a fixed point $\mathcal{F}$ of $T_G$ as defined in \autoref{lem:inv-fix}, if there is no action in $\den$ that is reachable from $s$ without doing an action in $\ddis$ and there exists at least one action $a \in \nonblock$ such that $s$ satisfies $\formon{a}$, then there exists a finite path $\pi$ from $s$ to some state $s'$ meeting all of the following conditions:
    \begin{enumerate}
        \item $s' \in \mathcal{F}$, and
        \item $\pi$ has length at least one, and
        \item $\pi$ is $\ddis$-free, and
        \item for all actions $a \in \nonblock$ such that $s$ satisfies $\formon{a}$, there is a state on $\pi$ that satisfies $\formoff{a}$ or there is a transition on $\pi$ labelled with an action in $\actelim{a}$.
    \end{enumerate}
\end{lemma}

However, for an intuitive explanation of the proof that $S_G$ is the greatest fixed point of $T_G$ we do not need this supporting lemma. So we do not go into its proof here.

\begin{lemma}\label{lem:inv-greatest}
    $S_G$ is the greatest fixed point of the transformer $T_G$ as defined in \autoref{lem:inv-fix}.
\end{lemma}

In the proof of \autoref{lem:inv-greatest}, we take an arbitrary state $s$ in an arbitrary fixed point $\mathcal{F}$ of $T_G$, and then prove that $s \in S_G$. This is done by constructing a path $\pi$ from $s$ that satisfies $P$, is $\block$-progressing, and $(\emptyseq, \ddis, \den)$-violating. 
The proof considers three cases. The first case is when $s$ is $\block$-locked. In that case, the empty path is trivially $\block$-progressing and violating, and by \autoref{prop:gen-finite} also satisfies $P$.
The second case is that it is possible to reach a state in which an action in $\den$ is enabled without taking actions in $\ddis$. If this is the case, then that path can be extended using feasibility of $P$ to create a path that is violating, satisfies $P$, and, by \autoref{prop:gen-to-progress}, is $\block$-progressing.
The most complicated case is the one in which neither of the previous two is true. 
The idea is that we construct a path from $s$ by repeatedly adding $\ddis$-free path segments to an initial path $\pi = s$, in a potentially endless construction. In every iteration, we consider whether the final state of the path constructed thus far is $\block$-locked. If so, then, similar to the first case, we have found a witness for $s \in S_G$. If not, then there is some non-blocking action $a$ that is on. And therefore, by the definition of $T_G$, there is a finite $\ddis$-free path on which $a$ is eliminated. We can disregard the possibility that we instead reach a state in which $\den$ is enabled, since we addressed that case separately. The segment we append to $\pi$ is the finite $\ddis$-free path on which $a$ is eliminated. By the persistent property, non-blocking actions for which $\formonsym$ is satisfied but are not eliminated remain ``on'', and hence can be eliminated later in $\pi$. We can therefore simply keep track of all the actions that for which $\formonsym$ is satisfied in states we encounter, and eliminate them all in turn. This produces an infinite, and therefore $\block$-progressing, path that satisfies $P$, and that is entirely $\ddis$-free. So in this case too, we construct a path that witnesses $s \in S_G$. 

Since the semantics of $\mathit{invariant_G}$ are characterised as the greatest fixed point of $T_G$, we can conclude the following from the definition of $S_G$.
\begin{corollary}\label{cor:inv}
    The set of states characterised by $\mathit{invariant_G}$ is exactly the set of states that admit $\block$-progressing, $(\emptyseq,\ddis,\den)$-violating paths that satisfy $P$.
\end{corollary}

We then prepend the $\dpre$ part of the formula.
\begin{lemma}\label{lem:violate}
    For all environments $\env$ and states $s \in \states$, it holds that $s \in \llbracket \mathit{violate_G} \rrbracket_{\env}$ if, and only if, $s$ admits a path that is $\block$-progressing, satisfies $P$ and is $(\dpre, \ddis, \den)$-violating.
\end{lemma}
This step is rather trivial, since it follows directly from \autoref{cor:inv} and the basic definition of the modal $\mu$-calculus. 

The final step of the proof is then to negate $\mathit{violate}_G$. \autoref{thm:gen} follows directly.

}
\end{document}